\documentclass[12pt]{report}
\usepackage{setspace}
\usepackage{subfigure}
\usepackage{xcolor}
\usepackage{xspace}
\usepackage{todo}
\usepackage{enumitem}
  \usepackage{forest}
  \definecolor{foldercolor}{RGB}{124,166,198}
  \definecolor{filecolor}{RGB}{78,78,78}
  \definecolor{changedfile}{RGB}{255,165,0}
\definecolor{newfile}{RGB}{34,139,34}
\usepackage{amssymb,graphicx,color}
\usepackage{amsfonts}
\usepackage{latexsym}
\usepackage{a4wide}
\usepackage{amsmath}
\usepackage{amsthm}
\usepackage{mathtools}
\usepackage{listings}
\usepackage{pgfplots}
\usepackage{geometry}
\usepackage{algorithm}
\usepackage[noend]{algpseudocode}
\makeatletter
\def\BState{\State\hskip-\ALG@thistlm}
\makeatother
\usepackage{tikz}
\usetikzlibrary{arrows.meta,positioning,decorations.markings,shapes,fit,backgrounds, decorations.pathreplacing}

\usetikzlibrary{arrows.meta,positioning,decorations.markings,shapes,fit,backgrounds}
\usetikzlibrary{decorations.pathreplacing}
\usepackage{xcolor}
\usepackage{caption}
\usepackage[
backend=biber,
style=numeric,
sorting=ynt
]{biblatex}
\usepackage{array}
\usepackage{booktabs}
\usepackage[T1]{fontenc}
\usepackage{fullpage}
\usepackage{lmodern}
\usepackage{graphicx}
\usepackage{fancyhdr}
\usepackage{multicol}
\usepackage{adjustbox}
\usepackage{float}
\usepackage{changepage}

\usepackage{multicol}

\usepackage[colorlinks=true,
            citecolor=blue,
            linkcolor=blue,
            urlcolor=blue,
            pdfstartview=FitH,
            bookmarks=true,
            bookmarksopen=true,
            bookmarksdepth=2,
            ]{hyperref}
            
\usepackage{cleveref} 
\definecolor{ForestGreen}{RGB}{34,139,34} 

\newtheorem{theorem}{Theorem}
\newtheorem{lemma}{Lemma}

\newtheorem{observation}{Observation}

\usepackage{listings,xcolor}

\definecolor{rustcomment}{RGB}{106,135,89}
\definecolor{rustkeyword}{RGB}{86,156,214}
\definecolor{ruststring}{RGB}{206,145,120}
\definecolor{rustnumber}{RGB}{181,206,168}
\definecolor{rusttype}{RGB}{78,201,176}

\lstdefinelanguage{Rust}{
    keywords={
        as, break, const, continue, crate, else, enum, extern, false, fn,
        for, if, impl, in, let, loop, match, mod, move, mut, pub, ref,
        return, self, Self, static, struct, super, trait, true, type, unsafe,
        use, where, while, async, await, dyn, abstract, become, box, do,
        final, macro, override, priv, typeof, unsized, virtual, yield
    },
    keywordstyle=\color{rustkeyword}\bfseries,
    ndkeywords={bool, u8, u16, u32, u64, u128, i8, i16, i32, i64, i128, f32, f64, char, str, String, Vec, Option, Result},
    ndkeywordstyle=\color{rusttype}\bfseries,
    sensitive=true,
    comment=[l]{//},
    commentstyle=\color{rustcomment}\itshape,
    string=[b]",
    stringstyle=\color{ruststring},
    numbers=left,
    numberstyle=\tiny\color{gray},
    numbersep=8pt,
}

\makeatletter
\renewcommand{\@makechapterhead}[1]{%
  \vspace*{20pt}
  {\parindent \z@ \raggedright \normalfont
    \ifnum \c@secnumdepth >\m@ne
        \huge\bfseries \@chapapp\space \thechapter
        \par\nobreak
        \vskip 10pt 
    \fi
    \interlinepenalty\@M
    \Huge \bfseries #1\par\nobreak
    \vskip 20pt 
  }}
\makeatother

\newcommand{\tarpon}{\textsf{Sharpedo}\xspace}

\definecolor{primaryblue}{RGB}{41, 128, 185}
\definecolor{secondaryblue}{RGB}{52, 152, 219}
\definecolor{lightgray}{RGB}{236, 240, 241}
\definecolor{darkgray}{RGB}{52, 73, 94}
\definecolor{successgreen}{RGB}{39, 174, 96}
\definecolor{warningred}{RGB}{231, 76, 60}

\tikzset{
    process/.style={
        rectangle, 
        rounded corners=2pt,
        minimum width=2.2cm, 
        minimum height=0.9cm,
        text centered,
        font=\footnotesize\sffamily,
        draw=primaryblue,
        fill=lightgray,
        line width=1pt,
        drop shadow={opacity=0.2, shadow xshift=1pt, shadow yshift=-1pt}
    },
    decision/.style={
        diamond, 
        aspect=1.8,
        minimum width=2cm, 
        minimum height=1.2cm,
        text centered,
        font=\footnotesize\sffamily,
        draw=secondaryblue,
        fill=white,
        line width=1pt,
        drop shadow={opacity=0.2, shadow xshift=1pt, shadow yshift=-1pt}
    },
    terminal/.style={
        rectangle,
        rounded corners=5pt,
        minimum width=2.2cm,
        minimum height=0.9cm,
        text centered,
        font=\footnotesize\sffamily\bfseries,
        draw=successgreen,
        fill=successgreen!20,
        line width=1.5pt,
        drop shadow={opacity=0.3, shadow xshift=1pt, shadow yshift=-1pt}
    },
    arrow/.style={
        -Stealth,
        line width=1pt,
        color=darkgray
    },
}

\usepackage{chngcntr}

\title{  	{ \includegraphics[scale=.06]{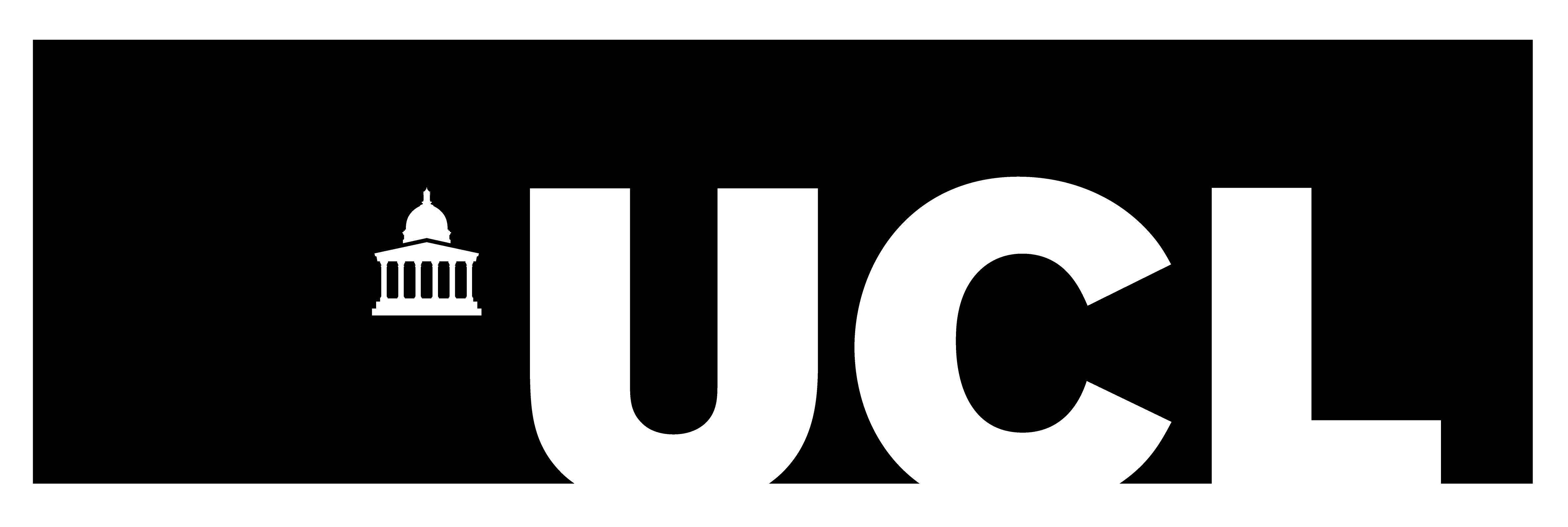}}\\
{{\Huge \tarpon: Dual-Mode Uncertified DAG-Based Consensus Protocol}}\\
		}
\date{Submission date: 15 September 2025}
\author{Zeno de Angeli\thanks{
{\bf Disclaimer:}
This report is submitted as part requirement for the Information Security MSc at UCL. It is
substantially the result of my own work except where explicitly indicated in the text. The report may be freely copied and distributed provided the source is explicitly acknowledged}
\\ \\
Information Security MSc\\ \\
Philipp Jovanovic\\Alberto Sonnino}

\begin{document}
\onehalfspacing
\maketitle
\chapter*{Acknowledgements}
\quad I would like to acknowledge both Philipp Jovanovic and Alberto Sonnino for their continued support throughout the making of this scientific report. I am particularly grateful for their help in guiding me to research such an interesting topic. Lastly, I would like to thank my peers for helping me throughout the academic year.
\begin{abstract}
    The increasing adoption of decentralized systems, in particular blockchain technology, has spurred research efforts in such distributed architecture field. Notably, consensus protocols have become a focal point of study for researches, with the goal of creating a protocol that would be both secure and performant. Consensus protocols have undergone a number of changes in recent years, going from classical linear-based consensus protocols to the most recent, highly performant, Directed Acyclic Graph based consensus protocols. In particular, the \textsf{Mysticeti}\xspace and \textsf{Mahi-Mahi}\xspace protocols both represent leading approaches to consensus protocols, leveraging a novel \textit{uncertified} Directed Acyclic Graph data structure to achieve substantial performance benefits compared to prior work\cite{MYSTICETI, MAHIMAHI}. \textsf{Mysticeti}\xspace and \textsf{Mahi-Mahi}\xspace differ in their underlying network assumptions, where the former makes assumptions about timing bounds for communication while the latter does not. \textsf{Mysticeti}\xspace trades robustness for lower levels of latency in ideal conditions, whereas \textsf{Mahi-Mahi}\xspace favours increased latency in order to ensure stronger progress guarantees in asynchronous settings. Although dual-mode protocols exist (i.e., protocols that can adapt do different network assumptions), none have been implemented to take full advantage of the novel data structure seen in the aforementioned protocols. Thus, our research project aims to trace the progression of consensus protocols, analyse how they address network assumptions constraints, and explore dual-mode protocols. Then, we present \tarpon, the first dual-mode \textit{uncertified} Directed Acyclic Graph consensus protocols which aims to achieve the best of both worlds: the high-performance of partially-synchronous protocols, while providing guarantees of progress even in asynchronous settings. Furthermore, \tarpon integrates a novel dynamic scheduling mechanism resulting in improved system efficiency and performance optimization. We conclude by providing formal proofs of \textit{safety} and \textit{liveness} for \tarpon, ensuring the correctness of the protocol.
\end{abstract}
\footnotesize
\tableofcontents
\normalsize
\setcounter{page}{1}

\chapter{Introduction}
\label{sec:introduction}
\quad The increase in popularity with blockchains and decentralised systems has lead researchers to increasingly dedicate their resources to address the technological challenges and limitations inherent in these systems, working to enhance scalability, security, and efficiency. Specifically, consensus protocols (i.e., protocols that allow a network of computers to come to an agreement on a common value) have become the centre of many research efforts. Consensus protocols, often, represent the bottleneck of blockchains, and distributed systems in general. There are two main metrics which determine the capabilities of a distributed system's latency and throughput. In the context of blockchains, we may have a user $A$ that decides to create transaction $tx_0$ to the system, which states that they are passing ownership of some value to another user $B$. The latency between submission of $tx_0$ till its finalisation (i.e., until consensus has been reached on $tx_0$), is dictated, in part, by the system's underlying consensus protocol. Moreover, the amount of transactions that a system can handle, typically measured in $tx/s$, is also bottle-capped by the consensus protocol. Underlying blockchains of famous distributed systems, such as Bitcoin and Ethereum, utilise a linear blockchain \cite{BITCOIN, ETHEREUM}, where a block holds a number of transactions. Meaning that a block is connected to the previously added block, creating a sequential list of blocks (i.e., a block chain). This structure, however, creates an inverse relationship between the two metrics discussed. For instance, decreasing the amount of transactions a block can hold decreases its latency, but, in turn, it also decreases the $tx/s$ (throughput). Furthermore, increasing participants of a consensus protocols does not improve either throughput or latency, as the blockchain remains linear, thus, only one block can be added per consensus. Recent blockchains have begun to adopt Directed Acyclic Graph (DAG) consensus protocols, which have been a breakthrough in terms of amount of transactions processed per second (throughput) and safety against faults \cite{BULLSHARK}. Instead of creating a sequence of blocks, DAG-based consensus protocols create a graph of blocks, which, after consensus, can be linearised. Throughput increases naturally, as each participant of consensus can propose their own block. However, these protocols retained a relatively high time-delay from the submission of a transaction to the finalisation of that same transaction, as each block requires a number of communication-delays between participants before it can be inserted into the DAG \cite{BULLSHARK,TUSK,HOTSTUFF}. This fact motivated the design of \textit{uncertified} DAG-based consensus protocols: protocols that achieved higher-throughput of "\textit{certified}" DAG-based consensus protocols, while reaching lower levels of latency \cite{MYSTICETI, MAHIMAHI}. These types of protocols take full advantage of the DAG structure to reach consensus, discarding the need of any additional communication between participants. Due to their novelty, \textit{uncertified} DAG-based consensus protocols present numerous research opportunities. Specifically, since consensus protocols are not only classified by their underlying data structure, but also by the different assumptions they make about their working environment, there is currently a lack of \textit{uncertified} DAG-based consensus protocols that take advantage of more flexible assumptions.

\section{Research Objectives}
\quad The design of consensus protocols requires assumptions to be made on the type of environment they should be made to operate. Beyond a threat model (i.e., what a malicious actor can do to disrupt the system), consensus protocols must also determine which network condition they can work in. These assumptions often force protocols to make sacrifices either in performance (such as latency) to provide more robust security, or in security to increase performance. However, there exist protocols that aim to combine network assumptions to achieve the best of both worlds, called \textbf{dual-mode consensus protocols}. Given the rapid evolution of this field, the \textit{uncertified} data structure employed in current state-of-the-art protocols that achieve optimal performance today lacks dual-mode equivalents. Therefore, the objective of this research is to build a consensus protocol that can adapt to different network settings, by combining the two state-of-the-art \textit{uncertified} DAG-based consensus protocols, \textsf{Mysticeti}\xspace and \textsf{Mahi-Mahi}\xspace, to create the first dual-mode \textit{uncertified} DAG-based consensus protocol \cite{MYSTICETI, MAHIMAHI}.
\newpage
\section{Contribution}
\quad This research project aims to fill the gap in current literature by providing the following contributions:
\begin{itemize}
    \item We present \tarpon, the first dual-mode \textit{uncertified} DAG-based consensus protocol, providing the best of both worlds: the low latency of partially-synchronous consensus protocols, combined with the added resilience intrinsic to asynchronous consensus protocols.
    \item We incorporate a dynamically changing scheduling system for \tarpon's dual-mode component, which will improve latency in periods of synchrony and asynchrony, accordingly. 
    \item We provide detailed algorithms and implementations, alongside extensive testing proving the functionality of \tarpon.
    \item We provide formal proofs that demonstrate the \textit{safety} and \textit{liveness} under byzantine assumptions for \tarpon and our dynamic dual-mode schedule. 
\end{itemize}
\section{General Overview}
\begin{itemize}
    \item Chapter~\ref{chap:LitReview}: We present an overview of consensus protocols, how they handle different restrictions, and how newer types of consensus protocols present clear opportunities for advancement.
    \item Chapter~\ref{chap:dual_mode_protocol}: This chapter introduces \tarpon, the solution to our problem statement.
    \item Chapter~\ref{chap:implementation}: This chapter further extends \tarpon, giving an overview of the practical implementations, alongside results obtained from the simulations. 
    \item Chapter~\ref{chap:Proofs}: We finish our presentation of \tarpon by providing formal proofs for the protocol, ensuring its correctness. 
    \item Chapter~\ref{chap:conclusion}: Lastly, we provide a conclusion to our project, alongside how the project can be advanced further.
\end{itemize}
\chapter{Literature Review}
\label{chap:LitReview}
\section{Consensus Protocols}
\quad Distributed systems are a collection of computers, commonly mentioned as nodes, who communicate with each other over a network, acting as one system to achieve a goal. The \textit{consensus} problem is defined as the requirement for all nodes in a system to agree on a value \cite{CFT,BYZAGREEMENT}. A consensus protocol aims to solve the \textit{consensus} problem which is required to achieve State Machine Replication (SMR) \cite{SMR}. This means that a machine that begins in the same state as another machine, given the same set of inputs in the same order, will give the same output. To achieve SMR, all the computers in the system must agree (i.e. reach consensus) on what inputs and in what order they should be processed. Consensus protocols are the core which resides in distributed networks allowing for SMR to happen \cite{SMR1}. They represent the protocols that are utilised to make sure that all participants of a network are able to achieve consensus on decisions, thus enabling the progress of the system. Consensus protocols aim to solve the fundamental issue of agreement in distributed systems. To be more specific, consensus protocols must satisfy the following properties:
\begin{itemize}
    \item \textit{Safety}: Safety means that "bad things will never happen". It is important to note that safety can \textbf{never} be violated, or else the entire system may enter an undesirable state (e.g., two different computers in a distributed system could have a different view of the system).
    \item \textit{Liveness}: Liveness means that "good things will, eventually, happen". Similarly to safety, an important factor of liveness is that progress will \textbf{eventually} be achieved. It is possible that during asynchronous periods, messages may take longer to reach other computers part of the distributed system, but, eventually (i.e., when the messages reach the other computers), progress will be achieved.
\end{itemize}
\subsection{Problem Specification} 
\label{sec:foundations}
\quad Reliable broadcast is a distributed systems primitive that guarantees the following: once an honest process broadcasts a message, all other honest processes will eventually receive the message \cite{ReliableBroadcast}. However, the reliable broadcast abstraction is not enough for distributed systems, as, given by their communicative nature, the ordering of messages is just as important. Atomic broadcast, is an abstraction that extends reliable broadcast by also guaranteeing messages will be received by other participants in the same order \cite{AB}. Atomic broadcast is equivalent to the \textit{consensus} problem, as participants of a system that solve atomic broadcast must agree on the order of messages, thus, \textit{consensus} can be reduced to atomic broadcast \cite{REDUCEDTO}. If a protocol solves atomic broadcast, they solve consensus, thus, sequentially, SMR can be achieved \cite{SMR}. Assuming that each process has the capability to broadcast messages to all other processes and to deliver messages that have been broadcast by any process in the system, a protocol that implements atomic broadcast guarantees the following properties \cite{NEWAB}:
\begin{itemize}
    \item \textit{Integrity}: For any message $m$, every process will, at most, deliver $m$ once, and only if $m$ has been broadcast by a process.
    \item \textit{Validity}: If process $p_i$ broadcasts a message $m$, then every process will eventually deliver $m$.
    \item \textit{Agreement}: If a process $p   _i$ delivers $m$, then all other processes eventually deliver $m$. 
    \item \textit{Total Order:} If a process $p_i$ delivers a message $m$ before another message $m'$, then no other process delivers $m'$ before $m$. 
\end{itemize}
\quad Thus, the common aim of consensus protocols is to solve atomic broadcast \cite{TUSK,BULLSHARK,MYSTICETI,MAHIMAHI}. This means consensus protocols must achieve the aforementioned properties to solve atomic broadcast. Together, when a protocol satisfies atomic broadcast, it also achieves the \textit{safety} and \textit{liveness} properties. While atomic broadcast is easy to achieve in ideal settings, the challenge lies in solving the problem under various network and threat models.
\section{Network Models}
\label{sec:network_models}
\quad Given the nature of distributed systems, communication uncertainty is modelled after the amount of power an adversary may have on delaying messages \cite{SYNCHRONY}. Therefore, consensus protocols are built with three network models in mind.
\paragraph{Synchronous Model} The synchronous model makes the assumption that there exists an upper bound, usually defined as $\Delta$, on the delivery of messages. The bound $\Delta$ can be set to any value, but once chosen, the system must adhere to this assumption throughout its operation \cite{SYNCHRONY}. There are advantages and disadvantages to such approach. If the synchronous assumptions hold, then they become easier to reason with. However, synchronous assumptions require precise timing assumptions because if $\Delta$ is set too high, then timeouts can potentially slow down the system, and if $\Delta$ is set too low, then, depending on network conditions, it will be hard to achieve progress. Lastly, there is also the possibility of a scenario where a validator sends a message to two distinct validators, and one may receive the message within $\Delta$ and another may receive the message after $\Delta$. This can break the \textit{safety} aspect of the model. 
\paragraph{Asynchronous Model}  Compared to the synchronous model, the asynchronous network model does not make \textbf{any} assumptions on the timing of messages. Instead, we assume that, eventually, messages sent out by honest participants are guaranteed to be received by other honest participants. A clear advantage of such approach is that we completely discard the need of calculating an upper bound $\Delta$ assumption. Furthermore, in periods of asynchrony, the system is guaranteed to make progress. However, as defined by the impossibility result of Fischer, Lynch, and Paterson (FLP) \cite{FLP}, it is impossible for a deterministic consensus protocol to reach consensus in the presence of even one faulty validator. If a node in a system is faulty, or is Byzantine, and it stops communicating with the remainder of the participants, then we are at risk of breaking either \textit{safety} or \textit{liveness} properties. Due to the asynchronous assumptions, we cannot utilise timeouts to "skip/ignore" malicious/faulty nodes. Imagine a scenario where a node is experiencing high levels of asynchrony, and it takes a long time for them to receive and send a message. We have made the assumption that the message sent, if the node is honest, will eventually be received, thus, skipping such node would break such assumptions. So, in an asynchronous consensus protocol, if a node is non-responsive, skipping the node could risk violating the \textit{safety} properties of the protocol, but if we wait for a response, we risk never making progress again, breaking the \textit{liveness} properties. 
\\
\null \quad Therefore, all protocols built upon asynchronous assumptions utilise some form of randomness in their consensus processes, as it has been proven that randomness can circumvent the results reached by FLP \cite{FLP,FLP1,FLP2}. This results from the fact that asynchronous consensus protocols, rather than guaranteeing termination, achieve termination with probability of 1. A popular method utilised in asynchronous consensus protocols requires the selection of leaders to be randomised (further defined in Section~\ref{sec:prop_slot_leader}). Each round of consensus selects a leader randomly from all participants, and there is a $p > 0$ probability of electing an honest leader. In such protocols, the choice of leader is often done retroactively: once enough participants propose a message at the end of the round (usually $2f+1$), a random leader is elected amongst all proposers. If the selected leader's proposal satisfies the protocol conditions, consensus is achieved. However, if the selected leader's proposal is invalid, incomplete, or if the leader fails to propose entirely, the protocol moves to the next round without progress. After enough rounds, the probability of making progress approaches $1$, ensuring \textit{liveness}.
\paragraph{Partially-Synchronous Model} Similarly to the synchronous model, partially-synchronous models make the assumptions that messages will be delivered within time bound $\Delta$, however, they differ in that they require certain assumptions to be satisfied. This assumption is called the \textit{Global Stabilisation Time} ($GST$), and it is a theoretical concept that represents a point in time where network conditions are ideal. Intuitively, this means that \textbf{after $GST$} has been reached (i.e., after the network conditions have stabilised), then messages will be received within the time bound $\Delta$. Similarly to synchronous network assumptions, these assumptions are easier to reason with, and allow for the utilisation of timeouts (often described as $GST + \Delta$) to guarantee \textit{liveness} to the system. The work by Dwork, Lynch and Stockmeyer (DLS) has been seminal for the creation of these "in-between" assumptions \cite{SYNCHRONY}. As proven by DLS, while \textit{liveness} is guaranteed only after $GST$ has been reached, \textit{safety} is always guaranteed in both periods of synchrony and asynchrony. 

\quad Consensus protocols are built with these assumptions in mind, thus each protocol will specify which assumption the protocol is built to endure. However, there also exist dual-mode protocols, which, as the name suggest, are protocols that can run under multiple network assumptions. These protocols, usually, utilise some mechanism to determine the condition of the network to dynamically switch mode \cite{DITTO, BULLSHARK}. Under ideal conditions, the protocol may run under synchronous assumptions, allowing for faster consensus (thus lower latency). However, if network conditions deteriorate, the protocol may decide to change to an asynchronous consensus protocol capable of making progress, albeit slower, under asynchronous conditions. As stated by many researches, these protocols, although providing the best of both worlds; are usually more complex, as they need to handle multiple distinct consensus protocols \cite{MAHIMAHI, MYSTICETI}. However, as we will see in Section~\ref{chap:dual_mode_protocol}, we address these complexities by providing a simplified approach to dual-mode mechanisms.
\section{Threat Models}
\quad Threat models represent a structured way to model possible threats that may target some systems. Given the unpracticality of protecting a system against all threats that one may think of, models have been created to simplify the process of protection against malicious entities.
\paragraph{Crash Fault Tolerance} A system deemed to be crash fault tolerant (CFT) is a system that is capable of handling a fixed amount of crash/faulty nodes in its system, where a faulty node stops communicating with the system \cite{CFT}. CFT protocols have a bound of $n \geq 2f+1$ nodes in the system, where $f$ represents the number of faulty nodes in the system. If all $f$ nodes are faulty, then there still exists $f+1$ (a majority of nodes, $> 50\%$, in the protocol that are able to reach consensus). Centralised protocols, such as Paxos \cite{PAXOS} and Raft \cite{RAFT}, run under CFT assumptions. Since these protocols are centralised, the main cause of concern is of nodes to crash during consensus, thus, CFT assumptions are enough to secure the system. However, decentralised systems introduce the risk of possible malicious participants during consensus, thus, such protocols must operate under stricter assumptions.
\paragraph{Byzantine Fault Tolerance} Similarly to CFT assumptions, Byzantine fault tolerant (BFT) protocols are capable of running under the presence of both faulty \textbf{and} malicious (also known as Byzantine) participants. More specifically, BFT protocols aim to solve the Byzantine agreement problem, a problem closely related to the consensus problem; however, it differs by the fact that processes must agree on a value in the presence of Byzantine processes \cite{BYZAGREEMENT}. Compared to CFT assumptions, BFT systems have a bound of $n\geq 3f+1$ total participants, where $f$ represents the number of faulty/Byzantine nodes \cite{BFT}. The bound of required participants increases due to the fact that Byzantine nodes, besides not communicating with the rest of the system, can also lie (sending false messages to participants) and equivocate (sending conflicting messages to different participants). Thus, a BFT system with $3f+1$ nodes ensures that even if $f$ participants are malicious, the remaining $2f+1$ honest participants can still outvote them. Furthermore, quorums are required to be of size $2f+1$, ensuring that any two different quorums will overlap with at least $f+1$ participants, and, since at most $f$ participants are faulty, both quorums will reach the same conclusion given the one honest participant. This is also known as \textbf{quorum intersection}, and is going to be relevant in the protocols that we will be analysing throughout this research paper. When dealing with Byzantine threats, we also have to strengthen our definition of atomic broadcast. As such, all BFT consensus protocols try to solve Byzantine Atomic Broadcast (BAB), which, compared to atomic broadcast, makes stronger assumptions that require all participants to be honest \cite{AB}. Both network models and threat models are assumptions critical to the development of consensus protocols. However, beyond these assumptions, the construction of the underlying blockchain data-structure itself is equally important. 
\section{Linear-chain Consensus Protocols} 
\quad In this section, we will show a brief overview of traditional linear-chain consensus protocols. The linear-chain aspect of consensus protocols refers to the underlying data structure that a distributed ledger can be based on. Each chain begins with a \textit{genesis} block, and future blocks get added onto each other (each block referencing a previous block), thus creating a "chain" of blocks. But how do blocks get proposed and added to the chain? We have popular chains, such as \textsf{Bitcoin}\xspace's, where blocks can be added by any individuals who are part of the network, and via the longest chain rule, the chain is created \cite{BITCOIN}. The longest chain rule dictates that honest nodes will always work on the longest chain they have available. However, such protocols only offer probabilistic finality, since a new longest chain could potentially exclude transactions that were accepted in a previous chain. Therefore, many other blockchains decided to utilise a form of blockchain that instead utilizes a committee to decide on the next block, ensuring the existence of a singular chain. Leader-based consensus protocols are consensus protocols that leverage the idea of a \textbf{quorum}. In such protocols, in every view (i.e., a point in time where a specific process is responsible for progressing the chain) a leader is elected (e.g., for example, via the amount of stakes that the participants have staked), and said leader proposes a block. Then, the remainder of the participants will vote on said blocks to determine whether it will be used to extend the block chain, as exemplified in Figure~\ref{fig:linear-chain}.
\vspace{2em}
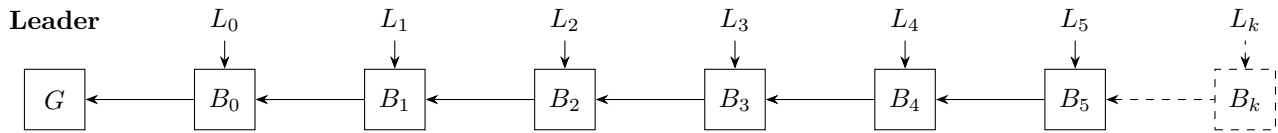
\begin{figure}[H]
    \centering
    \makebox[\textwidth][c]{
\begin{tikzpicture}[
    scale=1.5,
    node distance=1.8cm,
    validator/.style={draw, minimum size=0.8cm},
    >=Stealth
]
\footnotesize
\node at (-5.5, 0.7) {\textbf{Leader}};
\node (l0) at (-4, 0.7) {$L_0$};
\node (l1) at (-2.5, 0.7) {$L_1$};
\node (l2) at (-1, 0.7) {$L_2$};
\node (l3) at (0.5, 0.7) {$L_3$};
\node (l4) at (2, 0.7) {$L_4$};
\node (l5) at (3.5, 0.7) {$L_5$};
\node[dashed] (l_cont) at (5, 0.7) {$L_k$};
\node[validator] (g) at (-5.5, 0) {$G$};
\node[validator] (b0) at (-4, 0) {$B_0$};
\node[validator] (b1) at (-2.5, 0) {$B_1$};
\node[validator] (b2) at (-1, 0) {$B_2$};
\node[validator] (b3) at (0.5, 0) {$B_3$};
\node[validator] (b4) at (2, 0) {$B_4$};
\node[validator] (b5) at (3.5, 0) {$B_5$};
\node[validator, dashed] (b_cont) at (5, 0) {$B_k$};
\draw[<-] (g.east) -- (b0.west);
\draw[<-] (b0.north) -- (l0.south);
\draw[<-] (b0.east) -- (b1.west);
\draw[<-] (b1.east) -- (b2.west);
\draw[<-] (b2.east) -- (b3.west);
\draw[<-] (b3.east) -- (b4.west);
\draw[<-] (b4.east) -- (b5.west);

\draw[<-] (b1.north) -- (l1.south);
\draw[<-] (b2.north) -- (l2.south);
\draw[<-] (b3.north) -- (l3.south);
\draw[<-] (b4.north) -- (l4.south);
\draw[<-] (b5.north) -- (l5.south);
\draw[<-, dashed] (b_cont.north) -- (l_cont.south);
\draw[<-, dashed] (b5.east) -- (b_cont.west);
\end{tikzpicture}}
    \caption{Illustration of a linear-chain blockchain, where each block $B$ is proposed by a leader $L$. A blockchain always begins with a hardcoded genesis block, denoted as $G$.}
    \label{fig:linear-chain}
\end{figure}
\quad These protocols require what is known as a \textbf{quorum}, which represents the minimum amount (i.e., $2f+1$) of participants that "accept" a block to add it \textit{safely} to the block chain. \textsf{Practical Byzantine Fault Tolerant}\xspace (\textsf{PBFT}\xspace) is a seminal protocol which demonstrates that it is possible to achieve consensus in a decentralised system \cite{PBFT}. \textsf{PBFT}\xspace works as follows: once process $p$ receives message $m$ from a client, the protocol begins with a request to other processes (\textbf{pre-prepare}). Then, the processes verify the request content and form a \textbf{quorum certificate} (\textbf{prepare}). Finally, they \textbf{commit} to executing the operation, which they will then reply to the client \cite{PBFT}. \textsf{PBFT}\xspace represents the first instance of a practical implementation of a BFT consensus protocol, however, its implementation requires a message complexity of $O(n^2)$, as all processes must communicate with one another to achieve consensus. Figure~\ref{fig:PBFT}, provided in Appendix~\ref{Appendix:ConsProtStruct}, presents a visualization of the \textsf{PBFT}\xspace protocol.
\subsection{Hotstuff} 
\quad \textsf{Hotstuff}\xspace, a linear-chain consensus protocol designed by Yin et al. \cite{HOTSTUFF}, is a protocol that addresses the shortcomings of \textsf{PBFT}\xspace \cite{PBFT}. \textsf{Hotstuff}\xspace improves on the original \textsf{PBFT}\xspace protocol by utilising a leader-centric communication pattern, achieving a linear message complexity $O(n)$ (as shown by Figure~\ref{fig:HotStuff}, presented in Appendix~\ref{Appendix:ConsProtStruct}). Furthermore, the protocol itself runs under partially-synchronous assumptions, thus it can leverage timeouts to ensure \textit{liveness}. If a leader fails to create a \textbf{quorum certificate} before time $\Delta$ is reached, a \textbf{view-change} process is called. All participants of the system communicate with one another to move onto the next view, and change the elected leader. This ensures that a malicious leader cannot stall the system forever (which would cause the protocol to lose \textit{liveness}). However, \textit{liveness} is lost during periods of asynchrony as messages can be consistently delayed for longer than $\Delta$. Therefore, the system will be unable to reach consensus, and the \textit{view-change} process will be called continuously.  
\subsection{Validated Asynchronous Byzantine Agreement (VABA)} 
\quad As we have mentioned, asynchronous protocols cannot utilise timeouts due to the different assumptions made, thus, to circumvent the FLP impossibility result, such protocols must rely on a form of randomness to achieve \textit{liveness}. Abraham et al. \cite{VABA} propose a solution to \textsf{VABA}\xspace by creating a Multi-world \textsf{VABA}\xspace approach \cite{MOREVABA}. This Multi-world \textsf{VABA}\xspace approach combines two primitives to circumvent the FLP impossibility result \cite{FLP}:
\begin{itemize}
    \item A \textbf{randomness beacon}, often called a common coin or a global perfect coin \cite{MAHIMAHI}. Each process holds a beacon share that, when combined with enough other beacon shares (usually $2f+1$), reveals the beacon value.
    \item \textbf{Multi-World Instances}: Running $n$ instances of a view-based validated Byzantine consensus protocol (for instance, both mentioned \textsf{Hotstuff}\xspace and \textsf{PBFT}\xspace would qualify as such) \cite{HOTSTUFF, PBFT}.
\end{itemize}
\quad Multi-world \textsf{VABA}\xspace works as follows: $n$ instances of the view-based validated Byzantine consensus protocols are run, where for each instance $i$, the proposer is party $i$. Each instance is run until proposer $i$ acquires a commit certificate, proof that the proposer has reached consensus for the current view. Proposer $i$ will broadcast the certificate, and once $n-f$ proposers (i.e., $3f+1-f=2f+1$) have also broadcast their own certificate, all participants of the consensus protocol will stop. The randomness beacon is then utilised to randomly select one instance, which will also act as a view-change trigger. Then, the protocol will act like \textsf{Hotstuff}\xspace, or any other view-based validated Byzantine consensus protocol, where the selected process $i$ will be used as the main proposer for view $v$. Thus, multi-world \textsf{VABA}\xspace circumvents the FLP impossibility result by revealing only after the fact whose instance will be selected. A malicious participant, thanks to the properties of the randomness beacon, cannot know in advance who the proposer of view $v$ is going to be, thus, they cannot slow down the system. Intuitively, all $n-1$ instances act as decoys for the one instance that will be chosen after the $n-f$ finished processes; however, this also means that the $n-1$ instances are discarded after a view-change. Hence, Multi-world \textsf{VABA}\xspace achieves liveness at the cost of requiring a message-complexity of $O(n^2)$ as all parties will run their own version of \textsf{Hotstuff}\xspace, subsequentially increasing latency \cite{MOREVABA,HOTSTUFF}. This naturally leads to the question of how can we create a protocol that can take advantage of the low-latency given by partially-synchronous protocols, while maintaining the resilience of asynchronous protocols. 
\subsection{Ditto} 
\quad \textsf{Ditto}\xspace is a dual-mode consensus protocol, providing liveness in both synchronous and asynchronous network environments \cite{DITTO}. Unlike \textsf{Hotstuff}\xspace, \textsf{Ditto}\xspace replaces the view-synchronization with an asynchronous fallback \cite{HOTSTUFF}. While on its happy-path, Ditto runs an improved 2-chain version of \textsf{Hotstuff}\xspace, named \textsf{Jolteon}\xspace, thus running under the ideal linear communication complexity ($O(n)$). However, contrary to how \textsf{Hotstuff}\xspace works, when the network deteriorates, the \textsf{Multi-Valued Validated Byzantine}\xspace agreement (\textsf{MVBA}\xspace) protocol is run \cite{DITTO, MVBA, HOTSTUFF}. Furthermore, \textsf{Ditto}\xspace introduces the problem of fragmentation, where some processes might switch to the asynchronous fallback while others remain in the happy-path. To circumvent this, a process in \textsf{Ditto}\xspace requires to have at least $2f+1$ timeout messages before switching to the asynchronous fallback, acting as a synchronisation phase.
\\
\null \quad There are clear benefits by utilising partially-synchronous protocols such as \textsf{Hotstuff}\xspace \cite{HOTSTUFF}. First and foremost, these protocols run under linear communication complexity ($O(n)$), decreasing latency and increasing throughput. Asynchronous protocols, such as the Multi-World \textsf{VABA}\xspace we have seen, instead, require a quadratic communication cost $O(n^2)$, since each instance $n$ runs an instance of a linear communication complexity ($O(n)$) protocol \cite{VABA, MOREVABA}. \textsf{Ditto}\xspace, instead, runs with a linear communication cost $O(n)$ during its happy-path, and, during periods of asynchrony, resorts to an asynchronous protocol, running under quadratic communication cost $O(n^2)$ \cite{DITTO}. Clearly, \textsf{Ditto}\xspace provides the best of both worlds, however, we have also seen that \textsf{Ditto}\xspace introduces a number of complexities that require a careful planning to ensure that \textit{safety} is preserved. We shall now examine the next evolutionary step in consensus protocols, which, beyond achieving reduced latency and increased throughput, provides simplified requirements for dual-mode protocols.
\section{DAG-Based Consensus Protocols}
\quad Till now, we have seen consensus protocols that work by ordering blocks proposed by participants of the system, creating a linear blockchain. As stated, this approach can lead to severe limits on throughput and latency as there can only be one block accepted per view. Instead, newly researched blockchain protocols allow each participant to propose blocks that reference multiple blocks from a previous round \cite{DAGRIDER,TUSK,BULLSHARK}. When connected, these blocks form what is known as a Directed Acyclic Graph (DAG). A DAG, as the name suggests, is a graph that flows only in one direction, and thus, does not form cycles. This means that, when traversing a DAG, you can never end up where you started \cite{FORMALDAG,FORMALDAG2}. DAG-based consensus protocols utilise this property to organize DAGs into rounds and waves. Rounds (denoted as $r$) represent logical sequences, while waves (denoted as $w$) represent logical groupings of consecutive rounds (i.e., a wave is analogous to a view). 

\quad In the context of blockchains, each vertex in a DAG represents a block in a blockchain, and, as we will see, each block references a number of blocks of a previous round (i.e., a reference is an edge to a previous vertex). Creating a blockchain that forms a DAG can greatly increase throughput and reduce latency compared to linear blockchains. Looking at Figure~\ref{illustration_of_DAG}, we can see that during each round a validator can propose a block. When compared to linear blockchains, as seen in Figure~\ref{fig:linear-chain}, in a DAG-based blockchain each validator proposes blocks concurrently, creating interconnected parallel chains. Furthermore, when a leader proposal is committed in DAG-based consensus protocols, not only is the related block committed, but also its entire causal history. This means that all blocks that the leader proposal reference directly and indirectly are also committed, thus increasing throughput dramatically, since each round can include as many as $n$ blocks, where $n$ is the number of participants in the consensus protocol. This removes the linear bottleneck of linear-chain consensus protocols. Intuitively, while linear-chain consensus protocols allow only one leader per view to propose a block, in DAG-based consensus protocols each validator proposes a block; however, not every validator block is considered a leader block. We will see how \textit{safety} and \textit{liveness} are preserved in DAG-based consensus protocols. 
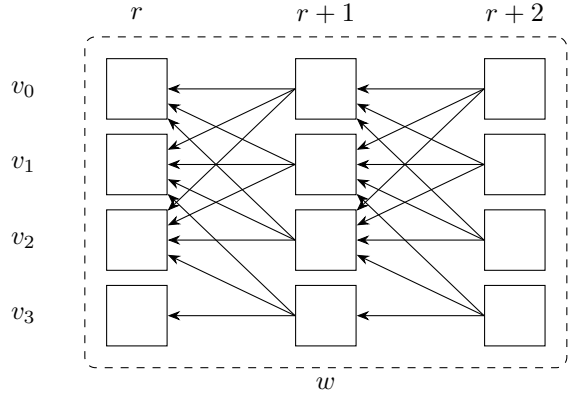
\begin{figure} [H]
    \centering
    \begin{tikzpicture}[scale=1,
    node distance=1.5cm and 2.5cm,
    validator/.style={draw, minimum size=0.8cm},
    propose/.style={validator},
    boost/.style={validator, draw=black},
    vote/.style={validator, draw=orange, thick},
    certify/.style={validator, draw=green!60!black, thick},
    supporting_validator/.style={validator, fill=green!20},
    support/.style={->, green!60!black, thick},
    >=Stealth,
    leader/.style={validator, fill=green!20},
    vote_support/.style={->, orange, thick},
    cert_support/.style={->, green!60!black, thick},
    >=Stealth
]
\footnotesize
\node at (0, 5) {$r$};
\node at (2.5, 5) {$r+1$};
\node at (5, 5) {$r+2$};
\node at (-1.5, 4) {$v_0$};
\node at (-1.5, 3) {$v_1$};
\node at (-1.5, 2) {$v_2$};
\node at (-1.5, 1) {$v_3$};
\node[propose] (v0r) at (0, 4) {};
\node[propose] (v1r) at (0, 3) {};
\node[propose] (v2r) at (0, 2) {};
\node[propose] (v3r) at (0, 1) {};
\node[propose] (v0r1) at (2.5, 4) {};
\node[propose] (v1r1) at (2.5, 3) {};
\node[propose] (v2r1) at (2.5, 2) {};
\node[propose] (v3r1) at (2.5, 1) {};
\node[propose] (v0r2) at (5, 4) {};
\node[propose] (v1r2) at (5, 3) {};
\node[propose] (v2r2) at (5, 2) {};
\node[propose] (v3r2) at (5, 1) {};
\draw[->] (v0r1.west) -- (v0r);
\draw[->] (v0r1.west) -- (v1r);
\draw[->] (v0r1.west) -- (v2r);
\draw[->] (v1r1.west) -- (v0r);
\draw[->] (v1r1.west) -- (v1r);
\draw[->] (v1r1.west) -- (v2r);
\draw[->] (v2r1.west) -- (v0r);
\draw[->] (v2r1.west) -- (v1r);
\draw[->] (v2r1.west) -- (v2r);
\draw[->] (v3r1.west) -- (v1r);
\draw[->] (v3r1.west) -- (v2r);
\draw[->] (v3r1.west) -- (v3r);
\draw[->] (v0r2.west) -- (v0r1);
\draw[->] (v0r2.west) -- (v1r1);
\draw[->] (v0r2.west) -- (v2r1);
\draw[->] (v1r2.west) -- (v0r1);
\draw[->] (v1r2.west) -- (v1r1);
\draw[->] (v1r2.west) -- (v2r1);
\draw[->] (v2r2.west) -- (v0r1);
\draw[->] (v2r2.west) -- (v1r1);
\draw[->] (v2r2.west) -- (v2r1);
\draw[->] (v3r2.west) -- (v3r1);
\draw[->] (v3r2.west) -- (v1r1);
\draw[->] (v3r2.west) -- (v2r1);
\node[draw, dashed, rounded corners, fit=(v0r) (v3r) (v0r2) (v3r2), inner sep=8pt, label=below:$w$] {};
\end{tikzpicture}
    \caption{Illustration of a DAG structure, where $w$ represents a wave, and $r$ is the round.}
    \label{illustration_of_DAG}
\end{figure}
\subsection{DAG-Rider \& Tusk}
\quad Keidar et al. \cite{DAGRIDER} propose \textsf{DAG-Rider}\xspace, the first practical DAG-based asynchronous consensus protocol. \textsf{DAG-Rider}\xspace's objective is to commit one proposed block per wave. Once a block is committed, itself and all of its causal history (until the last previously committed block) are considered to be committed. \textsf{DAG-Rider}\xspace is structured to be an asynchronous consensus protocol, thus, just like we have seen in \textsf{VABA}\xspace and \textsf{Ditto}\xspace, this protocol introduces some elements of randomness to circumvent the FLP impossibility result \cite{VABA,DITTO,FLP}. Each block in \textsf{DAG-Rider}\xspace includes a share of the common coin, and, once we reach the last round of a wave, $2f+1$ shares are combined to reveal the leader of the first round of the wave, retroactively. Due to how the protocol is constructed, compared to the multi-world \textsf{VABA}\xspace approach, which uses $n-1$ processes as decoys, blocks proposed by non-leader validators will still be included in the final sequence of blocks. This is due to the fact that the causal history of a committed block is also committed, reducing computational waste while increasing throughput. 
\\
\null\quad \textsf{Tusk}\xspace, a protocol created by Danezis et al. \cite{TUSK}, represents the first practical implementation of a DAG-based consensus protocol, and improves upon the design of \textsf{DAG-Rider}\xspace (see Figure~\ref{fig:TUSK}). Furthermore, Danezis et al. \cite{TUSK} propose \textsf{Narwhal}\xspace, a mempool layer protocol that separates the concerns of ordering data (consensus layer) from the dissemination and gathering of transactions (mempool layer). Thus, \textsf{Narwhal}\xspace gathers and constructs certificates of availability (batches of transactions), and  \textsf{Tusk}\xspace constructs the DAG out of the metadata obtained from \textsf{Narwhal}\xspace. Most importantly, this separation of concerns allows \textsf{Tusk}\xspace to reach consensus with a zero-message overhead beyond constructing the DAG itself. However, it is important to note that for a certificate of availability to be created, Narwhal requires a block to be signed by a quorum ($2f+1$ signatures), which results in a 3-message delay between certifications and addition to the DAG. 
\subsection{BullShark}
\quad Spiegelman et al. \cite{PARTIALLYSYNCHRONOUSBULL} present \textsf{Bullshark}\xspace, a consensus protocol that presents the first dual-mode solution for DAG-based consensus protocols. Given its dual-mode capabilities, \textsf{Bullshark}\xspace has been presented in two versions: a partially-synchronous version and a dual-mode version.

\quad Spiegelman et al. \cite{PARTIALLYSYNCHRONOUSBULL} present the first version of a DAG-based consensus protocol that runs under partially-synchronous assumptions (as defined in Section~\ref{sec:network_models}). The protocol works as follows: \textit{odd rounds} represent proposal rounds, while \textit{even rounds} represent voting round. Thus, a wave is composed of simply two rounds, and, since there is no requirement for any sort of randomness beacon, a block can be committed within a 6-message delay. Just like \textsf{Tusk}\xspace, \textsf{Bullshark}\xspace runs with \textsf{Narwhal}\xspace in the background; thus, each proposed block has to be certified before being broadcast (propose-certify-announce), creating the 3-message delay per round we see. Figure~\ref{fig:bullshark} visualizes this version of \textsf{Bullshark}\xspace, found in Appendix~\ref{Appendix:ConsProtStruct}. The complete version of \textsf{Bullshark}\xspace does not have the problems associated with partially-synchronous models, as it constructs the DAG to both enjoy fast commitments and maintain liveness in periods of asynchrony. A wave in \textsf{Bullshark}\xspace is composed of four rounds, where every odd round a \texttt{steady-state} leader proposes a block, which can be committed just like its partially-synchronous counterpart. However, if the system enters periods of asynchrony, \textsf{Bullshark}\xspace resorts to a \texttt{fallback} leader that is elected retro-actively in the last round of a wave, which is why \textsf{Bullshark}\xspace waves are four rounds long. In short, the protocol will utilise the fallback leader if the second \texttt{steady-state} leader of the wave failed to be committed, or if the \texttt{fallback} leader also failed to be committed. A \texttt{fallback} leader, elected via a common coin released in the last round of a wave, must be referenced $2f+1$ times in the round right after to be committed. When a \texttt{fallback} leader is committed, the next wave will go back to utilise the \texttt{steady-state} leaders. \textsf{Bullshark}\xspace achieves this dual-mode functionality by introducing voting types to the DAG edges. A validator can only vote for one type of leader per wave, and the voting type of a validator for a wave is dependent on what has happened the previous wave. Figure~\ref{fig:bullshark_both_ways}, found in Appendix~\ref{Appendix:ConsProtStruct}, 
visualizes the dual-mode component of \textsf{Bullshark}\xspace. This approach, however, increases complexity, as each validator must also keep track of what mode must be used for each wave.
\\
\null \quad \textsf{Bullshark}\xspace is, at the time of writing, the only dual-mode DAG-based consensus protocol, thus also the protocol from such class to be evaluated using a geo-distributed cloud infrastructure. Of importance, \textsf{Bullshark}\xspace achieves $\approx 100,000 - 130,000$ transactions per second independently on the size of the committee. Albeit less than \textsf{Tusk}\xspace, \textsf{Bullshark}\xspace achieves a stable latency of $2s$ compared to the $3s$ from \textsf{Tusk}\xspace \cite{TUSK, BULLSHARK}. This result is very important, as it shows that dual-mode protocols can achieve lower latency compared to fully asynchronous protocols, while guaranteeing eventual liveness in asynchronous settings. Furthermore, compared to \textsf{Ditto}\xspace, \textsf{Bullshark}\xspace maintains an optimal amortized communication complexity ($O(n)$) through either \texttt{steady-state} waves and \texttt{fallback} waves. 
\section{Uncertified DAG-Based Consensus Protocol}
\label{sec:uncertified_DAG}
\quad Research in DAG-based consensus protocols has led to the utilisation of a new data structure, \textit{uncertified} DAGs. Structure wise, \textit{uncertified} DAGs are similar to certified DAGs, as seen in Figure~\ref{illustration_of_DAG}; however, they differ in how a block is broadcast. In certified DAGs, before a block is added to the DAG, it has to be broadcast to each validator to collect enough signatures to certify it, and then, once certified, the block is broadcast again to all validators. This certification process increases latency, as broadcasting a certified block requires a 3 message exchange. First, validator $v$ broadcast an uncertified block to all validators in the system. Then, each validator signs the block and sends $v$ their respective signature. Finally, with enough unique signatures ($2f+1$), $v$ can broadcast the certified block to be added to the DAG. Thus, while \textsf{Bullshark}\xspace can commit a block in 2 rounds, the actual commitment process requires 6-message delays \cite{MYSTICETI,BULLSHARK}. \textit{Uncertified} DAGs remove this requirement to certify a block before broadcast, as the DAG itself provides validation. \cite{CORDIALMINERS,MYSTICETI,MAHIMAHI}. 
\subsection{Cordial Miners}
\quad  Keidar et al. \cite{CORDIALMINERS} present the first \textit{uncertified} DAG-based consensus protocol, called \textsf{Cordial Miners}\xspace. As aforementioned, an \textit{uncertified} DAG-based consensus protocol does not require a block to be certified before being added to the DAG. Instead, they prove how, simply by interpreting the DAG in a specific way, certificates can be formed explicitly. The protocol works as follows: A block $B$, for round $r$, must contain references to $2f+1$ unique valid blocks of the previous round $r-1$. Then, for block $B$ to be committed, it requires a supermajority of approval ($2f+1$ unique references to block $B$ in round $r+1$), and a further $2f+1$ blocks in round $r+2$ that observes, in its causal history, this pattern. This process, intuitively, is similar to how \textsf{Narwhal}\xspace requires a quorum to certify a block, however it is completely done in the DAG itself, removing the need to certify blocks. Keidar et al. \cite{CORDIALMINERS} present both an asynchronous and a partially-synchronous version of the protocol, proving that both settings work under an \textit{uncertified} DAG data structure. Interestingly, Keidar et al. \cite{CORDIALMINERS} suggest that \textsf{Cordial Miners}\xspace, and \textit{uncertified} DAG-based consensus protocol in the same vein, can be extended, similarly to \textsf{Bullshark}\xspace, to achieve dual-mode functionality. However, Keidar et al. \cite{CORDIALMINERS} suggest to utilise two leaders, like \textsf{Bullshark}\xspace, one deterministic (partially-synchronous mode) and one random (asynchronous mode) if the deterministic leader fails to commit \cite{BULLSHARK}. Such approach, however, can be rather complicated, and hard to implement.
\subsection{Mysticeti}
\quad Babel et al. \cite{MYSTICETI} introduce \textsf{Mysticeti}\xspace, a family of DAG-based consensus protocols. We focus specifically on \textsf{Mysticeti-C}\xspace, which represents the first implementation of an \textit{uncertified} DAG-based consensus protocol (referred to simply as \textsf{Mysticeti}\xspace throughout this paper). This protocol runs under partially-synchronous assumptions, thus does not require any form of randomness beacon to circumvent the FLP impossibility result. Instead, just like the partially-synchronous version of Bullshark, leverages timeouts to guarantee liveness after $GST$ has been reached \cite{FLP,BULLSHARK}. \textsf{Mysticeti}\xspace runs on what is knowns as a threshold logical clock (TLC) \cite{TLC}. Essentially, rounds advancement in \textsf{Mysticeti}\xspace are dependent on some sort of rules imposed by the system. Thus, \textsf{Mysticeti}\xspace can guarantee that at least one block (usually known as a primary block) can be committed per wave by leveraging timeouts. Compared to the partially-synchronous version of Bullshark, \textsf{Mysticeti}\xspace runs waves consisting of 3-rounds (propose-vote-certify); however, \textsf{Mysticeti}\xspace's blocks do not need to be certified before being proposed. This means that proposing a block is simply a 1-message delay, and to commit a block, it only requires 3-messages delay (compared to the 6-message delay found in the partially-synchronous version of Bullshark). Furthermore, \textsf{Mysticeti}\xspace introduces two additional enhancements that allow for lower latency and higher throughput. \textsf{Mysticeti}\xspace allows up to $n$ (where $n$ is the total amount of validators in the systems) validators to be considered leaders in a round. This approach was inspired by Multi-Paxos \cite{MULTIPAXOS}. Then, \textsf{Mysticeti}\xspace allows for pipe-lining. This means that every round can mark the start of a wave, allowing for a proposed leader block to be committed every round. Both of these enhancements are visualized in Figure~\ref{fig:}, detailed in Appendix~\ref{Appendix:ConsProtStruct}. We will see in Section~\ref{sec:structure_protocol_DAG} how \textsf{Mysticeti}\xspace, and other \textit{uncertified} DAG-based consensus protocols, leverage the DAG structure to make safe commit decision. Not only is \textsf{Mysticeti}\xspace, at the time of writing, the protocol with the highest throughput and lowest latency, but the Sui blockchain has also adopted \textsf{Mysticeti}\xspace as its consensus protocol \cite{SUI}. 
\subsection{Mahi-Mahi}
\quad \textsf{Mahi-Mahi}\xspace is an asynchronous \textit{uncertified} DAG-based consensus protocol proposed by Jovanovic et al. \cite{MAHIMAHI}. Just like \textsf{Mysticeti}\xspace, \textsf{Mahi-Mahi}\xspace introduces \textbf{Multi-Leader} and \textbf{Pipelining} to increase throughput and latency. However, unlike \textsf{Mysticeti}\xspace, waves consist of either 4 or 5 rounds, depending on network assumptions (we will go further in-depth in Section~\ref{sec:structure_protocol_DAG}). The extra rounds are required to increase propagation of a proposed leader block, since, without timeouts, we cannot guarantee that enough validators will see the proposed block before advancing to the next round. By increasing the propagation of a block, we can then, with high-probability, elect a leader in retrospect by utilising a random global coin, as seen with other similar protocols \cite{TUSK,BULLSHARK}. That said, unlike \textsf{Bullshark}\xspace, during the use of its \texttt{fallback} leader, \textsf{Mahi-Mahi}\xspace can commit with either a 4/5-message delay, a significant decrease from the 12-message delay given by \textsf{Bullshark}\xspace's \texttt{fallback} mode \cite{BULLSHARK}. 
\\
\null \quad We have seen incremental improvements in the evolution from linear-chain consensus protocols to DAG-based ones. \textsf{Ditto}\xspace represents a functional example of a dual-mode linear-chain consensus protocol, which achieves results comparable to similar state-of-the-art linear-chain consensus protocols \cite{DITTO}. \textsf{Bullshark}\xspace, likewise, shows that the same dual-mode functionality can be applied to DAG-based consensus protocols, while inheriting all the benefits inherent to DAG-based consensus protocols \cite{BULLSHARK}. \textit{Uncertified} DAG-based consensus protocols currently lack a dual-mode counterpart capable of utilising the inherent advantages of the \textit{uncertified} DAG data structure. Our work aims to address this by presenting the first dual-mode \textit{uncertified} DAG-based consensus protocol based on \textsf{Mysticeti}\xspace and \textsf{Mahi-Mahi}\xspace \cite{MYSTICETI,MAHIMAHI}.
\chapter{The \tarpon Protocol}
\label{chap:dual_mode_protocol}
\quad As we have talked about in Section~\ref{sec:uncertified_DAG}, while both \textsf{Mysticeti}\xspace and \textsf{Mahi-Mahi}\xspace are the first uncertified DAG-based consensus protocols that leverage the DAG structure itself to form certificates, both have clear strengths and weaknesses given their initial assumptions. Thus, we introduce \tarpon, the first dual-mode \textit{uncertified} DAG-based consensus protocol that combines \textsf{Mysticeti}\xspace and \textsf{Mahi-Mahi}\xspace to achieve liveness in both asynchronous and partially synchronous environments, while maintaining low latency comparable to \textsf{Mysticeti}\xspace. \tarpon differs from other dual-mode protocols by having a simplified approach to running its fallback mode (i.e., the asynchronous mode). As we have seen, whenever a block is committed in a DAG-based consensus protocol, its entire causal history is also committed. Thus, we simply need to run the asynchronous mode periodically at fixed, or changing, intervals, to achieve liveness in asynchronous settings. 
\section{System Design and Core Assumptions}
\quad We consider a communication network where $n=3f+1$ validators process transactions using the \tarpon protocol. An adversary can corrupt up to a set of $f$ validators, which we will call \textit{Byzantine}. These malicious validators can act arbitrarily (i.e can deviate from the protocol). \textit{Honest} validators, instead, are the remaining non-corrupted validators who will follow the protocol. The protocol makes the assumption that the links between validators are reliable and authenticated. This means that messages sent by honest validators will, eventually, reach all other honest parties, and that the identity of the sender can be easily verified by the receiver. Furthermore, the adversary is computationally bounded. Thus, the standard cryptographic properties utilised in the protocol (i.e. hash functions, digital signatures etc...) will hold. Under these assumptions we can say that \tarpon is \textit{safe} (this is proven in Chapter~\ref{chap:Proofs}) because no two honest validators will ever commit contradictory transactions. The dual-mode protocol runs under an adaptive network setting, meaning that it can make progress ( \textit{liveness}) in both an asynchronous setting, and a partially-synchronous setting. Before $GST$ is reached, messages may be delayed arbitrary, but will be delivered eventually with probability 1. Thus, \tarpon achieves \textit{liveness} during periods before $GSTS$ is reached. Furthermore, after $GST$, \tarpon achieves \textit{liveness} as transactions are guaranteed to commit within $\Delta$. Both these statements are proven in Section~\ref{chap:Proofs}. 
\section{Structure of the Dual-Mode Protocol DAG}
\quad As with \textsf{Mysticeti}\xspace and \textsf{Mahi-Mahi}\xspace, \tarpon works on a series of \textit{logical rounds} \cite{MYSTICETI, MAHIMAHI}. During every round, a validator proposes a unique, signed block (i.e. a single block per round) which contains a number of transactions; however, malicious validators can equivocate, meaning that they propose more than one block per round. Throughout a round, validators receive transactions from \textit{users} and blocks from other validators. Together, a validator creates a block that consists of \textit{fresh transactions} (i.e. transactions that have not been included in previous blocks), alongside hash references to previous blocks. The newly created block is broadcast when the honest validator includes at least $2f+1$ hash references to blocks from the previous round, after which the validator signs the block. There can be scenarios where the inclusion of a transaction, sent by a user to a validator, fails (e.g. a malicious validator does not want to include this specific transaction into the block). In this case, the user, after a fixed amount of time, will send the transaction to a different validator.  
\label{sec:structure_protocol_DAG}
\subsection{Block Creation} 
\quad A block proposed by an honest validator is structured as follows: 
\begin{figure}[H]
    \centering
\[
\begin{array}{c}
\textbf{Block Structure} \\
\hline
\begin{array}{ll}
\text{ Author }  & : \text{ Validator ID and signature on block contents} \\
\text{ Round }  & : \text{ Round number} \\
\text{ Transactions} & : \text{ List of transactions} \\
\text{ References} & : \text{ Minimum } 2f+1 \text{ distinct hashes from round } r-1 \\
\text{ Coin Share} & : \text{ Share of the global perfect coin}
\end{array}
\end{array}
\]
\end{figure}
\quad A block is considered to be valid if: the signature of validator $v$ is valid, and $v$ is part of the validator set; all references point to different valid blocks from round $r-1$, and the first referenced block points to $v$'s block from round $r-1$; lastly, the share of the global coin is valid, and as long as the global coin is implemented via a threshold signature scheme, the shares of a global coin can be verified.  
\subsection{Round and Wave} 
\label{subsection:roundandwave}
\quad \tarpon defines two types of wave structures, where the wave length is denoted as $w$, a partially-synchronous one and an asynchronous one, where rounds are given specific roles respective of the wave structure. 
\\
\null \quad The partially-synchronous wave consists of 3 rounds, which, just like \textsf{Mysticeti}\xspace, allows for the lowest bound, of a 3-message delay, for a block to be committed \cite{PBFT}. A leader of a wave first proposes a block (\textbf{propose} round), then all other validators vote on the block (\textbf{vote} round), and lastly, they collect certificates (\textbf{certify} round).

\quad The asynchronous wave, instead, consists of either 4 or 5 rounds. As explained by Jovanovic et al. \cite{MAHIMAHI}, a wave length of 5 maximises resilience against an active asynchronous adversary, meaning that it is more likely to commit a block if a wave is composed of five rounds. However, under more realistic asynchronous scenarios, a wave length of 4 allows for lower latency while maintaining sufficient commit probability. Due to asynchronous assumptions, we cannot rely on timeouts to guarantee that a leader proposal will reach all honest validators in round $r+1$. Boost rounds act as propagation rounds for the proposal before the voting round, where longer wavelengths create more boost round. Since \tarpon's asynchronous mode closely follows \textsf{Mahi-Mahi}\xspace's, it inherits analysis under both asynchronous and random network models. The \textbf{asynchronous network model} makes the worst-case assumption (i.e. $f$ nodes are always faulty), while the \textbf{random network model} makes average-case assumptions. This means that when a validator proposes a block $b$ at round $r$, $b$ will include $2f+1$ references from round $r-1$ uniformly-chosen at random. The proofs of \textit{liveness}, provided by Jovanovic et al. \cite{MAHIMAHI}, demonstrate how \textit{liveness} is achieved both when $w=4$ and $w=5$. Specifically, when $w=4$, there are weaker \textit{liveness} guarantees, especially under the \textit{asynchronous network model}; however, under the \textit{random network model}, the protocol has a high-probability of committing each wave. 
\subsection{DAG Patterns} 
\label{sec:DAG_Patterns}
\quad With the structure of waves described in the earlier section, a validator can identify DAG patterns to decide the outcome of proposed blocks:
\begin{itemize}
    \item \textit{Certificate Pattern:} A certificate pattern is identified when there are at least $2f+1$ blocks in round $r+w-1$ that support a block proposed in round $r$. This means that $2f+1$ blocks in round $r+w-1$ have a reference to block $B$. Visualized in Figure~\ref{fig:mysticeti_direct_commit}.
    \item \textit{Skip Pattern:} A skip pattern is identified when there are at least $2f+1$ blocks in round $r+w-1$ (Voting round) that do not support a block proposed in round $r$. This means that $2f+1$ blocks in round $r+w-1$ do not have a reference to block $B$. Visualized in Figure~\ref{fig:mysticeti_direct_skip}.
\end{itemize}
\quad While \textsf{Bullshark}\xspace only allows certified blocks to be proposed every round, in the case of \tarpon analysing the DAG by identifying the aforementioned patterns allows us to determine certificates implicitly. As mentioned, this is one of the major innovations that has led protocols such as \textsf{Mysticeti}\xspace and \textsf{Mahi-Mahi}\xspace to achieve such low levels of latency \cite{MYSTICETI,MAHIMAHI}. Essentially by moving the certification process of \textsf{Narwhal}\xspace into the DAG itself, removing the 3-message delay that was required to certify blocks, and instead, certifying blocks directly by analysing the DAG \cite{TUSK}. 
\begin{figure}[H]
    \begin{adjustwidth}{-1cm}{-1cm} 
\captionsetup{width=0.4\textwidth, font=small}
    \centering
    \begin{minipage}{0.45\textwidth}
        \centering
\begin{figure}[H]
\centering
\begin{tikzpicture}[
    node distance=1.5cm and 2.5cm,
    validator/.style={draw, minimum size=0.8cm},
    propose/.style={validator, thick},
    boost/.style={validator, draw=black},
    vote/.style={validator, draw=orange, thick},
    certify/.style={validator, draw=green!60!black, thick},
    supporting_validator/.style={validator, fill=green!20},
    support/.style={->, green!60!black, thick},
    >=Stealth,
    leader/.style={validator, fill=green!20},
    vote_support/.style={->, orange, thick},
    skipping_validator/.style={validator, fill=pink!40},
    cert_support/.style={->, green!60!black, thick},
    >=Stealth
]
\footnotesize
\node at (0, 5) {$r$};
\node at (2.5, 5) {$r+1$};
\node at (5, 5) {$r+2$};

\node at (-1.5, 4) {$v_0$};
\node at (-1.5, 3) {$v_1$};
\node at (-1.5, 2) {$v_2$};
\node at (-1.5, 1) {$v_3$};

\node[leader] (v0r) at (0, 4) {$L1_{a}$};
\node[boost] (v1r) at (0, 3) {};
\node[boost] (v2r) at (0, 2) {};
\node[boost] (v3r) at (0, 1) {$L1_{b}$};

\node[supporting_validator] (v0r1) at (2.5, 4) {};
\node[supporting_validator] (v1r1) at (2.5, 3) {};
\node[supporting_validator] (v2r1) at (2.5, 2) {};
\node[boost] (v3r1) at (2.5, 1) {};

\node[supporting_validator] (v0r2) at (5, 4) {};
\node[supporting_validator] (v1r2) at (5, 3) {};
\node[supporting_validator] (v2r2) at (5, 2) {};
\node[supporting_validator] (v3r2) at (5, 1) {};

\draw[support] (v0r1.west) -- (v0r);
\draw[->] (v0r1.west) -- (v1r);
\draw[->] (v0r1.west) -- (v2r);

\draw[support] (v1r1.west) -- (v0r);
\draw[->] (v1r1.west) -- (v1r);
\draw[->] (v1r1.west) -- (v2r);

\draw[support] (v2r1.west) -- (v0r);
\draw[->] (v2r1.west) -- (v1r);
\draw[->] (v2r1.west) -- (v2r);

\draw[->] (v3r1.west) -- (v1r);
\draw[->] (v3r1.west) -- (v2r);
\draw[->] (v3r1.west) -- (v3r);

\draw[support] (v0r2.west) -- (v0r1);
\draw[support] (v0r2.west) -- (v1r1);
\draw[support] (v0r2.west) -- (v2r1);

\draw[support] (v1r2.west) -- (v0r1);
\draw[support] (v1r2.west) -- (v1r1);
\draw[support] (v1r2.west) -- (v2r1);

\draw[support] (v2r2.west) -- (v0r1);
\draw[support] (v2r2.west) -- (v1r1);
\draw[support] (v2r2.west) -- (v2r1);

\draw[support] (v3r2.west) -- (v1r1);
\draw[support] (v3r2.west) -- (v2r1);
\draw[->] (v3r2.west) -- (v3r1);

\node[draw, dashed, rounded corners, fit=(v0r) (v3r) (v0r2) (v3r2), inner sep=8pt, label={[align=center]below:$w$}] {};
\end{tikzpicture}
\caption{Certificate Pattern: $L1_a$ is marked as \texttt{to-commit} because validator $v_0$, $v_1$ and $v_2$ voted for the proposed block in round $r+1$, and $2f+1$ constrain in their history this pattern in round $r+2$.}
\label{fig:mysticeti_direct_commit}
\end{figure}
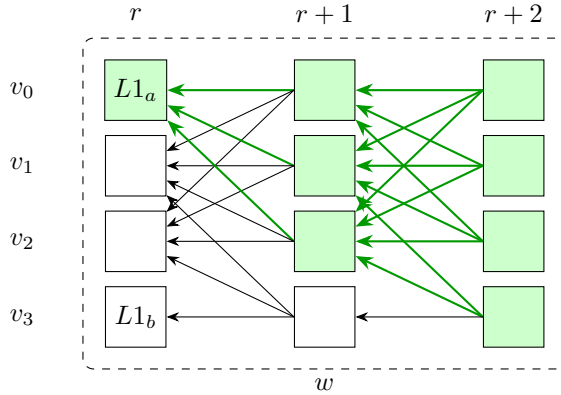
    \end{minipage}
    \hspace{0.08\textwidth}
    \begin{minipage}{0.45\textwidth}
        \centering
    \begin{figure}[H]
    \centering
    \begin{tikzpicture}[
    node distance=1.5cm and 2.5cm,
    validator/.style={draw, minimum size=0.8cm},
    propose/.style={validator, thick},
    boost/.style={validator, draw=black},
    vote/.style={validator, draw=orange, thick},
    certify/.style={validator, draw=green!60!black, thick},
    supporting_validator/.style={validator, fill=green!20},
    support/.style={->, green!60!black, thick},
    >=Stealth,
    leader/.style={validator, fill=red!20},
    vote_support/.style={->, orange, thick},
    skipping_validator/.style={validator, fill=pink!40},
    cert_support/.style={->, green!60!black, thick},
    >=Stealth
]
\footnotesize
\node at (0, 5) {$r$};
\node at (2.5, 5) {$r+1$};
\node at (5, 5) {$r+2$};

\node at (-1.5, 4) {$v_0$};
\node at (-1.5, 3) {$v_1$};
\node at (-1.5, 2) {$v_2$};
\node at (-1.5, 1) {$v_3$};

\node[boost] (v0r) at (0, 4) {$L1_{a}$};
\node[boost] (v1r) at (0, 3) {};
\node[boost] (v2r) at (0, 2) {};
\node[boost, leader] (v3r) at (0, 1) {$L1_b$};

\node[skipping_validator] (v0r1) at (2.5, 4) {};
\node[skipping_validator] (v1r1) at (2.5, 3) {};
\node[skipping_validator] (v2r1) at (2.5, 2) {};
\node[boost] (v3r1) at (2.5, 1) {};

\node[boost] (v0r2) at (5, 4) {};
\node[boost] (v1r2) at (5, 3) {};
\node[boost] (v2r2) at (5, 2) {};
\node[boost] (v3r2) at (5, 1) {};

\draw[->] (v0r1.west) -- (v0r);
\draw[->] (v0r1.west) -- (v1r);
\draw[->] (v0r1.west) -- (v2r);

\draw[->] (v1r1.west) -- (v0r);
\draw[->] (v1r1.west) -- (v1r);
\draw[->] (v1r1.west) -- (v2r);

\draw[->] (v2r1.west) -- (v0r);
\draw[->] (v2r1.west) -- (v1r);
\draw[->] (v2r1.west) -- (v2r);

\draw[->] (v3r1.west) -- (v1r);
\draw[->] (v3r1.west) -- (v2r);
\draw[->] (v3r1.west) -- (v3r);

\draw[->] (v0r2.west) -- (v0r1);
\draw[->] (v0r2.west) -- (v1r1);
\draw[->] (v0r2.west) -- (v2r1);

\draw[->] (v1r2.west) -- (v0r1);
\draw[->] (v1r2.west) -- (v1r1);
\draw[->] (v1r2.west) -- (v2r1);

\draw[->] (v2r2.west) -- (v0r1);
\draw[->] (v2r2.west) -- (v1r1);
\draw[->] (v2r2.west) -- (v2r1);

\draw[->] (v3r2.west) -- (v1r1);
\draw[->] (v3r2.west) -- (v2r1);
\draw[->] (v3r2.west) -- (v3r1);

\node[draw, dashed, rounded corners, fit=(v0r) (v3r) (v0r2) (v3r2), inner sep=8pt, label={[align=center]below:$w$}] {};

\end{tikzpicture}
    \caption{Skip Pattern: Leader block $L1_d$ is marked as \texttt{to-skip} because validator $v_0$, $v_1$ and $v_2$ did not vote for $L1_d$ (i.e. $L1_d$ has $2f+1$ blames) in round $r+1$, thus creating a \textit{skip pattern}.}
    \label{fig:mysticeti_direct_skip}
\end{figure}
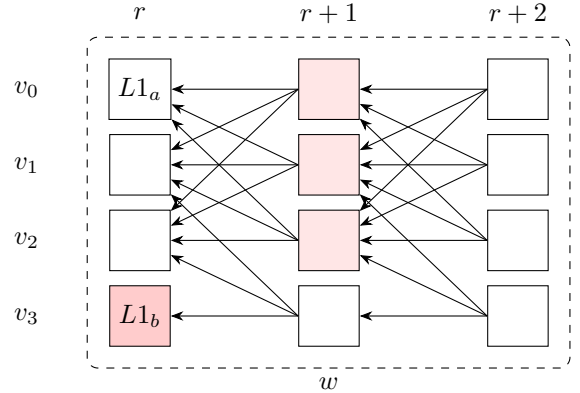
    \end{minipage}
    \end{adjustwidth}
\captionsetup{width=0.9\textwidth, font=small}
\end{figure}
\vspace{1em}
\subsection{Liveness} 
\quad Due to \tarpon's properties, \textit{liveness} can be achieved in two different scenarios: before  $GST$ has been reached, and after. Chapter~\ref{chap:Proofs} provides formal proofs for both scenarios, but, intuitively, \textit{liveness} is achieved as follows:
\begin{itemize}
    \item \textbf{Case 1 (before $GST$)}: In \tarpon's asynchronous mode we make no timing assumptions, but instead we leverage a perfect global coin to thwart possible malicious participants from manipulating the formation of certificates. As stated earlier, the FLP impossibility result states that deterministic protocols cannot ensure liveness \textbf{and} safety in the presence of even one faulty validator \cite{FLP}. By using the perfect global coin, no validator will know whether they are the leader of a specific round until enough rounds have passed to form a certificate for said round. Thus, probabilistically there is a non zero chance of committing a leader block, ensuring \textit{liveness}, eventually \cite{MAHIMAHI}. 
    \item \textbf{Case 1 (after $GST$)}: In the partially-synchronous mode we rely on timeouts to maintain \textit{liveness}. We know that all honest validators will receive messages sent by other honest participants within $\Delta+max(GST,t)$, where $t$ represents the time a message has been sent by an honest validator. Every round a proposer is deterministically selected to be considered the primary block of the round. For example, assume a primary block $B$ has been proposed in round $r$. Other honest validators will either need to include this block's reference or wait a timeout (i.e. $\Delta$) before being allowed to propose a block in round $r+1$. Furthermore, if an honest validator includes a reference to the primary block in the block proposed in $r+1$, they have to wait an additional timeout before being allowed to propose a block for $r+2$. This ensures that there will always exist a certificate for the primary block, after $GST$ has been reached \cite{MYSTICETI}.
\end{itemize}
\subsection{Equivocation Tolerance} \quad Equivocation is handled by construction. When a validator equivocates by proposing multiple blocks for the same round, honest validators may receive all equivocating blocks. However, each honest validator can only support at most one block per slot $s$ when constructing their own block proposals. Due to the $2f+1$ votes required for certification and quorum intersection, at most one of the equivocating blocks can gather sufficient support to be certified and marked as \texttt{to-commit}. This is exemplified in Figure~\ref{fig:equivocation}. Although the figure demonstrates the partially-synchronous mode, the asynchronous mode handles equivocation in the same manner. 
\begin{figure}[H]
    \centering
    \begin{tikzpicture}[
    node distance=1.5cm and 2.5cm,
    validator/.style={draw, minimum size=0.8cm},
    propose/.style={validator},
    boost/.style={validator, draw=black},
    vote/.style={validator, draw=orange, thick},
    certify/.style={validator, draw=green!60!black, thick},
    supporting_validator/.style={validator, fill=green!20},
    support/.style={->, green!60!black, thick},
    >=Stealth,
    leader/.style={validator, fill=green!20},
    vote_support/.style={->, orange, },
    skipping_validator/.style={validator, fill=pink!40},
    cert_support/.style={->, green!60!black, thick},
    equivocating/.style={validator, fill=red!20},
    >=Stealth,
equivocatingcommit/.style={validator, fill=green!20},
    >=Stealth
]
\footnotesize
\node at (0, 5.4) {$r$};
\node at (2.5, 5.4) {$r+1$};
\node at (5, 5.4) {$r+2$};
\node at (-1.5, 4) {$v_0$};
\node at (-1.5, 3) {$v_1$};
\node at (-1.5, 2) {$v_2$};
\node at (-1.5, 1) {$v_3$};
\node[equivocatingcommit] (v0r) at (0, 4) {$L1_{a}$};
\node[equivocating] (v3r_R) at (0.4, 4.5) {\scriptsize$L1_{a'}$};
\node[propose] (v1r) at (0, 3){};
\node[propose] (v2r) at (0, 2){};
\node[propose] (v3r_L) at (0, 1){};
\node[propose] (v0r1) at (2.5, 4) {};
\node[propose] (v1r1) at (2.5, 3) {};
\node[propose] (v2r1) at (2.5, 2) {};
\node[propose] (v3r1) at (2.5, 1) {};
\node[propose] (v0r2) at (5, 4) {};
\node[propose] (v1r2) at (5, 3) {};
\node[propose] (v2r2) at (5, 2) {};
\node[propose] (v3r2) at (5, 1) {};
\draw[->] (v0r1.west) -- (v1r);
\draw[->] (v0r1.west) -- (v3r_R);
\draw[->] (v0r1.west) -- (v2r);
\draw[->] (v1r1.west) -- (v0r);
\draw[->] (v1r1.west) -- (v1r);
\draw[->] (v1r1.west) -- (v2r);
\draw[->] (v2r1.west) -- (v1r);
\draw[->] (v2r1.west) -- (v2r);
\draw[->] (v2r1.west) -- (v0r);
\draw[->] (v3r1.west) -- (v1r);
\draw[->] (v3r1.west) -- (v3r_L); 
\draw[->] (v3r1.west) -- (v0r); 
\draw[->] (v0r2.west) -- (v0r1);
\draw[->] (v0r2.west) -- (v1r1);
\draw[->] (v0r2.west) -- (v2r1);
\draw[->] (v1r2.west) -- (v0r1);
\draw[->] (v1r2.west) -- (v1r1);
\draw[->] (v1r2.west) -- (v2r1);
\draw[->] (v2r2.west) -- (v0r1);
\draw[->] (v2r2.west) -- (v1r1);
\draw[->] (v2r2.west) -- (v2r1);
\draw[->] (v3r2.west) -- (v1r1);
\draw[->] (v3r2.west) -- (v2r1);
\draw[->] (v3r2.west) -- (v3r1);
\node[draw, dashed, rounded corners, fit=(v0r) (v3r_L) (v3r_R) (v0r2) (v3r2), inner sep=8pt, label={[align=center]below:$w$}] {};
\end{tikzpicture}
    \caption{Example of equivocation tolerance found in \tarpon. Validator $v_0$ is the equivocating validators (Blocks $L1_a$ and $L'1_a$ are proposed during the same round $r$). However, since the remaining honest validators will only support one block per validator per round, only block $L1_a$ receives enough votes to be eventually committed.}
    \label{fig:equivocation}
\end{figure}
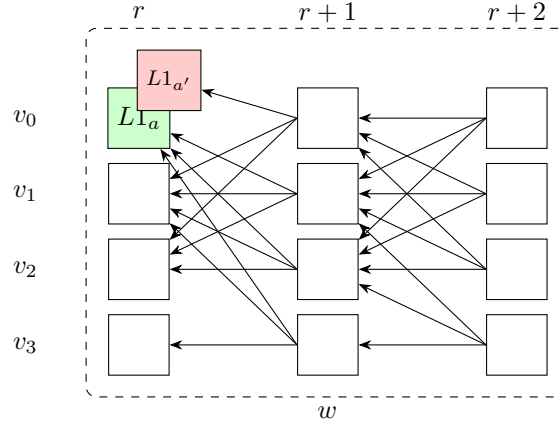
\vspace{1em}
\section{The \tarpon Consensus Protocol}
\subsection{Proposer Slot and Leader Selection}
\label{sec:prop_slot_leader}
\paragraph{Proposer Slot} Just like \textsf{Mysticeti}\xspace and \textsf{Mahi-Mahi}\xspace, \tarpon utilises proposer slots. Proposers slots can either be empty, or include the block proposed by a validator for a specific round (forming a tuple (validator, round)). Each proposer slot can have the following states: \texttt{to-commit}, \texttt{to-skip} and \texttt{undecided}. We have seen how a proposer slot can be given the state of the former two, by identification of DAG Patterns. However, the \texttt{undecided} state is what differentiates \tarpon (including \textsf{Mysticeti}\xspace and \textsf{Mahi-Mahi}\xspace) from previous works, such as \textsf{Bullshark}\xspace \cite{MYSTICETI,MAHIMAHI,BULLSHARK}. This state forces all following slots to wait for this one to be marked, ensuring a deterministic sequence of proposers for all validators. Furthermore, each round can have up to $n=3f+1$ proposer slots, allowing each block to be directly committed, decreasing latency in ideal settings. However, as we will see in the decision rules, having too many proposer slots can also increase latency when proposers are either slow (e.g., during periods of asynchrony) or act maliciously (e.g. they equivocate) \cite{MYSTICETI}.  
\paragraph{Leader selection} Leader selection acts differently depending on the mode being run. During the partially-synchronous mode, where we have the ability to use timeouts, leaders are selected in a deterministic fashion (i.e. via a round-robin algorithm). However, during \tarpon's asynchronous mode, we cannot have a deterministic approach in leader selection, or and adversary may be able to exploit the lack of timeouts to halt the system. Thus, as we have said in Section~\ref{sec:network_models}, we utilise $2f+1$ shares of the global perfect coin to decide the leader after the fact. In \tarpon's asynchronous mode, assume that $r$ is the propose round, the global perfect coin is reconstructed in either $r+4$ or $r+3$ depending on \tarpon's wave length configuration. It is important to note that the global perfect coin can elect multiple leaders, up to $3f+1$ leaders, and can also select slots whose validator has not proposed any blocks. 
\subsection{Decision Rules}
\label{subsection:decision_rules}
\quad We will now go over the \tarpon's decision rules. \tarpon follows closely the decision rules found in both \textsf{Mysticeti}\xspace and \textsf{Mahi-Mahi}\xspace, aiming to try and classify every proposer slot as either \texttt{to-commit} or \texttt{to-skip}, where all proposer slots start in an \texttt{undecided} state. 
\subsubsection{Direct Rule}
\quad \tarpon begins by applying the direct rule on a DAG, starting from the latest proposer slot. The direct rule tries to mark a slot as either \texttt{to-commit} or \texttt{to-skip} by either observing $2f+1$ \textit{commit patterns} or a \textit{skip pattern} for that slot. For example, as seen in Figure~\ref{fig:mysticeti_direct_commit}, $L1_a$ is initially marked as \texttt{undecided}, but, via the direct rule, can be marked as \texttt{to-commit} because of the $2f+1$ certificates present. In Figure~\ref{fig:mysticeti_direct_skip}, instead, $L1_d$ is marked as \texttt{to-skip}, via the direct rule, as a result of a \textit{skip pattern}. the direct rule works comparably during \tarpon's asynchronous mode. If \tarpon observes neither $2f+1$ \textit{commit patterns} nor a \textit{skip pattern}, then the slot is left as \texttt{undecided}, and the protocol will require the \textit{indirect rule} to mark the slot accordingly. Examples of \tarpon's direct rule for its asynchronous mode can be found in Appendix~\ref{tarpon_direct_rule}.
\subsubsection{Indirect Rule}
\quad We now move onto the cases where the direct rule fails, thus \tarpon has to rely on the indirect rule to mark a proposer slot. The indirect rule works by looking at future slots to determine whether the \texttt{undecided} slot can be marked as \texttt{to-commit} or \texttt{to-skip}. Assume a slot $s$, in round $r$, has been marked as \texttt{undecided} via the direct rule. The indirect rule first needs to locate an \textit{anchor}, meaning a slot where its round $r'$ is greater than the decision round of slot $s$ (i.e. $r'>r+w-1$). If the anchor is marked as \texttt{undecided} then slot $s$ is also marked as such. If the \textit{anchor} is marked as \texttt{to-commit}, then the rule checks that the anchor contains at least $1$ certificate for slot $s$. If the certificate link exists, then slot $s$ is marked as \texttt{to-commit}, if it does not, then the slot is marked as \texttt{to-skip}. If an encountered \textit{anchor} has been marked as \texttt{to-skip}, the rule moves onto the next available \textit{anchor}. 
\begin{figure}[H]
\centering
\begin{tikzpicture}[
    node distance=1.5cm and 2.5cm,
    validator/.style={draw, minimum size=0.8cm},
    propose/.style={validator, thick},
    boost/.style={validator, draw=black},
    vote/.style={validator, draw=orange, thick},
    certify/.style={validator, draw=green!60!black, thick},
    supporting_validator/.style={validator, fill=green!20},
    support/.style={->, green!60!black, thick},
    >=Stealth,
    leader/.style={validator, fill=green!20},
    vote_support/.style={->, orange, thick},
    skipping_validator/.style={validator, fill=pink!40},
    cert_support/.style={->, green!60!black, thick},
    >=Stealth
]
\footnotesize
\node at (0, 5) {$r$};
\node at (2.5, 5) {$r+1$};
\node at (5, 5) {$r+2$};
\node at (7.5, 5) {$r+3$};
\node at (10, 5) {$r+4$};
\node at (12.5, 5) {$r+5$};

\node at (-1.5, 4) {$v_0$};
\node at (-1.5, 3) {$v_1$};
\node at (-1.5, 2) {$v_2$};
\node at (-1.5, 1) {$v_3$};
\node[boost, leader] (v0r) at (0, 4) {$L1_{a}$};
\node[boost] (v1r) at (0, 3) {};
\node[boost] (v2r) at (0, 2) {};
\node[boost] (v3r) at (0, 1) {};

\node[supporting_validator] (v0r1) at (2.5, 4) {};
\node[boost] (v1r1) at (2.5, 3) {};
\node[boost] (v2r1) at (2.5, 2) {};
\node[boost] (v3r1) at (2.5, 1) {};

\node[supporting_validator] (v0r2) at (5, 4) {};
\node[boost] (v1r2) at (5, 3) {};
\node[boost] (v2r2) at (5, 2) {};
\node[boost] (v3r2) at (5, 1) {};

\node[supporting_validator] (v0r3) at (7.5, 4) {};
\node[supporting_validator] (v1r3) at (7.5, 3) {};
\node[supporting_validator] (v2r3) at (7.5, 2) {};
\node[boost] (v3r3) at (7.5, 1) {};

\node[supporting_validator] (v0r4) at (10, 4) {};
\node[boost] (v2r4) at (10, 2) {};
\node[boost] (v3r4) at (10, 1) {};

\node[boost] (v0r5) at (12.5, 4) { };
\node[boost, leader] (v1r5) at (12.5, 3) {$L2_a$};
\node[boost] (v2r5) at (12.5, 2) {};
\node[boost] (v3r5) at (12.5, 1) {};

\draw[support] (v0r1.west) -- (v0r);
\draw[->] (v0r1.west) -- (v1r);
\draw[->] (v0r1.west) -- (v2r);

\draw[->] (v1r1.west) -- (v0r);
\draw[->] (v1r1.west) -- (v1r);
\draw[->] (v1r1.west) -- (v2r);

\draw[->] (v2r1.west) -- (v0r);
\draw[->] (v2r1.west) -- (v1r);
\draw[->] (v2r1.west) -- (v2r);

\draw[->] (v3r1.west) -- (v1r);
\draw[->] (v3r1.west) -- (v2r);
\draw[->] (v3r1.west) -- (v3r);

\draw[support] (v0r2.west) -- (v0r1);
\draw[->] (v0r2.west) -- (v1r1);
\draw[->] (v0r2.west) -- (v2r1);

\draw[->] (v1r2.west) -- (v0r1);
\draw[->] (v1r2.west) -- (v1r1);
\draw[->] (v1r2.west) -- (v2r1);

\draw[->] (v2r2.west) -- (v0r1);
\draw[->] (v2r2.west) -- (v1r1);
\draw[->] (v2r2.west) -- (v2r1);

\draw[->] (v3r2.west) -- (v1r1);
\draw[->] (v3r2.west) -- (v2r1);
\draw[->] (v3r2.west) -- (v3r1);

\draw[support] (v0r3.west) -- (v0r2);
\draw[->] (v0r3.west) -- (v1r2);
\draw[->] (v0r3.west) -- (v2r2);

\draw[support] (v1r3.west) -- (v0r2);
\draw[->] (v1r3.west) -- (v1r2);
\draw[->] (v1r3.west) -- (v2r2);

\draw[support] (v2r3.west) -- (v0r2);
\draw[->] (v2r3.west) -- (v1r2);
\draw[->] (v2r3.west) -- (v2r2);

\draw[->] (v3r3.west) -- (v1r2);
\draw[->] (v3r3.west) -- (v2r2);
\draw[->] (v3r3.west) -- (v3r2);

\draw[support] (v0r4.west) -- (v0r3);
\draw[support] (v0r4.west) -- (v1r3);
\draw[support] (v0r4.west) -- (v2r3);

\draw[->] (v2r4.west) -- (v0r3);
\draw[->] (v2r4.west) -- (v1r3);
\draw[->] (v2r4.west) -- (v2r3);

\draw[->] (v3r4.west) -- (v1r3);
\draw[->] (v3r4.west) -- (v2r3);
\draw[->] (v3r4.west) -- (v3r3);

\draw[->] (v0r5.west) -- (v0r4);
\draw[->] (v0r5.west) -- (v3r4);
\draw[->] (v0r5.west) -- (v2r4);

\draw[support] (v1r5.west) -- (v0r4);
\draw[->] (v1r5.west) -- (v3r4);
\draw[->] (v1r5.west) -- (v2r4);

\draw[->] (v2r5.west) -- (v0r4);
\draw[->] (v2r5.west) -- (v3r4);
\draw[->] (v2r5.west) -- (v2r4);

\draw[->] (v3r5.west) -- (v0r4);
\draw[->] (v3r5.west) -- (v2r4);
\draw[->] (v3r5.west) -- (v3r4);

\node[draw, dashed, rounded corners, fit=(v0r) (v3r) (v0r4) (v3r4), inner sep=8pt, label={[align=center]below:$w$}] {};

\end{tikzpicture}
\caption{Example of the indirect rule marking $L1_a$ as \texttt{to-commit} in the asynchronous mode (with wave length $w=5$). While $L1_a$ does have the required $2f+1$ votes from blocks in round $r+3$, it is missing the necessary $2f+1$ certificates that are formed when observing the $2f+1$ votes. Since the anchor $L2_a$ does observe a certificate link between $L2_a$ and $L1_a$, $L1_a$ can be marked as \texttt{to-commit}.} 
\label{fig:mahi_mahi_indirect_commit}
\end{figure}
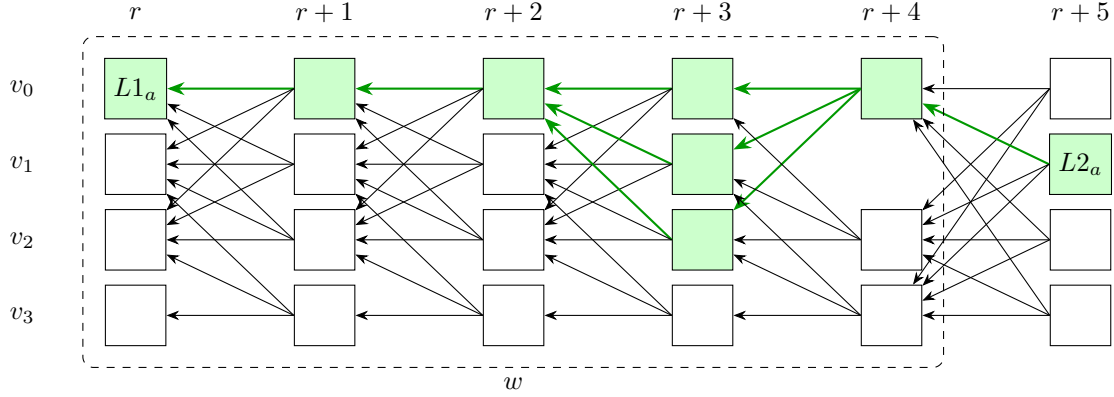
\begin{figure}[H]
\centering
\begin{tikzpicture}[
    node distance=1.5cm and 2.5cm,
    validator/.style={draw, minimum size=0.8cm},
    propose/.style={validator, thick},
    boost/.style={validator, draw=black},
    vote/.style={validator, draw=orange, thick},
    certify/.style={validator, draw=green!60!black, thick},
    supporting_validator/.style={validator, fill=green!20},
    support/.style={->, green!60!black, thick},
    nosupport/.style={->, red!60!black, thick},
    >=Stealth,
    leader/.style={validator, fill=green!20},
    vote_support/.style={->, orange, thick},
    skipping_validator/.style={validator, fill=pink!40},
    cert_support/.style={->, green!60!black, thick},
    >=Stealth
]
\footnotesize
\node at (0, 5) {$r$};
\node at (2.5, 5) {$r+1$};
\node at (5, 5) {$r+2$};
\node at (7.5, 5) {$r+3$};
\node at (10, 5) {$r+4$};
\node at (12.5, 5) {$r+5$};

\node at (-1.5, 4) {$v_0$};
\node at (-1.5, 3) {$v_1$};
\node at (-1.5, 2) {$v_2$};
\node at (-1.5, 1) {$v_3$};
\node[boost] (v0r) at (0, 4) {};
\node[boost] (v1r) at (0, 3) {};
\node[boost] (v2r) at (0, 2) {};
\node[skipping_validator] (v3r) at (0, 1) {$L1_a$};

\node[boost] (v0r1) at (2.5, 4) {};
\node[boost] (v1r1) at (2.5, 3) {};
\node[boost] (v2r1) at (2.5, 2) {};
\node[skipping_validator] (v3r1) at (2.5, 1) {};

\node[boost] (v0r2) at (5, 4) {};
\node[boost] (v1r2) at (5, 3) {};
\node[boost] (v2r2) at (5, 2) {};
\node[skipping_validator] (v3r2) at (5, 1) {};

\node[boost] (v0r3) at (7.5, 4) {};
\node[boost] (v1r3) at (7.5, 3) {};
\node[skipping_validator] (v2r3) at (7.5, 2) {};
\node[skipping_validator] (v3r3) at (7.5, 1) {};

\node[boost] (v0r4) at (10, 4) {};
\node[boost] (v1r4) at (10, 3) {};
\node[boost] (v2r4) at (10, 2) {};
\node[skipping_validator] (v3r4) at (10, 1) {};

\node[boost] (v0r5) at (12.5, 4) {};
\node[boost] (v1r5) at (12.5, 3) {};
\node[skipping_validator] (v2r5) at (12.5, 2) {$L2_a$};
\node[boost] (v3r5) at (12.5, 1) {};

\draw[->] (v0r1.west) -- (v0r);
\draw[->] (v0r1.west) -- (v1r);
\draw[->] (v0r1.west) -- (v2r);

\draw[->] (v1r1.west) -- (v0r);
\draw[->] (v1r1.west) -- (v1r);
\draw[->] (v1r1.west) -- (v2r);

\draw[->] (v2r1.west) -- (v0r);
\draw[->] (v2r1.west) -- (v1r);
\draw[->] (v2r1.west) -- (v2r);

\draw[->] (v3r1.west) -- (v1r);
\draw[->] (v3r1.west) -- (v2r);
\draw[nosupport] (v3r1.west) -- (v3r);

\draw[->] (v0r2.west) -- (v0r1);
\draw[->] (v0r2.west) -- (v1r1);
\draw[->] (v0r2.west) -- (v2r1);

\draw[->] (v1r2.west) -- (v0r1);
\draw[->] (v1r2.west) -- (v1r1);
\draw[->] (v1r2.west) -- (v2r1);

\draw[->] (v2r2.west) -- (v0r1);
\draw[->] (v2r2.west) -- (v1r1);
\draw[->] (v2r2.west) -- (v2r1);

\draw[->] (v3r2.west) -- (v1r1);
\draw[->] (v3r2.west) -- (v2r1);
\draw[nosupport] (v3r2.west) -- (v3r1);

\draw[->] (v0r3.west) -- (v0r2);
\draw[->] (v0r3.west) -- (v1r2);
\draw[->] (v0r3.west) -- (v2r2);

\draw[->] (v1r3.west) -- (v0r2);
\draw[->] (v1r3.west) -- (v1r2);
\draw[->] (v1r3.west) -- (v2r2);

\draw[->] (v2r3.west) -- (v0r2);
\draw[->] (v2r3.west) -- (v1r2);
\draw[->] (v2r3.west) -- (v2r2);
\draw[nosupport] (v2r3.west) -- (v3r2);

\draw[->] (v3r3.west) -- (v1r2);
\draw[->] (v3r3.west) -- (v2r2);
\draw[nosupport] (v3r3.west) -- (v3r2);

\draw[->] (v0r4.west) -- (v0r3);
\draw[->] (v0r4.west) -- (v1r3);
\draw[->] (v0r4.west) -- (v2r3);

\draw[->] (v1r4.west) -- (v0r3);
\draw[->] (v1r4.west) -- (v1r3);
\draw[->] (v1r4.west) -- (v2r3);

\draw[->] (v2r4.west) -- (v0r3);
\draw[->] (v2r4.west) -- (v1r3);
\draw[->] (v2r4.west) -- (v2r3);

\draw[->] (v3r4.west) -- (v1r3);
\draw[nosupport] (v3r4.west) -- (v2r3);
\draw[nosupport] (v3r4.west) -- (v3r3);

\draw[->] (v0r5.west) -- (v0r4);
\draw[->] (v1r5.west) -- (v1r4);
\draw[->] (v0r5.west) -- (v2r4);

\draw[->] (v1r5.west) -- (v0r4);
\draw[->] (v1r5.west) -- (v1r4);
\draw[->] (v1r5.west) -- (v2r4);

\draw[->] (v2r5.west) -- (v1r4);
\draw[->] (v2r5.west) -- (v2r4);
\draw[nosupport] (v2r5.west) -- (v3r4);

\draw[->] (v3r5.west) -- (v1r4);
\draw[->] (v3r5.west) -- (v2r4);
\draw[->] (v3r5.west) -- (v3r4);


\node[draw, dashed, rounded corners, fit=(v0r) (v3r) (v0r4) (v3r4), inner sep=8pt, label={[align=center]below:$w$}] {};

\end{tikzpicture}
\caption{Example of the indirect rule marking $L1_a$ as \texttt{to-skip} in the asynchronous mode ($w=5$). While $L1_a$ does not have the required $2f+1$ votes from blocks in round $r+3$, it also does not have $2f+1$ blames required to directly skip the block. Since the anchor $L2_a$ does not observe a certificate link between $L2_a$ and $L1_a$, $L1_a$ is skipped.}
\label{fig:mahi_mahi_indirect_skip}
\end{figure}
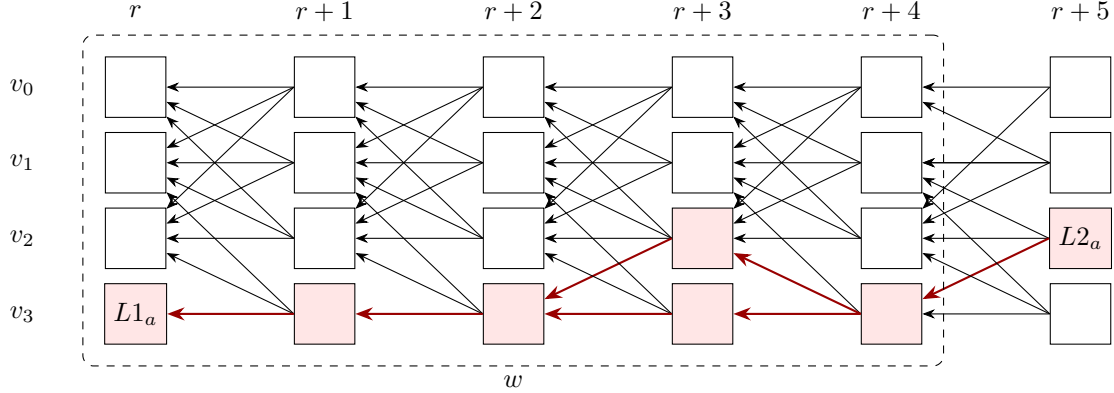
\quad Figure~\ref{fig:mahi_mahi_indirect_commit} and Figure~\ref{fig:mahi_mahi_indirect_skip} are examples of both outcomes of the indirect rules during \tarpon's asynchronous mode (while examples of the indirect rules during \tarpon's partially-synchronous mode can be found in Appendix~\ref{tarpon_indirect_rule}). This is where the number of proposer slots can cause slow-downs during periods of asynchrony (or when Byzantine validators equivocate). Since it increases the probability of the direct rule failing, \texttt{undecided} slots may accumulate. This will cause the protocol to utilise the indirect rule more often, subsequentially delaying the committent of proposed blocks. For example, in Figure~\ref{fig:mahi_mahi_indirect_commit}, assume a scenario where $L1_a$ has yet to be marked as \texttt{to-commit} via the indirect rule, all subsequent slots must wait for $L1_a$ to become decided before they can also become either skipped or committed. Thus, all leader slots from $L1_a$ upward remain \texttt{undecided}, which stalls the commit sequence and impacts performance \cite{MYSTICETI}. 
\subsection{Sequence} 
\label{subsection:sequence}
\quad Once all the slots in the DAG have been processed, all validators will extract, deterministically, the sequence of the proposers blocks. This newly extracted sequence is ordered (e.g., while the decision rules begins with the latest slot, the derived sequence begins with the oldest slot). The validators then iterate over this sequence, committing slots marked as \texttt{to-commit} and ignoring slots marked as \texttt{to-skip}. Such process will continue until the first \texttt{undecided} slots is encountered. To illustrate this, imagine a sequence of decided leaders $ [\textcolor{green!60!black}{L_{1a}},\textcolor{gray}{L_{2a}},\textcolor{red}{L_{3a}},\textcolor{green!60!black}{L_{4a}},\textcolor{red}{L_{5a}},\textcolor{green!60!black}{L_{6a}}]$ (obtained after running the decision rule on a DAG), where green, red, and grey represent leader slots marked as \texttt{to-commit}, \texttt{to-skip}, and \texttt{undecided}, respectively. The leader slots marked as \texttt{to-skip} are removed (thus we have the following sequence:  $ [\textcolor{green!60!black}{L_{1a}},\textcolor{gray}{L_{2a}},\textcolor{green!60!black}{L_{4a}},\textcolor{green!60!black}{L_{6a}}]$). However, since $L_{2a}$ is marked as \texttt{undecided}, \tarpon's final sequence is $[\textcolor{green!60!black}{L_{1a}}]$. Once the sequence is finalised, the validator will linearize the subDAG of each committed block, by carrying out a depth-first search, and ignore already included blocks. Since the validator processes leader blocks deterministically and subDAGs sequentially, the final commit sequence maintains a deterministic ordering (proven in Chapter~\ref{chap:Proofs}).
\section{Dual-Mode}
\quad We now introduce the dual-mode component of \tarpon. As shown in Chapter~\ref{chap:LitReview}, dual-mode protocols often introduce complexities to achieve dual-mode functionality, by either having to maintain two differing protocols or by increasing the complexity of the DAG construction \cite{DITTO,BULLSHARK}. With \tarpon, we introduce the following solution: instead of introducing complex mechanisms that change the structure of the DAG, the protocol will, instead, switch to its asynchronous mode every $K$ rounds, for a single wave. Intuitively, this means that in periods of asynchrony, eventually, a leader of one of these $K$-th rounds will be committed, committing its causal history alongside it. \tarpon can, therefore, achieve latency levels closer to those of partially-synchronous protocols, while inheriting the liveness robustness of asynchronous ones \cite{MYSTICETI, MAHIMAHI}.

\quad Similarly to the \textsf{Mysticeti}\xspace and \textsf{Mahi-Mahi}\xspace protocols, \tarpon can be configured to run with a pipeline in mind, and to allow for multiple proposers per round. This means that each round will act as the start of a new wave, and that each round can have multiple leaders $l$, where $l \leq n$. Figure~\ref{fig:dualmode_wave_structure} shows how the asynchronous mode and partially synchronous mode intertwine together. Of note, we can see that, irregardless of whether the pipeline mode is active, some decision rounds will overlap (e.g. given Figure~\ref{fig:dualmode_wave_structure}, round $r+5$ is the decision round for both propose rounds $r+1$ and $r+2$). Furthermore, some partially-asynchronous rounds may reach the decision round before a decision round for the previous asynchronous round has been reached (e.g. given Figure~\ref{fig:dualmode_wave_structure}, round $r+4$ is the decision round for round $r+2$, while round $r+1$ still has not reached its designated decision round $r+5$). Intuitively, \textit{safety} is preserved due to how leaders are sequenced. A leader sequence will return upon reaching the first undecided leader, thus, a leader block proposed in round $r+2$ will have to wait for leader blocks in round $r+1$ to be decided first. 
\begin{figure}[H]
\centering
\begin{tikzpicture}[
    node distance=1.5cm and 2.5cm,
    validator/.style={draw, minimum size=0.8cm},
    propose/.style={validator, thick},
    boost/.style={validator, draw=black},
    vote/.style={validator, draw=orange, thick},
    certify/.style={validator, draw=green!60!black, thick},
    supporting_validator/.style={validator, fill=green!20},
    support/.style={->, green!60!black, thick},
    >=Stealth,
    leader/.style={validator, fill=green!20},
    vote_support/.style={->, orange, thick},
    skipping_validator/.style={validator, fill=pink!40},
    cert_support/.style={->, green!60!black, thick},
    >=Stealth
]
\footnotesize

\node at (0, 4.8) {$r$};
\node at (2.5, 4.8) {$r+1$};
\node at (5, 4.8) {$r+2$};
\node at (7.5, 4.8) {$r+3$};
\node at (10, 4.8) {$r+4$};
\node at (12.5, 4.8) {$r+5$};

\node at (-1.5, 4) {$v_0$};
\node at (-1.5, 3) {$v_1$};
\node at (-1.5, 2) {$v_2$};
\node at (-1.5, 1) {$v_3$};

\node[boost] (v0r.west) at (0, 4) {$L1_{a}$};
\node[boost] (v1r) at (0, 3) {};
\node[boost] (v2r) at (0, 2) {$L1_b$};
\node[boost] (v3r) at (0, 1) {};

\node[boost] (v0r1) at (2.5, 4) {};
\node[boost] (v1r1) at (2.5, 3) {$L2_{b}$};
\node[boost] (v2r1) at (2.5, 2) {$L2_{a}$};
\node[boost] (v3r1) at (2.5, 1) {};

\node[boost] (v0r2) at (5, 4) {$L3_{b}$};
\node[boost] (v1r2) at (5, 3) {};
\node[boost] (v2r2) at (5, 2) {$L3_{a}$};
\node[boost] (v3r2) at (5, 1) {};

\node[boost] (v0r3) at (7.5, 4) {};
\node[boost] (v1r3) at (7.5, 3) {$L4_{b}$};
\node[boost] (v2r3) at (7.5, 2) {};
\node[boost] (v3r3) at (7.5, 1) {$L4_{a}$};

\node[boost] (v0r4) at (10, 4) {$L5_{a}$};
\node[boost] (v1r4) at (10, 3) {};
\node[boost] (v2r4) at (10, 2) {$L5_{b}$};
\node[boost] (v3r4) at (10, 1) {};

\node[boost] (v0r5) at (12.5, 4) {};
\node[boost] (v1r5) at (12.5, 3) {$L6_{a}$};
\node[boost] (v2r5) at (12.5, 2) {};
\node[boost] (v3r5) at (12.5, 1) {$L6_{b}$};

\draw[->] (v0r1.west) -- (v0r);
\draw[->] (v0r1.west) -- (v1r);
\draw[->] (v0r1.west) -- (v2r);

\draw[->] (v1r1.west) -- (v0r);
\draw[->] (v1r1.west) -- (v1r);
\draw[->] (v1r1.west) -- (v2r);

\draw[->] (v2r1.west) -- (v0r);
\draw[->] (v2r1.west) -- (v1r);
\draw[->] (v2r1.west) -- (v2r);

\draw[->] (v3r1.west) -- (v1r);
\draw[->] (v3r1.west) -- (v2r);
\draw[->] (v3r1.west) -- (v3r);

\draw[->] (v0r2.west) -- (v0r1);
\draw[->] (v0r2.west) -- (v1r1);
\draw[->] (v0r2.west) -- (v2r1);

\draw[->] (v1r2.west) -- (v0r1);
\draw[->] (v1r2.west) -- (v1r1);
\draw[->] (v1r2.west) -- (v2r1);

\draw[->] (v2r2.west) -- (v0r1);
\draw[->] (v2r2.west) -- (v1r1);
\draw[->] (v2r2.west) -- (v2r1);

\draw[->] (v3r2.west) -- (v1r1);
\draw[->] (v3r2.west) -- (v2r1);
\draw[->] (v3r2.west) -- (v3r1);

\draw[->] (v0r3.west) -- (v0r2);
\draw[->] (v0r3.west) -- (v1r2);
\draw[->] (v0r3.west) -- (v2r2);

\draw[->] (v1r3.west) -- (v0r2);
\draw[->] (v1r3.west) -- (v1r2);
\draw[->] (v1r3.west) -- (v2r2);

\draw[->] (v2r3.west) -- (v0r2);
\draw[->] (v2r3.west) -- (v1r2);
\draw[->] (v2r3.west) -- (v2r2);

\draw[->] (v3r3.west) -- (v1r2);
\draw[->] (v3r3.west) -- (v2r2);
\draw[->] (v3r3.west) -- (v3r2);

\draw[->] (v0r4.west) -- (v0r3);
\draw[->] (v0r4.west) -- (v1r3);
\draw[->] (v0r4.west) -- (v2r3);

\draw[->] (v1r4.west) -- (v0r3);
\draw[->] (v1r4.west) -- (v1r3);
\draw[->] (v1r4.west) -- (v2r3);

\draw[->] (v2r4.west) -- (v0r3);
\draw[->] (v2r4.west) -- (v1r3);
\draw[->] (v2r4.west) -- (v2r3);

\draw[->] (v3r4.west) -- (v1r3);
\draw[->] (v3r4.west) -- (v2r3);
\draw[->] (v3r4.west) -- (v3r3);

\draw[->] (v0r5.west) -- (v0r4);
\draw[->] (v0r5.west) -- (v1r4);
\draw[->] (v0r5.west) -- (v2r4);

\draw[->] (v1r5.west) -- (v0r4);
\draw[->] (v1r5.west) -- (v1r4);
\draw[->] (v1r5.west) -- (v2r4);

\draw[->] (v2r5.west) -- (v0r4);
\draw[->] (v2r5.west) -- (v1r4);
\draw[->] (v2r5.west) -- (v2r4);

\draw[->] (v3r5.west) -- (v1r4);
\draw[->] (v3r5.west) -- (v2r4);
\draw[->] (v3r5.west) -- (v3r4);

\node[anchor=east] at (-0.1,0.4-0) {$\scriptstyle w$};
\node[anchor=east] at (2.4,0.4-0.2) {$\color{blue}\scriptstyle w+1$};
\node[anchor=east] at (4.9,0.4-0.4) {$\scriptstyle w+2$};
\node[anchor=east] at (7.4,0.4-0.6) {$\scriptstyle w+3$};
\node[anchor=east] at (9.9,0.4-0.8) {$\scriptstyle w+4$};
\node[anchor=east] at (11.9,0.4-1) {$\scriptstyle w+5$};
\draw[black, line width=1pt] (0,0.4-0) -- (5,0.4-0);
\draw[blue, line width=1pt] (2.5,0.4-0.2) -- (12.5,0.4-0.2);
\draw[black, line width=1pt] (5,0.4-0.4) -- (10,0.4-0.4);
\draw[black, line width=1pt] (7.5,0.4-0.6) -- (12.5,0.4-0.6);
\draw[black, line width=1pt] (10,0.4-0.8) -- (12.5,0.4-0.8);
\draw[black, line width=1pt, dotted] (12.5,0.4-0.8) -- (13,0.4-0.8);
\draw[black, line width=1pt] (12,0.4-1) -- (12.5,0.4-1);
\draw[black, line width=1pt, dotted] (12.5,0.4-1) -- (13,0.4-1);
\end{tikzpicture}
\caption{Wave patterns of \tarpon's dual-mode protocol ($2$ leaders per round, pipe-lining enabled). Round $r+1$ represents the asynchronous mode of \tarpon with a wave length $w$ set to $5$. The remainder of the rounds are run through the partially-synchronous mode. Furthermore, we can see that, since $L2_a$ does not have enough votes to be directly committed, the anchor used to determine the status of the slot must come from the first round \textbf{after} wave $w+1$ has finished, thus $r+6$.  }
\label{fig:dualmode_wave_structure}
\end{figure}
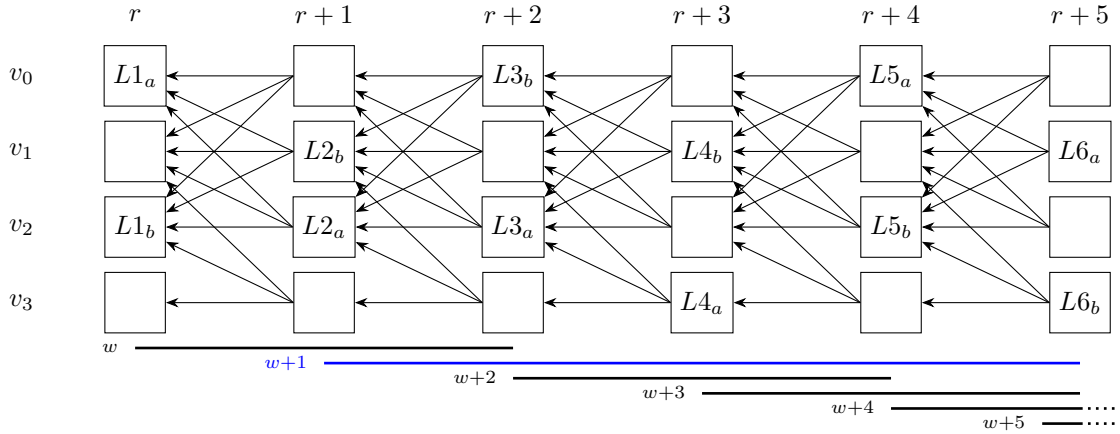
\subsection{Dynamic Scheduling}
\label{subsection:dynamic_scheduling}
\quad As we have seen, periodically switching to Mahi-mahi will grant us the liveness property even under asynchronous settings; however, such approach can be further improved. For example, imagine a scenario where $GST$ is perennially achieved, meaning that \tarpon's network connectivity is working as intended. In such scenario, it would be ideal to run the asynchronous round less often, while in the opposite scenario, one where connectivity is causing the direct-rule to reach an \texttt{undecided} conclusion often, we would want to run the asynchronous run more often. Thus, how do we allow each validator to make a deterministic decision to change the interval rate of \tarpon's asynchronous mode? 

\quad \textsf{Hammerhead}\xspace is a consensus protocol, built over the partially-synchronous version of \textsf{Bullshark}\xspace, which presents a suitable approach to what we are trying to achieve \cite{BULLSHARK,HAMMERHEAD}. In the case of \textsf{Hammerhead}\xspace, Tsimos et. al \cite{HAMMERHEAD} present a reputation-based leader-election mechanism for DAG-based consensus protocols. There are cases where leaders might crash, thus, if a validator is selected as leader for a round, in such round no progress can be achieved. By creating a reputation-based leader-election mechanism, leaders who are consistently voted for will be selected more often, decreasing the possibility of electing a faulty leader. They utilise the DAG itself to determine how many times a leader is voted for because the causal history of a committed leader will be identical for everyone. As we have also said, once a leader block is committed by a validator, all other honest validators will reach the same conclusion. This allows each validator to make deterministic choices by observing the subDAG (i.e the causal history of a block) of a committed leader block.

\quad We can follow a similar approach to create a dynamic-frequency mechanism to optimise the performance of \tarpon. While \textsf{Hammerhead}\xspace utilises the subDAGs of each committed leader to determine the reputation of all leaders, we simply require the subDAG of a singular leader block to determine how we should update the schedule. The dynamic schedule works as follows:
\begin{itemize}
    \item Assume that every $K$th leader has to be decided through the asynchronous rules.
    \item Once leader $K$ is marked as \texttt{to-commit}, we analyse its subDAG  up to but not including the previous $K$th leader (i.e. we analyse the subDAG until the round after the the last asynchronous committed leader).
    \item Within this same subDAG, we determine how many of the leaders in \tarpon's committed leader sequence were direct commits, that is, leader proposals that have been committed via the direct rule (with the rules specified in Section~\ref{subsection:decision_rules}). Given the ratio of observed direct commits versus the amount of expected commits, we can change the scheduling accordingly. The expected commits include both leader slots marked as \texttt{to-commit} and \texttt{to-skip}, as skipped leaders are also a sign of asynchrony.
\end{itemize}
\begin{figure}[H]
\centering
\begin{tikzpicture}[
    node distance=1.5cm and 2.5cm,
    validator/.style={draw, minimum size=0.8cm},
    propose/.style={validator, thick},
    boost/.style={validator, draw=black},
    vote/.style={validator, draw=orange, thick},
    certify/.style={validator, draw=green!60!black, thick},
    supporting_validator/.style={validator, fill=green!20},
    support/.style={->, green!60!black, thick},
    >=Stealth,
    leader/.style={validator, fill=green!20},
    vote_support/.style={->, orange, thick},
    skipping_validator/.style={validator, fill=pink!40},
    cert_support/.style={->, green!60!black, thick},
    transparent/.style={opacity=0.3},
    transparent_edge/.style={->, opacity=0.3, gray},
    highlighted/.style={fill=yellow!30},
    anchor/.style={fill=blue!30, draw=blue, ultra thick},
    subdag_edge/.style={->, black},
    >=Stealth
]
\footnotesize
\node at (0, 4.8) {$r$};
\node at (2.5, 4.8) {$r+1$};
\node at (5, 4.8) {$r+2$};
\node at (7.5, 4.8) {$r+3$};
\node at (10, 4.8) {$r+4$};

\node at (-1.5, 4) {$v_0$};
\node at (-1.5, 3) {$v_1$};
\node at (-1.5, 2) {$v_2$};
\node at (-1.5, 1) {$v_3$};
\node[boost, highlighted] (v0r) at (0, 4) {$L1_a$};
\node[boost, highlighted] (v1r) at (0, 3) {};
\node[boost, highlighted] (v2r) at (0, 2) {$L1_b$};
\node[boost, highlighted] (v3r) at (0, 1) {};

\node[supporting_validator, highlighted] (v0r1) at (2.5, 4) {};
\node[supporting_validator, highlighted] (v1r1) at (2.5, 3) {$L2_b$};
\node[boost, highlighted] (v2r1) at (2.5, 2) {};
\node[boost, highlighted] (v3r1) at (2.5, 1) {$L2_a$};

\node[supporting_validator, highlighted] (v0r2) at (5, 4) {$L3_b$};
\node[boost, highlighted] (v1r2) at (5, 3) {};
\node[supporting_validator, highlighted] (v2r2) at (5, 2) {$L3_a$};
\node[boost, transparent] (v3r2) at (5, 1) {};

\node[boost, highlighted] (v0r3) at (7.5, 4) {};
\node[supporting_validator, highlighted] (v1r3) at (7.5, 3) {$L4_a$};
\node[supporting_validator, highlighted] (v2r3) at (7.5, 2) {};
\node[boost, transparent] (v3r3) at (7.5, 1) {$L4_b$};

\node[supporting_validator] (v0r4) at (10, 4) {$L5_a$};
\node[boost, transparent] (v1r4) at (10, 3) {};
\node[boost, transparent] (v2r4) at (10, 2) {$L5_b$};
\node[boost, transparent] (v3r4) at (10, 1) {};

\draw[subdag_edge] (v0r1.west) -- (v0r);
\draw[subdag_edge] (v0r1.west) -- (v1r);
\draw[subdag_edge] (v0r1.west) -- (v2r);

\draw[subdag_edge] (v1r1.west) -- (v0r);
\draw[subdag_edge] (v1r1.west) -- (v1r);
\draw[subdag_edge] (v1r1.west) -- (v2r);

\draw[subdag_edge] (v3r1.west) -- (v1r);
\draw[subdag_edge] (v3r1.west) -- (v2r);
\draw[subdag_edge] (v3r1.west) -- (v3r);

\draw[subdag_edge] (v2r1.west) -- (v0r);
\draw[subdag_edge] (v2r1.west) -- (v1r);
\draw[subdag_edge] (v2r1.west) -- (v2r);

\draw[subdag_edge] (v0r2.west) -- (v0r1);
\draw[subdag_edge] (v0r2.west) -- (v1r1);
\draw[subdag_edge] (v0r2.west) -- (v2r1);

\draw[subdag_edge] (v2r2.west) -- (v0r1);
\draw[subdag_edge] (v2r2.west) -- (v1r1);
\draw[subdag_edge] (v2r2.west) -- (v3r1);

\draw[transparent_edge] (v3r2.west) -- (v0r1);
\draw[transparent_edge] (v3r2.west) -- (v1r1);
\draw[transparent_edge] (v3r2.west) -- (v3r1);

\draw[subdag_edge] (v1r2.west) -- (v1r1);
\draw[subdag_edge] (v1r2.west) -- (v2r1);
\draw[subdag_edge] (v1r2.west) -- (v3r1);

\draw[subdag_edge] (v1r3.west) -- (v0r2);
\draw[subdag_edge] (v1r3.west) -- (v1r2);
\draw[subdag_edge] (v1r3.west) -- (v2r2);

\draw[subdag_edge] (v2r3.west) -- (v0r2);
\draw[subdag_edge] (v2r3.west) -- (v2r2);

\draw[transparent_edge] (v3r3.west) -- (v0r2);
\draw[transparent_edge] (v3r3.west) -- (v2r2);
\draw[transparent_edge] (v3r3.west) -- (v3r2);

\draw[subdag_edge] (v0r3.west) -- (v0r2);
\draw[subdag_edge] (v0r3.west) -- (v1r2);
\draw[subdag_edge] (v0r3.west) -- (v2r2);

\draw[subdag_edge] (v0r4.west) -- (v0r3);
\draw[subdag_edge] (v0r4.west) -- (v1r3);
\draw[subdag_edge] (v0r4.west) -- (v2r3);

\draw[transparent_edge] (v1r4.west) -- (v0r3);
\draw[transparent_edge] (v1r4.west) -- (v1r3);
\draw[transparent_edge] (v1r4.west) -- (v2r3);
\draw[transparent_edge] (v2r4.west) -- (v0r3);
\draw[transparent_edge] (v2r4.west) -- (v1r3);
\draw[transparent_edge] (v2r4.west) -- (v2r3);
\draw[transparent_edge] (v3r4.west) -- (v1r3);
\draw[transparent_edge] (v3r4.west) -- (v2r3);
\draw[transparent_edge] (v3r4.west) -- (v3r3);
\end{tikzpicture}
\caption{Example of subDAG analysis of a committed anchor. Assume that $L5_a$ is the $K$th leader that will trigger a schedule update. Leaders from round $r+4$, $r+3$ and $r+2$ will not be considered for the analysis, as the subDAG does not give enough information to determine whether these proposals have been directly committed or not. However, we can see how leader $L2_a$, while having enough votes to be directly committed, was not marked as \texttt{to-commit} directly by $L5_a$, as $L5_a$ causal history does not include the additional necessary vote provided by $L3_b$.}
\label{fig:dualmode_subdag_analysis}
\end{figure}
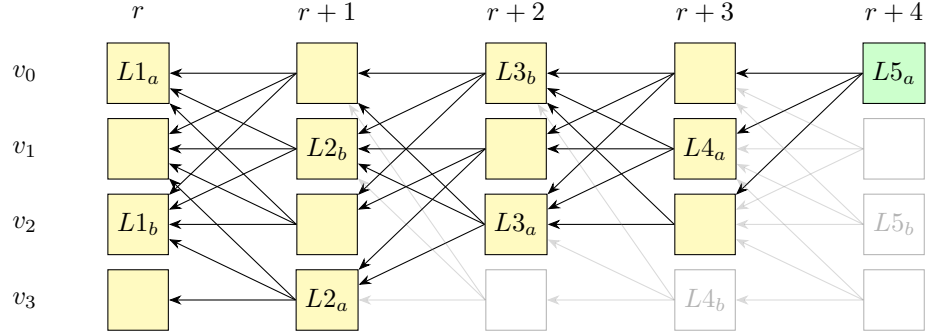
\subsection{Sequence (Extended)}
\label{sequenceextended}
\quad There is one specific scenario that can endanger the \textit{safety} property of \tarpon. Specifically, imagine a scenario where we have a validator $v_i$ who is experiencing high-levels of network turbulence between the other validators. In such scenario, $v_i$ is unable to reach either \texttt{to-skip} or \texttt{to-commit} decisions for leader blocks, and is left with a large number of \texttt{undecided} slots. Assume that we have enough \texttt{undecided} slots to encompass two asynchronous waves ($L$ and $L2$). Once $v_i$ regains an appropriate level of connectivity, they will begin working through all their \texttt{undecided} slots. Since decision rules are made top to bottom, $v_i$ will decide $L2$ before $L1$. However, $L2$'s schedule may change once $L1$ is processed, and if this happens, then $v_i$ will have a different schedule amongst all other validators. To stop this from happening, we need to make sure that $v_i$ can retroactively change the schedule, and reanalyse the remaining leaders under the new schedule. Thus, the flow is such that, once a validator runs the decision rule and obtains a list of leader slots marked as \texttt{to-commit}, we also cut-off the sequence up to, and including, the first asynchronous leader slot marked as \texttt{to-commit}. Lastly, the algorithms for \tarpon can be found in Appendix~\ref{Appendix:Algorithms}.
\chapter{Implementation}
\label{chap:implementation}
\quad \tarpon was developed as a network, multi-core dual-mode protocol validator that has been implemented in Rust \cite{RUST}. The protocol builds upon a forked codebase of the \textsf{Mysticeti}\xspace protocol, which comprises of approximately $14,000$ lines of code. \tarpon has both been open-sourced\footnote{\url{https://github.com/zenodea/dualmode_mysticeti}}, and provided alongside the submission of this project. \tarpon leverages \texttt{tokio}, which allows for asynchronous networking operations, and communications are established through raw TCP sockets \cite{TOKIO}. The protocol utilizes the following cryptographic components for security: \texttt{ed25519-consensus} for asymmetric cryptographic operations (enabling validators to uniquely sign blocks and allowing other validators to verify these signatures), and \texttt{blake2} for cryptographic hashing functions (used to create reference hashes of blocks) \cite{ED25519,BLAKE}. The architecture follows a hybrid threading model where network communications operate asynchronously while the core dual-mode consensus logic executes synchronously within a dedicated thread. This design choice enhances testing reliability, reduces race condition risks, and enables precise profiling of the critical consensus pathway.

\quad In this section, we present a detailed examination of the core implementation components that directly correspond to the theoretical algorithms presented in the referenced work. \tarpon's implementation demonstrates the practical realization of the dual-mode protocol's key mechanisms, providing concrete evidence of the feasibility and performance characteristics of the proposed approach. We structure this analysis by focusing on two fundamental aspects of \tarpon that represent the most significant contributions of our work: the decision rule and the dual-mode scheduler. 
\section{Decision Rule}
\quad The implementation of \tarpon works as follows: whenever a validator receives a block, they run an idempotent function (handled by the \texttt{core} module) that tries to commit as many blocks that have not been committed yet (i.e, marked as \texttt{undecided}). We have a universal committer (representing Algorithm~\ref{alg:main}) that controls a number of base\_committers. The universal committer also handles the pipeline, and the multiple leaders per round aspect (where each extra leader represents a base\_committer instance). We will now examine the main changes made to these structures to accommodate for \tarpon's dual-mode property.
\subsection{Universal Committer} 
\quad The UniversalCommitter module manages multiple BaseCommitters with different \texttt{leader\_offset} and \texttt{round\_offset} configurations to support multi-leader consensus and pipelining. The UniversalCommitter, similarly to \textsf{Mysticeti}\xspace and \textsf{Mahi-Mahi}\xspace, calls the \texttt{try\_commit} function (represented by the \textsc{TryDecide} process in Algorithm~\ref{alg:main}). The \texttt{try\_commit} function takes the highest leader proposal round and highest committed leader proposal round, and tries to decide as many as possible within this range. However, now we also pass an \texttt{AsyncState} structure, explained more in-depth in Section~\ref{sec:impl_dual_mode}, that determines whether the round in question should be interpreted via the partially-synchronous or asynchronous mode, just like the \textsc{TryDecide} process in Algorithm~\ref{alg:main}. The complete changes made to the module can be observed in Appendix~\ref{Appendix:UniversalCommitter}. 
\subsection{Base Committer} 
\quad The base committer contains the core commit logic for the consensus protocol, providing methods like \texttt{try\_direct\_decide} and \texttt{try\_indirect\_decide} that can be called to determine whether a leader block can be committed or should be skipped (representing the entirety of Algorithm~\ref{alg:decider}). The base committer also handles the pipeline and multi-proposer settings of \tarpon. The primary change made to the BaseCommitter module is the calculations done when interpreting the DAG. Specifically, the voting and decision rounds change depending on the current mode. As the decision round is used to calculate the voting round, changing the round wave length consequently places the voting round correctly. The calculation of the decision round is now $r+w-1$, where $w$ is chosen based on the current \texttt{AsyncState} (either $w=3$ or $w=4/w=5$). The second major change resides in the indirect rule, specifically, we require the anchor's round $r'$ to take into consideration which mode of operation the leader proposal, which requires the anchor, utilised. Similarly as before, the anchor round is now calculated as $r'\geq r+w$, where $w$ is chosen based on the current \texttt{AsyncState}. Further changes can be seen in Appendix~\ref{Appendix:BaseCommitter}.
\newpage
\quad Figure~\ref{fig:flowchartConsensusProtocol} represents the modified workflow of \tarpon's decider instance. Specifically of note is the extra argument beyond the anchor $v$. $\mathcal{A}$ represents the current asynchronous state, storing information which allows the decider to decide whether anchor $v$ should be decided via \tarpon's partially-synchronous or asynchronous mode. The \texttt{try\_commit} function of \tarpon's decider will run for all undecided anchors from the highest committed slot, thus, the process seen in Figure~\ref{fig:flowchartConsensusProtocol} is repeated as many times as \texttt{undecided} slots there are.
\vspace{2em}
\begin{figure}[H]
\centering
\makebox[\textwidth][c]{
\begin{tikzpicture}[scale=0.8, transform shape, node distance=0.6cm and 1.8cm]
\tikzstyle{startstop} = [rectangle, draw=black, minimum width=1.8cm, minimum height=0.6cm, text centered, font=\scriptsize, fill=green!20]
\tikzstyle{process} = [rectangle, draw=black, minimum width=1.6cm, minimum height=0.6cm, text centered, font=\scriptsize]
\tikzstyle{decision} = [diamond, draw=orange, thick, minimum width=1.2cm, minimum height=0.6cm, text centered, font=\scriptsize, inner sep=1pt]
\tikzstyle{decisionsmall} = [diamond, aspect=2, draw=orange, thick, minimum width=1.2cm, minimum height=0.3cm, text centered, font=\scriptsize, inner sep=1pt]
\footnotesize
\tikzstyle{arrow} = [->, thick]


\node (start) [startstop, align=center] {\texttt{try\_decide}($v$,$\mathcal{A}$)};
\node (loop) [process, above=of start, align=center] {$\forall v \in \text{\texttt{undecidedLeader}}$};
\node (trycommit) [startstop, above=of loop, align=center] {\texttt{try\_commit}($\mathcal{A}$)};
\node (asynccheck) [decision, right=of start, align=center] {$\mathcal{A}$.\texttt{isAsync}($v_{\text{round}}$)};
\node (asynctrue) [process, above right=0.4cm and 1.5cm of asynccheck, align=center] {\texttt{async = true }};
\node (asyncfalse) [process, below right=0.4cm and 1.5cm of asynccheck, align=center] {\texttt{async = false}};
\node (directdecisions) [decisionsmall, right=4cm of asynccheck, align=center] {\texttt{direct\_decide(}$v$\texttt{,async)}};
\node (indirectdecisions) [decisionsmall, right=1cm of directdecisions, align=center] {\texttt{\texttt{indirect\_decide(}$v$\texttt{,async)}}};
\node (return) [startstop, above=of directdecisions, align=center] {\texttt{Return} $v_\texttt{decision}$};
\node (returnindirect) [startstop, above=of indirectdecisions, align=center] {\texttt{Return} $v_\texttt{decision}$};
\node (failed) [startstop, below=of indirectdecisions, align=center] {$v_\texttt{decision} ==$ \texttt{undecided}};
\draw [arrow] (start) -- (asynccheck);
\draw [arrow] (trycommit) -- (loop);
\draw [arrow] (loop) -- (start);
\draw [arrow] (asynccheck) |- node[above, pos=0.5] {\scriptsize \texttt{true}} (asynctrue);
\draw [arrow] (asynccheck) |- node[below, pos=0.5] {\scriptsize \texttt{false}} (asyncfalse);
\draw [arrow] (asynctrue.south) |- (directdecisions.west);
\draw [arrow] (asyncfalse.north) |- (directdecisions.west);
\draw [arrow] (directdecisions) -- node[right, pos=0.1] {\scriptsize \texttt{decided}} (return);
\draw [arrow] (directdecisions) -- node[above, pos=0.1] {\small \texttt{$\perp$}} (indirectdecisions);
\draw [arrow] (indirectdecisions) -- node[right, pos=0.1] {\scriptsize \texttt{decided}} (returnindirect);
\draw [arrow] (indirectdecisions) -- node[right, pos=0.1] {\small \texttt{$\perp$}} (failed);
\end{tikzpicture}}
    \caption{Flowchart decision rule followed by \tarpon. $v$ represents the anchor that is being decided on, while $\mathcal{A}$ represents the current asynchronous state of the system. We run the \texttt{try\_decide(}$v,\mathcal{A}$\texttt{)} function for each available undecided leader slot, from the highest proposer slot, to the lowest, undecided, proposer slot.}
    \label{fig:flowchartConsensusProtocol}
\end{figure}
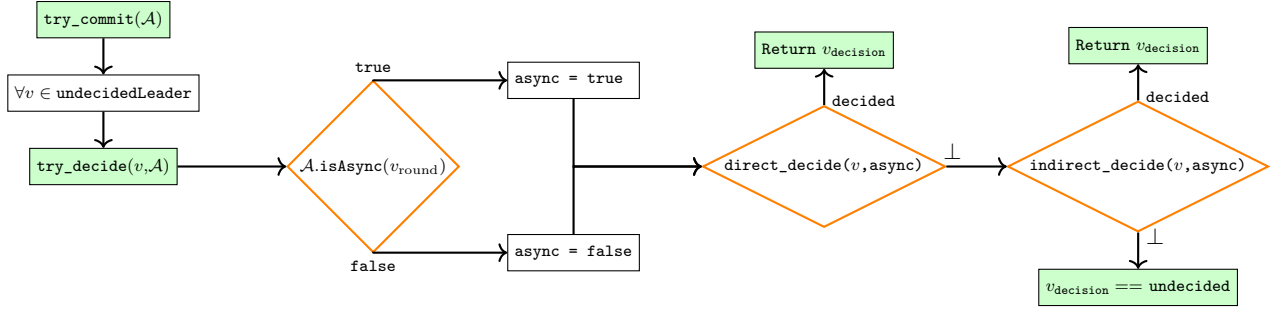
\section{Dual-Mode Scheduler}
\label{sec:impl_dual_mode}
\quad The \texttt{DualModeScheduler} structure handles all the logic required to update the asynchronous mode frequency accordingly. The structure incorporates the following key components:
\begin{itemize}
    \item \texttt{AsyncState:} The \texttt{AsyncState} is a structure that holds the information necessary for the committer to interpret the DAG correctly. It holds the previously committed asynchronous mode round, alongside the current frequency that the asynchronous mode is run at. This structure is given to the \texttt{UniversalCommitter} to run the correct decision rules depending on the round. The implementation of this structure can be viewed in Appendix~\ref{appendix:AsyncState}.
    \item \textbf{DualModeConfig}: The \texttt{DualModeConfig} structure defines configuration parameters for the dual-mode consensus scheduler, including: wave length, leader count, and adaptive interval bounds that control the balance between synchronous and asynchronous consensus modes. The modularity of this structure also simplifies testing matters, as we can easily create new configurations for specific scenarios. The implementation of this structure can be viewed in Appendix~\ref{appendix:DualModeConfig}.
    \item Furthermore, \tarpon's \texttt{DualModeScheduler} holds the sequence of leaders marked as \texttt{to-commit} and a \texttt{sequence\_length} which holds the total amount of decided leaders. When we run the \texttt{try\_commit} function we also store the total amount of decided leaders (both leader slots marked as either \texttt{to-commit} or \texttt{to-skip}). Then we send both values to the \texttt{DualModeScheduler} to be utilised whenever we trigger a schedule update (i.e., once we mark an asynchronous leader as \texttt{to-commit}). 
\end{itemize}
\quad With these structures, \tarpon's \texttt{DualModeScheduler} can handle all the logic required to perform safe schedule updates. Once we commit an asynchronous leader, we call the \texttt{update\_dualmode\_schedule} which pops the first leader from the sequence of decided leaders stored in the scheduler (which will always be an asynchronous leader when the function is called). Appendix~\ref{appendix:DualModeScheduler} has an in-depth explanation of each of the functions used by the \texttt{DualModeScheduler} structure, however, two functions merit particular focus as the most challenging to implement and the most significant for the research contributions presented herein: the subDAG collection and the direct commits counter. Both of these functions are run once \tarpon protocol triggers a schedule update.  
\subsection{Collect Subdag Function} 
\quad Once a schedule update is triggered, \tarpon begins by collecting the subDAG of the newest leader block, via the \texttt{collect\_subdag()} (see Listing~\ref{listing:collectsubdag} in Appendix~\ref{appendix:DualModeScheduler}). The function follows similarly to how the linearizer collects the blocks from the subDAG of newly committed leaders. However, unlike the linearizer, we utilise the anchor's causal history to collect all the blocks included until we reach the oldest committed leader in the sequence of leaders stored in the dual-mode scheduler. By doing so, we ensure that the collected subDAGs will be identical across all validators, guaranteeing they reach the same consensus decision. It is important to note that even if some blocks are directly committed by a majority of validators, if the selected validator only indirectly committed those blocks, this still counts as a "bad network signal." 
\subsection{Count Direct Commits Function} 
\quad Once we obtain the subDAG, we count the direct commits within said subDAG (see Listing~\ref{listing:countdirectcommits} and Listing~\ref{listing:isdirectlycommitted} in Appendix~\ref{appendix:DualModeScheduler}). While very similar to how \tarpon's base committer handles the direct rule, compared to them we utilise a vector of \texttt{blocks} obtained via the \texttt{collect\_subdag()} function. This is a crucial difference because it means we exclusively use the causal history of the anchor used for the schedule update. As stated earlier, this means that, while some other honest validators may have marked a proposed block as \texttt{to-commit} via the direct rule, we only care to see how the anchor in question ruled the specific leader block.
\\
\null \quad Once the number of direct commits has been collected, we utilise this number and compare it to the expected number of direct commits (i.e., all proposers slots should be directly committed), and modify the asynchronous interval via a deterministic function it (seen in Appendix~\ref{appendix:DualModeScheduler}). Lastly, just like we have seen in Figure~\ref{fig:dualmode_subdag_analysis}, we make sure that the number of expected commits includes only the amount of commits that can be selected given a subDAG. For example, imagine a sequence $\mathcal{L} =  [L_{1a},L_{2a},L_{3a},L_{4a},L_{5a}]$, where they have all been marked as \texttt{to-commit} and $L_{5a}$ is a leader that has been committed under the asynchronous mode. Furthermore, assume that pipeline mode is turned on. Given the subDAG of $L_{5a}$, proposals $[L_{2a},L_{3a},L_{4a}]$ slots cannot be determined with the given information, thus, the number of expected commits from $\mathcal{L}$ would be $1$, that is $L_{1a}$. This is equivalent to the grayed out blocks as displayed in Figure~\ref{fig:dualmode_subdag_analysis}, where $L5_b$ and $L4_b$ are not included in the expected commits due as the DAG does not provide enough information to make a decision.  
\section{Sequence Cutoff} 
\quad Once we finish running the decision rule, as explained in Section~\ref{subsection:sequence} and Section~\ref{sequenceextended}, \tarpon iterates through the ordered sequence of leaders, and tries to commit as many as possible. However, given the protocol's unique dynamic schedule, extra security measures are put in place to ensure its \textit{safety}. Once the \texttt{try\_commit()} function is run, we iterate through the leader sequence, committing as many blocks as possible. However, the iteration will halt in two scenarios: first if an \texttt{undecided} leader is encountered, and then if leaders who's both been marked as \texttt{to-commit}, and that has been committed with \tarpon's asynchronous mode. If both are true, the sequence is truncated up to and including that said leader, ensuring schedule changes are made consistently throughout all validators. The implementation of the sequence cutoff mechanism is further presented in Appendix~\ref{appendix:safetymechanism}.
\\
\null \quad Figure~\ref{fig:flowchart_leader} depicts the complete flow of \tarpon. Whenever a validator receives a block, the idempotent \texttt{try\_commit()} function is run trying to decide as many leaders as possible with the current \texttt{AsyncState}. The protocol, then, loops through the ordered sequence of leaders, returned by \texttt{try\_commit()}, and tries to commit as many as possible. At the same time, if either the leader is \texttt{undecided} or the most recent committed leader has been committed through the asynchronous mode, then we stop iterating through the leader sequence, and discard subsequent leaders, including the one that made the iteration stop. In the latter case, however, we also update the asynchronous interval before running \texttt{try\_commit} again.
\begin{figure}[H]
\centering
\begin{tikzpicture}[scale=0.9, transform shape, node distance=1cm and 4cm]
\tikzstyle{startstop} = [rectangle, draw=black, minimum width=1.8cm, minimum height=0.7cm, text centered, font=\scriptsize, fill=green!20]
\tikzstyle{process} = [rectangle, draw=black, minimum width=1.8cm, minimum height=0.7cm, text centered, font=\scriptsize]
\tikzstyle{decision} = [diamond, draw=orange, thick, minimum width=2.7cm, minimum height=0.7cm, text centered, font=\scriptsize, inner sep=1pt]
\tikzstyle{decisionsmall} = [diamond, draw=orange, thick, minimum width=1.4cm, minimum height=0.7cm, text centered, font=\scriptsize, inner sep=1pt]
\tikzstyle{arrow} = [->, thick]
\node (start) [startstop, align=center] {Validator\\Receives Block};
\node (trycommit) [process, right=of start, align=center] {Run\\\texttt{try\_commit()}};
\node (getleaders) [process, right=of trycommit] {Get \texttt{LeadersSequence}};

\node (nextleader) [process, below=of getleaders] {\texttt{Leaders Sequence.pop}};
\node (checkstatus) [decisionsmall, left=of nextleader, align=center] {\texttt{Leader.Status}};
\node (commit) [process, left=of checkstatus] {\texttt{Commit(Leader)}};
\node (final1) [startstop, below=of checkstatus] {\texttt{Return}};
\node (checkasync) [decision, below=of commit] {\texttt{is\_async(Leader.round)}};
\node (return2) [startstop, below=of checkasync, align=center] {\texttt{Return} \&\\\texttt{UpdateSchedule}};

\node (morecheck) [decision, below=of nextleader] {\texttt{LeadersSequence.empty}};
\node (end) [startstop, below=of morecheck] {\texttt{Return}};

\draw [arrow] (start) -- (trycommit);
\draw [arrow] (trycommit) -- (getleaders);
\draw [arrow] (getleaders) -- (nextleader);
\draw [arrow] (nextleader) -- (checkstatus);
\draw [arrow] (checkstatus) -- node[right, pos=0] {\scriptsize \texttt{undecided}} (final1);
\draw [arrow] (checkstatus) -- node[above,pos=0.2] {\scriptsize \texttt{decided}} (commit);
\draw [arrow] (commit) -- (checkasync);
\draw [arrow] (checkasync) -- node[left,pos=0] {\scriptsize \texttt{true}} (return2);
\draw [arrow] (checkasync) -- node[pos=0.05, above] {\scriptsize \texttt{false}} (morecheck);
\draw [arrow] (morecheck) -- node[right, pos=0] {\scriptsize \texttt{true}} (end);

\draw [arrow] (morecheck) |- node[pos=0, right] {\scriptsize \texttt{false}} ++(0,2) -| (nextleader);

\end{tikzpicture}
\caption{Flowchart of the \tarpon protocol. The function \texttt{try\_commit} represents the function found in Figure~\ref{fig:flowchartConsensusProtocol}. The remainder of the flowchart visualizes how \tarpon handles the processing of the returned sequence of leaders. Of note, if the \texttt{LeaderSequence} is prematurely cut short via the \texttt{is\_async} function, the unprocessed leaders remaining in the \texttt{LeaderSequence} are discarded.}
\label{fig:flowchart_leader}
\end{figure}
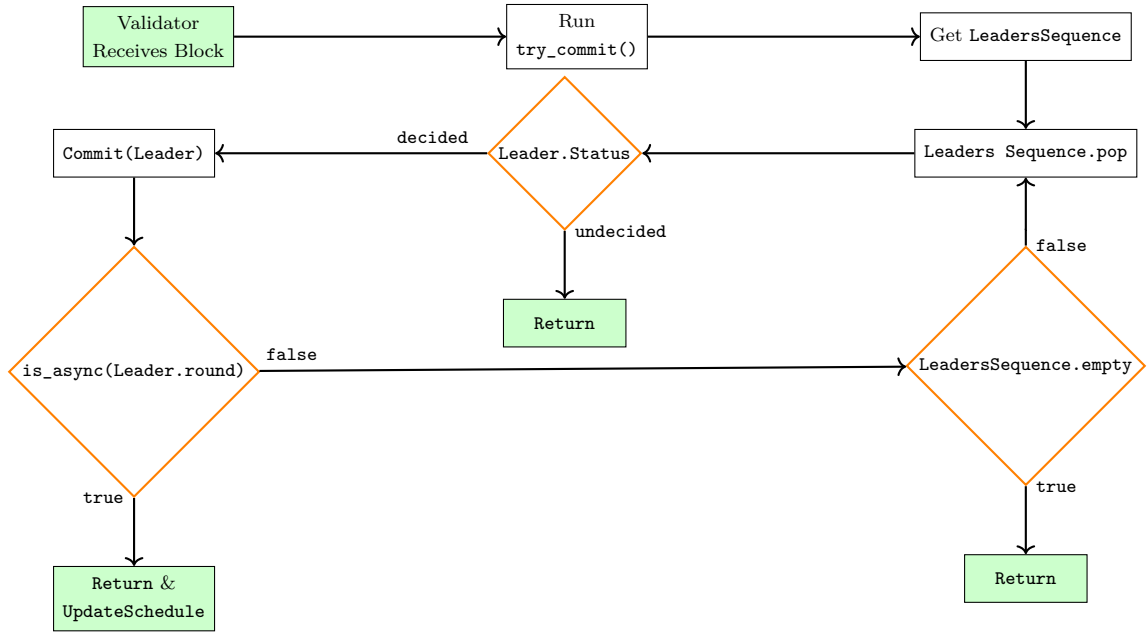
\section{Tests}
\label{section:tests}
\quad To confirm the correctness of \tarpon's implementation, rigorous testing has been performed on all components that relate to either the decision rule or the dual-mode scheduler. This section will give a brief overview of how tests were conducted, however, additional information is presented in Appendix~\ref{appendixL:testConsensusProtocol}.
\newpage
\subsection{Decision Rule} 
\quad First and foremost, \tarpon's implementation of the decision rule has been tested under a large number of scenarios:
\begin{itemize}
    \item \textbf{Partially-Synchronous Mode:} \tarpon has been tested to work correctly under the partially-synchronous mode. This simply means that \tarpon's tests where run in such a way that the asynchronous mode was never reached. Replicating the results obtained from the \textsf{Mysticeti}\xspace protocol \cite{MYSTICETI}.
    \item \textbf{Asynchronous Mode:} \tarpon has been tested to work correctly under the asynchronous mode. Just like in the previous scenario, the tests here make sure that the protocol works when the \texttt{asynchronous\_wave} is set to either $4$ or $5$. 
    \item \textbf{Dual-Mode:} \tarpon has been tested to work correctly when switching between modes. This means making sure that the protocol correctly switches between partially-synchronous mode and asynchronous mode all in one run. As in the previous scenario, all the tests pass whether \texttt{asynchronous\_wave} is set to either $4$ or $5$
\end{itemize}
\quad As \tarpon supports a number of configurations, it has been tested under a combination of scenarios (e.g., \texttt{pipeline = true} and number of leaders \texttt{= 1}, etc...). This combination of configurations allows us to make sure that \tarpon works correctly under all settings. Due to the substantial amount of similar code across tests, we present only the implementation details for the dual-mode scenario, found in Appendix~\ref{appendixL:testConsensusProtocol}. In particular, we present our contributed tests, which verify that \tarpon's asynchronous mode operates correctly, and that it is capable of transitioning between modes correctly.
\subsection{Dual-Mode Scheduler} 
\quad Once the implementation for the consensus protocol has been properly tested, further tests have been created for the dual-mode scheduler. These tests mainly aimed at making sure that the deterministic algorithms (e.g. updating the rate for the asynchronous mode) worked as intended. We focused \tarpon's dual-mode scheduler testing on the more complex functions; specifically, the functions that allowed the dynamic portion of the scheduler to work as intended. This meant making sure that \tarpon's \texttt{is\_directly\_committed()} and \texttt{collectSubdag()} functions worked as intended under different scenarios. 
\begin{itemize}
    \item \texttt{collectSubdag()}: These tests made sure that a subDAG was collected appropriately. This meant making sure that, in scenarios where a subDAG of an anchor was disconnected from many blocks (i.e., in case of asynchrony), then, the collected subDAG only included blocks that could be accessed through the anchor's causal history. 
    \item \texttt{is\_directly\_committed()}: Just like the tests conducted on the consensus protocol, this function was tested to work correctly under both multi-leader and pipeline modes, as well as correctly ignoring leader slots marked as \texttt{to-commit} through the indirect rules. 
\end{itemize}
\quad Additionally, the correctness of the leader sequence truncation has been confirmed through tests performed in the core consensus module (\texttt{core.rs}). This test verifies that the system properly identifies asynchronous rounds within committed sequences and truncates at the appropriate boundary, ensuring the correct re-calculation of the asynchronous schedule.
\section{Simulations}
\label{section:simulations}
\quad The implementation of the dual-mode protocol has been rigorously tested via a simulation layer, originally introduced by Babel et al. \cite{MYSTICETI} for the \textsf{Mysticeti}\xspace protocol.. The provided simulation layer allowed us to replicate the functionality of \texttt{tokio} runtime and TCP networking. This simulation layer not only allowed us to simulate real-life WAN latencies, but to also mimic the passage of time. In this way, we were able to test the dual-mode functionality of \tarpon in numerous scenarios, allowing us to verify that \tarpon's dual-mode scheduler was working as intended. Like stated in Section~\ref{sec:impl_dual_mode}, \tarpon's simulation tests utilise a different \texttt{DualModeConfig} then the default preset, to allow the simulation test to hit the asynchronous mode more often then in a normal use case. The discrete event simulator tests confirmed that \tarpon's dual-mode scheduler correctly increases and decreases depending on network conditions, which was achieved by creating the following scenarios for the simulations (further detail can be seen in Appendix~\ref{appendix:simulation}):
\begin{itemize}
    \item \textbf{Ideal Case:} Latency range for all validators is within the timeouts set by the protocol; thus, all leaders proposal are successfully committed. Since the amount of direct commits is $100\%$, the scheduler does not detect any issues, thus it decides to run the asynchronous mode less frequently, decreasing the intervals frequency.
    \item \textbf{Worst Case:} $f$ validators are given a higher latency range then the remaining $2f+1$ validators. This will cause them to miss their leader slots in rounds where they are elected leaders, consequently making their proposal marked as \texttt{to-skip}. The test shows that frequency of \tarpon's asynchronous mode increases, as there are not enough commits to reach the required ratio. Interestingly, we cannot set more than $f$ validators to have higher latency due to the construction of \tarpon. When there are more than $f$ slow validators, the scheduler actually does not detect any problems with the system and subsequently decreases the rate of the asynchronous mode. Since validators need $2f+1$ references to propose a block, if $f+1$ validators are slow, the other validators must wait for at least one slow validator to complete the $2f+1$ references required. Thus, they all slow down to the same high-latency as the slow validators, thus allowing slow validators to directly commit their blocks. 
\end{itemize}
\quad Lastly, as stated by Babel et al. \cite{MYSTICETI}, the simulation layer appears to provide similar results when compared to geo-distributed cloud infrastructure tests. Specifically, we can obtain simulated latencies for the time it takes for a transaction to be committed. To understand the benefits of \tarpon's dual-mode functionality, it is beneficial to test the protocol under a number of different settings:
\begin{itemize}
    \item \textbf{Partially-Synchronous mode only:} This setting will reflect the latency of commitments under partially-synchronous assumptions (i.e., emulating how \textsf{Mysticeti}\xspace performs). We run \tarpon exclusively under its partially-synchronous mode, thus we will call this the \textit{p-sync} mode.
    \item \textbf{Asynchronous mode only:} This setting emulates the latency of commitments under asynchronous assumptions (i.e., it emulates how \textsf{Mahi-Mahi}\xspace performs). We run \tarpon exclusively under its asynchronous mode, thus we will call this the \textit{async} mode. 
    \item \textbf{Dual-Mode:} This setting will reflect how \tarpon performs when the dual-mode scheduler is turned on (with an initial asynchronous interval $K$ for every 30 rounds). We will call this the \textit{dual} mode. Of note, $K$ has been chosen so that, in our simulated setting, the protocol can reach multiple asynchronous-rounds. 
\end{itemize}
\quad Both the \textit{async} and \textit{dual} mode tests will be run utilising both configurable wavelengths ($w=4$ and $w=5$). While the simulation layer provides a number of metrics, we will only be focusing on the following, due to its relevance to the project: \texttt{tx commit(ms)}. 
The \texttt{tx commit(ms)} metric measures the transaction commit latency, defined as the time from when a transaction is received by a validator until it is committed in \tarpon.
\begin{table}[H]
    \centering
\begin{tabular}{lcc}
\toprule
\textbf{Operation Mode} & \textbf{Latency} ($ms$) & \textbf{vs. Baseline} \\
\midrule
\textit{p-sync} mode &  350.8 & {\textbf{baseline}} \\
\midrule
\textit{dual} mode ($w=4$) & 355.2 & {+1.3\%} \\
\textit{dual} mode ($w=5$) & 362.2 & {+3.3\%} \\
\midrule
\textit{async} mode ($w=4$) & 439.5 &{+25.3\%} \\
\textit{async} mode ($w=5$) & 530.9 &{+51.3\%}  \\
\bottomrule
\end{tabular}
\caption{Results for average latencies (in milliseconds). These results are visualized in Figure~\ref{visualizationtarpon}, found in Appendix~\ref{appendix:TARPADDITIONAL}}
\label{tab:latency_results}
\end{table}

\quad Table~\ref{tab:latency_results} present the results of our simulation tests on \tarpon. As expected, when \tarpon is run exclusively under its partially-synchronous mode, it performs best compared to both of the two other modes. Furthermore, we can see that, if \tarpon were to be run solely under its asynchronous mode, the latency would increase substantially. Thus, here we reach the true benefit of \tarpon's dual-mode functionality. We can clearly see that \tarpon, when run under its intended mode (that is with its dual-mode functionality), exhibits a minimal increase in latency when compared to the \textit{p-sync} mode. In compensation, however, \tarpon gains the "ability" to achieve liveness in asynchronous settings, achieving greater flexibility for a marginal increment in latency. Furthermore, \tarpon's latency increase between $w=4$ and $w=5$ is minimal compared to the $\approx 20\%$ increase we see in the \textit{async} mode. This means we can sacrifice a minimal amount of latency to gain stronger guarantees of \textit{liveness} in periods of asynchrony. That said, it is important to note that these simulation results cannot replace tests performed via geo-distributed cloud infrastructures, as done by prior works \cite{BULLSHARK,TUSK,MYSTICETI,MAHIMAHI}. Such tests would also help us determine what interval value $K$ would be the most optimal for \tarpon. Nonetheless,
Section~\ref{section:furtherwork} presents the possible next steps that can be taken beyond the scope of this project.

\quad We have not only presented \tarpon, the first dual-mode \textit{uncertified} DAG-based protocol, but also demonstrated its functionality. Following this, we conclude this project by presenting the formal proofs of \tarpon's correctness.
\chapter{Dual-Mode Protocol Security}
\label{chap:Proofs}
\quad We shall now prove that \tarpon never violates the \textit{safety} property, and that, eventually, achieves \textit{liveness} both before and after $GST$ is reached. We will do so by proving that \tarpon solves BAB, which allows validators to reach consensus on a sequence of messages, achieving State Machine Replication (SMR) \cite{BAB,SMR}. While similar to the atomic broadcast protocol described in Section~\ref{sec:foundations}, BAB imposes stricter requirements that are formally defined as follows \cite{BAB}:
\begin{itemize}
    \item \textbf{Integrity:} For an honest validator $v_i$, and every round number $r \in \mathbb{N}$, an honest validator $v_j$ will output \textit{a\_deliver$_j \; (b,r,v_i)$} at most once, independent of $b$. 
    \item \textbf{Total Order:} If an honest validator $v_i$ outputs \textit{a\_deliver$_i \; (b,r,v_k)$} before \textit{a\_deliver$_i \; (b',r',v_k')$}, then all honest validators $v_j$ output \textit{a\_deliver$_j \; (b,r,v_k)$} before \textit{a\_deliver$_j \; (b',r',v_k')$}.
    \item \textbf{Validity:} If an honest validator $v_i$ calls \textit{a\_broadcast$_i \; (b,r)$}, then, every other honest validator $v_j$ will eventually output \textit{a\_deliver$_j \; (b,r,v_i)$}. 
    \item \textbf{Agreement:} If an honest validator $v_i$ outputs \textit{a\_deliver$_i \; (b,r,v_k)$}, then every other honest validator $v_j$ will eventually output \textit{a\_deliver$_j \; (b,r,v_k)$}. 

\end{itemize}
\quad By proving that \tarpon meets the aforementioned properties, we also prove indirectly that it also achieves the \textit{safety} and \textit{liveness} properties. As we will see, due to the similarity between the protocols, the proofs revolving around safety are reproduced from those presented in the work of Jovanovic et al. \cite{MAHIMAHI} for the \textsf{Mahi-Mahi}\xspace protocol, with slight modifications to fit \tarpon's distinct dual-mode mechanism. The proofs of \textit{liveness}, instead, are identical to the proofs as seen in both \textsf{Mysticeti}\xspace and \textsf{Mahi-Mahi}\xspace, depending on whether $GST$ has been reached or not \cite{MYSTICETI,MAHIMAHI}. The primary contributions are the proofs that \tarpon's dynamic schedule does not interfere with the aforementioned proofs, but instead, maintains both \textit{safety} and \textit{liveness}. These proofs have been adapted from the works of Tsimos et al. \cite{HAMMERHEAD}, modified to take advantage of \tarpon's stronger \textit{liveness} guarantees.
\section{Total Order and Integrity Proofs}
\quad We will begin by proving that \tarpon satisfies the \textit{total order} and \textit{integrity} properties of BAB \cite{BAB}. The initial proofs of \textit{total order} and \textit{integrity}, as stated, are reproduced  from the work by Jovanovic et al. \cite{MAHIMAHI}. These are included both for completeness and because they are required to prove the security of the dynamic scheduler.
\begin{lemma}
\label{lem:certificate_propagation}
If in round $r$, $2f+1$ blocks from distinct validators certify a block $b$, then all blocks at future rounds $r'>r$ will have a path to a certificate for $b$ from round $r$.
\end{lemma}
\begin{proof}
We prove the lemma by induction on $r'$. For the base case $r'= r + 1$. Let $b'$ be a block at round $r'$. Since $b'$ points to $2f + 1$ blocks at round $r$, by quorum intersection, $b'$ must point to at least one of the certificates for $b$. For the induction case, assume the lemma holds up to round $r'$ and consider the case of round $r' + 1$. Let $b'$ be a block at round $r' + 1$. By the induction hypothesis, $2f + 1$ blocks at round $r'$ have paths to round-$r$ certificates for $b$. Since $b'$ points to $2f + 1$ blocks from round $r'$, by quorum intersection, $b'$ must point to at least one block that has a path to a round-$r$ certificate for $b$ \cite{MAHIMAHI}.
\end{proof}
\begin{observation}
\label{obv:only_one_vote}
A block cannot vote for more than one block proposal from a given validator, in a given round.    
\end{observation}
\begin{proof}
 This is by construction. Honest validators interpret support in the DAG through deterministic depth-first traversal. So even if a block $b$ in the vote round has paths to multiple leader round blocks from the same validator $v$ (i.e., equivocating blocks), all honest validators will interpret $b$ to vote for only one of $v$’s blocks (the first block to appear in the depth-first traversal starting from $b$) \cite{MAHIMAHI}.
\end{proof}
\begin{lemma}
\label{lem:atmost_single_block}
    At most a single block per round from the same validator can be certified.
\end{lemma}
\begin{proof}
    Assume by contradiction that in a given round $r$, there exist two distinct blocks $b$ and $b'$ from the same validator $v$ such that both $b$ and $b'$ are certified. This means that there exist round-$(r + w - 1)$ blocks $c_b$ and $c_{b'}$ that certify $b$ and $b'$, respectively. $c_b$ and $c_{b'}$ must point to $2f + 1$ votes for $b$ and $b'$, respectively. By quorum intersection, there exists an honest validator that has voted for both $b$ and $b'$ in the vote round. Since honest validators only produce a single block per round, this implies that there exists a block that votes for both $b$ and $b'$, contradicting Observation~\ref{obv:only_one_vote} \cite{MAHIMAHI}.
\end{proof}
\begin{observation}
\label{obv:local_dag_contains_certificate}
    If an honest validator $v$ directly or indirectly commits a block $b$, then $v$'s local DAG contains a certificate for $b$.
\end{observation}
\begin{proof}
    Follows immediately from \tarpon's direct and indirect commit rules
\end{proof}
\quad Given that \tarpon's leader slots are decided differently, depending on the current mode of operation, we adapt the following observation to better suite the protocol.
\begin{observation}
\label{obv:agree_sequence_leader}
Honest validators agree on the sequence of leader slots.
\end{observation}
\begin{proof}
This is true for both \tarpon's operation modes. With its asynchronous mode, it follows immediately from the properties of the common coin. While for its partially-synchronous mode, it follows immediately from the use of \tarpon's deterministic leader election algorithm. 
\end{proof}
\begin{lemma}
\label{lem:direct_commit_consistency}
    If an honest validator $v$ commits some block $b$ in a slot $s$, then no other honest validator decides to directly skip the slot $s$
\end{lemma} 
\begin{proof}
    Assume by contradiction that some honest validator $v'$ decides to directly skip $s$. Then it must be the case that in the local DAG of $v'$, at least $2f + 1$ validators did not vote for $b$. However, since $v$ commits $b$ at $s$, by Observation~\ref{obv:local_dag_contains_certificate}, there must exist a certificate for $b$ at $s$. So in $v$'s local DAG there must be $2f + 1$ validators that vote for $b$. By quorum intersection, at least one honest validator both voted for $b$ and did not vote for $b$. Since honest validators produce a single block in the vote round, this is a contradiction \cite{MAHIMAHI}.
\end{proof}

\begin{lemma}
\label{lem:general_commit_consistency}
If an honest validator directly commits some block in a slot $s$, then no other honest validator decides to skip the slot $s$.
\end{lemma}

\begin{proof}
Assume by contradiction that an honest validator $v$ directly commits block $b$ in slot $s$ while another honest validator $v'$ decides to skip $s$. By Lemma~\ref{lem:direct_commit_consistency}, $v'$ cannot directly skip $s$; it must be the case therefore that $v'$ skips $s$ using the indirect decision rule. Let $r$ be the round of $s$. Since $v$ directly commits $b$, there exist $2f + 1$ certificates for $b$ at $s$. Therefore, by Lemma~\ref{lem:certificate_propagation}, all blocks at rounds $r' > r + w - 1$, including the anchor of $s$, have a path to a certificate for $b$ at $s$. Thus, $v'$ cannot decide to skip $s$ using the indirect decision rule. We have reached a contradiction \cite{MAHIMAHI}.
\end{proof}

\begin{lemma}
\label{lem:sameblock_committed}
If a slot $s$ is marked as \texttt{to-commit} by two honest validators, then slot $s$ contains the same block for both honest validators
\end{lemma}

\begin{proof}
Let $v$ and $u$ be two honest validators and assume that $v$ commits block $b$ at slot $s$. We will show that if $u$ commits slot $s$, then $s$ contains $b$ at $s$. Let $w$ be the validator that produced block $b$. By Observation~\ref{obv:local_dag_contains_certificate}, for $b$ to be committed at slot $s$ at $v$, there must exist at least one certificate for $b$. By Observation~\ref{obv:agree_sequence_leader}, $v$ and $u$ agree that $s$ must contain a block by $w$. By Lemma~\ref{lem:atmost_single_block}, at most a single block per round from $w$ can be certified. So $u$ cannot have a certificate for any other block than $b$ at slot $s$ \cite{MAHIMAHI}.
\end{proof}
\begin{lemma}
\label{lem:both_validator_commit_skip}
If a slot $s$ is decided at two honest validators $v$ and $v'$, then either both validators commit $s$, or both validators skip $s$.
\end{lemma}
\begin{proof}
Assume by contradiction that there exists a slot $s$ such that $v$ and $v'$ decide differently at $s$. We consider a finite execution prefix and assume wlog that $s$ is the highest slot at which $v$ and $v'$ decide differently (*). Further assume wlog that $v$ commits $s$ and $v'$ skips $s$. By Lemma~\ref{lem:general_commit_consistency} and Lemma~\ref{lem:direct_commit_consistency}, neither $v$ nor $v'$ could have used the direct decision rule for $s$; they must both have used the indirect rule. Consider now the anchor of $s$: $v$ and $v'$ must agree on which slot is the anchor of $s$, since by our assumption (*) above, they make the same decisions for all slots higher than $s$, including the anchor of $s$. Let $s'$ be the anchor of $s$; $s'$ must be committed at both $v$ and $v'$. Thus, by Lemma~\ref{lem:sameblock_committed}, $v$ and $v'$ commit the same block $b'$ at $s'$. But then $v$ and $v'$ cannot reach different decisions about slot $s$ using the indirect decision rule. We have reached a contradiction \cite{MAHIMAHI}.
\end{proof}
With Lemma~\ref{lem:both_validator_commit_skip}, we prove that there is consistency amongst honest validators' commit sequences. If an honest validator commits a leader block, then all other honest validators will also commit that same block. The same is true in the case if a honest validator decides to skip a leader block. However, we now need to prove that, beyond leader blocks, also non-leader blocks are delivered in the same order by honest validators.
\paragraph{Causal history \& delivery conditions.} Consider an honest validator $v$. We call the causal history of a block $b$ in $v$'s DAG, the transitive closure of all blocks referenced by $b$ in $v$'s DAG, including $b$ itself. In \tarpon, a block $b$ is delivered by an honest validator $v$ if (1) there exists a committed leader block $l$ in $v$'s DAG such that $b$ is in $l$'s causal history (2) all slots up to $l$ are decided in $v$'s DAG and (3) $b$ has not been delivered as part of a lower slot's causal history. In this case we say $b$ is delivered at slot $s$, or delivered with block $l$ \cite{MAHIMAHI}.

\begin{lemma}
\label{lem:same_block_delivered}
If a block $b$ is delivered by two honest validators $v$ and $v'$, then $b$ is delivered at the same slot $s$, and $b$ is delivered with the same leader block $l$, at both $v$ and $v'$.
\end{lemma}

\begin{proof}
Let $s$ be the slot at which $b$ is delivered at validator $v$, and $l$ the corresponding leader block in $s$, also at validator $v$. Consider now the slot $s'$ at which $b$ is delivered at validator $v'$, and $l'$ the corresponding leader block. Assume by contradiction that $s' \neq s$. If $s' < s$, then $v$ would have also delivered $b$ at slot $s'$, since by Lemma~\ref{lem:sameblock_committed} must commit the same leader blocks in the same slots, so $v$ could not have delivered $b$ again at slot $s$; a contradiction. Similarly, if $s < s'$, then $v'$ would have already delivered $b$ at slot $s$, since by Lemma~\ref{lem:sameblock_committed} $v$ and $v'$ must have committed the same block in slot $s$; contradiction. Thus it must be that $s = s'$, and by Lemma~\ref{lem:sameblock_committed}, $l = l'$ \cite{MAHIMAHI}.
\end{proof}

\quad Having established that both leader and non-leader blocks are delivered in consistent order by honest validators, we now turn to proving that \tarpon's dynamic scheduling mechanism maintains safety across schedule transitions. As previously noted, we prove the safety of the dynamic scheduling mechanism by adapting proofs from Tsimos et al. \cite{HAMMERHEAD}, developed for the \textsf{Hammerhead}\xspace protocol. We first show that blocks broadcast by honest validators will, eventually, by added to the local DAG of honest validators.

\begin{lemma}
\label{lem:localDag_included}
If a block $b$ produced by an honest validator $v$ references some block $b'$, then $b'$ will eventually be included in the local DAG of every honest validator.
\end{lemma}
\begin{proof}
This is ensured by the synchronizer sub-component in each validator: if some validator $w$ receives $b$ from $v$, but does not have $b'$ yet, $w$ will request $b'$ from $v$; since $v$ is honest and the network links are reliable, $v$ will eventually receive $w$'s request, send $b'$ to $w$, and $w$ will eventually receive $b'$. The same is recursively true for any blocks from the causal history of $b'$, so $w$ will eventually receive all blocks from the causal history of $b'$ and thus include $b'$ in its local DAG \cite{MAHIMAHI}.
\end{proof}

\begin{observation}
\label{lem:dag-propagation}
If an honest validator $v$ broadcasts a block $b$ at round $r$, then every honest validator will eventually include $b$ in their local DAG.
\end{observation}

\begin{proof}
This holds by construction. Since network links are reliable, all honest validators eventually receive $b$ from $v$. 
\end{proof}
\begin{observation}
\label{obv:consistent_commit_view}
If an honest validator commits a block $b$ at slot $s$, then every honest validator that commits slot s will commit the same block b, and when they do commit, they will have the same causal history for b in their DAG.
\end{observation}

\begin{proof}
By Observation~\ref{lem:dag-propagation}, if an honest validator commits block $b$, then every honest validator will eventually include $b$ in their local DAG. By Lemma~\ref{lem:sameblock_committed}, if two honest validators both commit slot $s$, they commit the same block $b$ at that slot. When an honest validator commits $b$, by Lemma~\ref{lem:localDag_included}, they must have the complete causal history of $b$ in their local DAG, which is deterministic and identical across all validators \cite{MAHIMAHI}.
\end{proof}
\vspace{1em}
\begin{lemma}[Schedule Agreement]
\label{lem:schedule_agreement}
If an honest validator $v_i$ switches to schedule $S$, then eventually every honest validator will switch to schedule $S$.
\end{lemma}
\begin{proof}
We prove by induction that all honest validators will eventually converge to the same schedule:
\begin{itemize}
\item \textbf{Base case:} Consider the initial schedule $S_0$, where the asynchronous mode runs every $K_0$ rounds, and $S_{0_{start\_round}}=0$. Let $v_i$ be the first honest validator to switch to schedule $S_1$. Given Algorithm~\ref{alg:main} (process \textsc{TryCommit}), $v_i$ must have committed some anchor $\ell_i$ for round $r_i$ where $r_i \equiv 0 \pmod{K_0}$ to trigger the schedule update. Consider another honest validator $v_j$ who has committed every leader proposal up to round $r_j$ where $r_j < r_i$. If $r_j \geq r_i$, then by Algorithm~\ref{alg:main} (process \textsc{TryCommit}), $v_j$ would have already switched to $S_1$ by round $r_i$ (given Observation~\ref{obv:consistent_commit_view}). Thus $v_j$ will eventually commit some anchor $\ell_m$ where $r_m > r_i$. By quorum intersection, since validator $v_i$ committed anchor $\ell_i$ in round $r_i$, there must exist a path from $\ell_m$ to $\ell_i$. Therefore, by Lemma~\ref{lem:same_block_delivered} and the protocol's linearizer component, validator $v_j$ will order anchor $\ell_i$. Consequently, $v_j$ will trigger a schedule switch and, by Observation~\ref{obv:consistent_commit_view}, $v_j$ will go from schedule $S_0$ to schedule $S_1$.
\item \textbf{Inductive step:} Assume the statement holds for all schedules $S_0, S_1, \ldots, S_k$. We prove it also holds for $S_{k+1}$. Let $v_i$ be the first honest validator switching from $S_k$ to $S_{k+1}$. By Algorithm~\ref{alg:main} (process \textsc{TryCommit}), $v_i$ must have committed some anchor $\ell_i$ for round $r_i \equiv 0 \pmod{K_k + S_{{k-1}_{\text{start\_round}}}}$. For any other honest validator $v_j$ currently in some schedule $S_r$ where $r < k + 1$, by the induction hypothesis, $v_j$ will eventually switch to $S_k$. Consider $v_j$ having committed up to round $r_j < r_i$. Eventually, $v_j$ will switch to $S_k$ and subsequently commit some anchor $\ell_m$ for round $r_m > r_i$. By quorum intersection, since $v_i$ committed $\ell_i$ in round $r_i$, there exists a path from $\ell_m$ to $\ell_i$. Therefore, by Lemma~\ref{lem:same_block_delivered} and the protocol's linearizer component, validator $v_j$ will order anchor $\ell_i$. Consequently, $v_j$ will trigger a schedule switch and, by Observation~\ref{obv:consistent_commit_view}, $v_j$ will go from schedule $S_k$ to schedule $S_{k+1}$.
\end{itemize}
\end{proof}
\newpage
\begin{lemma}
\label{lem:view_syncrhonization}
Let $S_{\max}$ represent the highest schedule achieved by any honest party. Eventually, every honest party can reach schedule $S_{\max}$.
\end{lemma}

\begin{proof}
Let $P_i$ be the honest party that advanced to schedule $S_{\max}$. By Observation~\ref{obv:consistent_commit_view}, all anchor blocks that $P_i$ committed to reach $S_{\max}$ will eventually be added to every honest party's DAG. Since the schedule switching rule is deterministic and all honest parties process the same committed anchors, they will all compute the same schedule transitions and eventually reach $S_{\max}$.
\end{proof}
\begin{theorem}[Total Order]
\tarpon satisfies the total order property of Byzantine Atomic Broadcast.
\end{theorem}
\begin{proof}
By construction, due to the deterministic algorithm to order the causal history of a block $b$, and via Lemma~\ref{lem:same_block_delivered}, \tarpon attains the BAB total order property. Furthermore, Lemma~\ref{lem:schedule_agreement} and~\ref{lem:view_syncrhonization} demonstrate that \tarpon's dynamic scheduling mechanism preserves the total order property by ensuring all honest validators eventually converge to the same schedule.
\end{proof}
\begin{theorem}[Integrity]
\tarpon satisfies the integrity property of Byzantine Atomic Broadcast.
\end{theorem}
\begin{proof}
This is by construction. The deterministic linearizer algorithm present in the protocol removes blocks with duplicate sequence numbers before creating the final sequence of blocks. Since a block $b$ is delivered via the causal history of some leader block $l$ only if $b$ has not been delivered already, thus honest validators cannot propose block $b$ multiple times \cite{MAHIMAHI}. 
\end{proof}
\section{Validity and Agreement Proofs}
\quad Due to \tarpon encapsulating both partially-synchronous and asynchronous assumptions, we can prove, separately, that the protocol satisfies both the \textit{validity} and \textit{agreement} BAB properties, under both network assumptions. Furthermore, we show that, regardless of network conditions, \tarpon's schedule achieves liveness.
\newpage
\begin{lemma}
\label{lem:asynchrony_liveness}
Under asynchronous network conditions (before $GST$), the protocol will eventually make progress.
\end{lemma}
\begin{proof}
 Assume a period of time where $GST$ is not reached for a long time, as \tarpon's asynchronous mode allows $w\in\{4,5\}$, we have two cases where we have to prove liveness:
\begin{itemize}
    \item \textbf{Case} $w=4$ \textbf{:} By Theorem 3 and Theorem 4 of \textsf{Mahi-Mahi}\xspace \cite{MAHIMAHI}.
    \item \textbf{Case} $w=5$ \textbf{:} By Theorem 5 and Theorem 6 of \textsf{Mahi-Mahi}\xspace \cite{MAHIMAHI}.
\end{itemize}
Thus, in either case, an honest validator proposal will eventually be committed and delivered consistently across all validators, ensuring liveness despite the absence of timing guarantees.
\end{proof}
\begin{lemma}
\label{lem:partial_synchrony_liveness}
After $GST$ has been reached, the protocol will eventually make progress within a bounded number of rounds.
\end{lemma}
\begin{proof}
Assume the asynchronous rounds are scheduled sufficiently far apart to allow network synchronization. By the \textit{validity} and \textit{agreement} properties of \textsf{Mysticeti}\xspace \cite{MYSTICETI}, honest validator proposals will be committed within a bounded number of rounds due to the timing guarantees after $GST$.
\end{proof}
\begin{lemma}
\label{lem:dual_mode_liveness}
The protocol will eventually make progress.
\end{lemma}
\begin{proof}
This follows directly from Lemma~\ref{lem:asynchrony_liveness} and Lemma~\ref{lem:partial_synchrony_liveness}.
\end{proof}
\begin{lemma}
\label{lem:mode_agreement}
All honest validators agree on whether to use the partially-synchronous (\textsf{Mysticeti}\xspace) or asynchronous (\textsf{Mahi-Mahi}\xspace) mode for each round.
\end{lemma}
\begin{proof}
From Lemma~\ref{lem:schedule_agreement}, we know that all honest validators eventually converge to the same schedule. Since the choice of consensus mode is determined by the current schedule using a deterministic function, all honest validators operating under the same schedule will make identical consensus mode decisions for each round.
\end{proof}
\quad We now prove that the dynamic schedule achieves \textit{liveness}, irrespective of whether $GST$ has been reached or not. As before, portions of the following proofs have been adapted from the work of Tsimos et al. \cite{HAMMERHEAD}. In this case, specifically, the liveness Lemma for \textsf{Hammerhead}\xspace, have been translated to take advantage of \tarpon stronger \textit{liveness} guarantees
\begin{lemma}
\label{lem:schedule_switch}
Let $S$ be a schedule. If all honest validators are in schedule $S$, then eventually all honest validators will switch to the next schedule.
\end{lemma}
\begin{proof}
From Lemma~\ref{lem:dual_mode_liveness}, the protocol will eventually make progress. This means that eventually we will reach the necessary round $S_{\text{start\_round}} + K$ to trigger a schedule switch. When this round is reached, the protocol will commit the corresponding anchor block. Since all honest validators are operating under the same schedule $S$ and agree on the consensus mode (Lemma~\ref{lem:mode_agreement}), they will all observe the same committed blocks and apply the same deterministic schedule update rule.
\end{proof}
\begin{lemma}
\label{lem:schedule_liveness}
Let $S_{\max}$ be the latest schedule any honest party has advanced to. Eventually some honest party will enter $S' = S_{\max} + 1$.
\end{lemma}
\begin{proof}
From Lemma~\ref{lem:view_syncrhonization}, every honest party will eventually be at $S_{\max}$. Once all honest parties are at schedule $S_{\max}$, we have two cases:
\begin{itemize}
\item \textbf{Case 1 (Partially synchronous conditions):} If $GST$ is reached, by Lemma~\ref{lem:partial_synchrony_liveness}, the dual-mode protocol will eventually commit an anchor block. By Lemma~\ref{lem:schedule_switch}, this triggers advancement to $S_{\max} + 1$.
\item \textbf{Case 2 (Asynchronous conditions):} If $GST$ is not reached, by Lemma~\ref{lem:asynchrony_liveness}, the dual-mode protocol will eventually commit an anchor block. By Lemma~\ref{lem:schedule_switch}, this triggers advancement to $S_{\max} + 1$.
\end{itemize}
In both cases, eventual progress to $S_{\max} + 1$ is guaranteed.
\end{proof}
\begin{theorem}[Validity]
\tarpon satisfies the validity property of Byzantine Atomic Broadcast.
\end{theorem}
\begin{proof}
Let $v$ be an honest validator and $b$ a block broadcast by $v$. By Observation~\ref{lem:dag-propagation}, $b$ is eventually delivered by every honest validator. From Lemma~\ref{lem:dual_mode_liveness}, the protocol will eventually make progress and commit some leader blocks. From Lemma~\ref{lem:schedule_liveness}, the system will continue to advance through schedules, ensuring that some leader containing $b$ in its causal history will eventually be committed. Therefore, all honest validators will eventually deliver $b$.
\end{proof}
\begin{theorem}[Agreement]
\tarpon satisfies the agreement property of Byzantine Atomic Broadcast.
\end{theorem}
\begin{proof}
Suppose an honest validator $v$ delivers a block $b$. Then $b$ was delivered as part of some committed leader's causal history. From Lemma~\ref{lem:schedule_liveness}, all honest validators will eventually converge to the same schedule sequence and therefore observe the same committed leaders. From Lemma~\ref{lem:dual_mode_liveness}, all validators will continue to make progress and process these committed leaders. Therefore, all honest validators will eventually deliver $b$.
\end{proof}
\chapter{Conclusion}
\label{chap:conclusion}
\quad Dual-mode consensus protocols remain the ideal protocols for decentralised distributed ledgers. By providing liveness in both synchronous and asynchronous network setting, these protocols can better reflect real-life scenarios where we cannot guarantee networks to remain stable. While partially-synchronous assumptions are easier to reason with, and faster then their asynchronous counterparts, liveness is lost in periods of asynchrony \cite{SAFETYLIVENESS}. On the other hand, asynchronous systems sacrifice latency to increase robustness in periods of asynchrony. We have seen with both \textsf{Ditto}\xspace and \textsf{Bullshark}\xspace's evaluations that dual-mode consensus protocols can achieve results that rival similar state-of-the-art protocols \cite{DITTO, BULLSHARK}. Thus, with this research project, we have presented the next step in dual-mode consensus algorithms. \tarpon, the first dual-mode consensus algorithm built on an \textit{uncertified} data structure, achieves the 3-message-delay lower bound for the greater part of its duration, while providing liveness even during periods of asynchrony. Empirical evaluation was beyond the scope of this work; however, based on theoretical analysis, \tarpon should achieve similar results (both throughput and latency) as \textsf{Mysticeti}\xspace, while maintaining lower latency compared to \textsf{Mahi-Mahi}\xspace \cite{MYSTICETI,MAHIMAHI}. Furthermore, beyond providing detailed specifications and implementation, in Chapter~\ref{chap:dual_mode_protocol} and Chapter~\ref{chap:implementation} we also demonstrate that \textit{safety} and \textit{liveness} are maintained with \tarpon. Hence, these results establish well defined future research objectives.
\vspace{1em}
\section{Further Work}
\label{section:furtherwork}
\quad Given the result of our project, the protocol that has been constructed is guaranteed to achieve both \textbf{safety} and \textbf{liveness} in both synchronous and asynchronous settings. Thus, the next step of the project would be to evaluate the protocol by running experiments on a geo-distributed cloud infrastructure. These experiments would follow a similar structure to the experiments performed on \textsf{Mysticeti}\xspace and \textsf{Mahi-Mahi}\xspace (as seen in Figure~\ref{fig:mysticeti_results} and Figure~\ref{fig:mahi_results} respectively), allowing us to improve and evaluate the protocol. 
\newpage
\quad We do not know which value $K$ for the interval found in \tarpon's dynamic schedule would be optimal. Therefore, by running experiments, we could determine the most effective \texttt{min\_async\_interval} and \texttt{max\_async\_interval} bounds for the interval, alongside other improvements such as the optimal \texttt{adjustment\_factor} after a schedule update (a description of all these variables can be found in Appendix~\ref{appendix:DualModeConfig}). Furthermore, the evaluations would allow us to see how \tarpon performs against other consensus protocols (such as the frequently discussed \textsf{Mysticeti}\xspace and \textsf{Mahi-Mahi}\xspace protocols) \cite{MYSTICETI, MAHIMAHI}. 

\quad Specifically, by following similar evaluation procedures as related protocols, through these experiments we would able to demonstrate the following claims:
\begin{itemize}
    \item \textbf{C1:} \tarpon has similar throughput and lower latency compare to the \textsf{Mahi-Mahi}\xspace consensus protocol when operating under ideal conditions \cite{MAHIMAHI}. 
    \item \textbf{C2:} \tarpon has similar throughput and marginally higher latency compared to the \textsf{Mysticeti}\xspace consensus protocol when operating under ideal conditions \cite{MYSTICETI}. 
    \item \textbf{C3:} \tarpon, parametrised with a wave length of $4$ for its asynchronous mode, has similar throughput and lower latency as when configured with a wave length of $5$. 
\end{itemize}
\quad Similarly to prior work, to determine the validity of the aforementioned claims, \tarpon should be tested under a number of different scenarios. Specifically, \tarpon should be tested with both a small and large committee size (i.e., committee of size 10 and 50). Furthermore, \tarpon requires evaluation under both ideal-conditions (i.e., all committee members are working as intended), alongside performance under crash-faults (i.e., some committee members are fault). It is important to note that the evaluation of BFT consensus protocols under the presence of Byzantine faults is an open research question, thus \tarpon relies on formal proofs to guarantee \textit{safety} and \textit{liveness} (provided in Chapter~\ref{chap:Proofs}) \cite{BFTRESEARCH}. Furthermore, as mentioned by related research, while it is important that \tarpon handles Byzantine faults correctly, the presence of such faults in real-life proof-of-stake blockchains are rare, and crash-faults are much more common occurrences \cite{MAHIMAHI, MYSTICETI}.
\small
\printbibliography
\appendix
\chapter{Code Structure}
\small
\section{Implementation}
\subsection{Decision Rule}
\subsubsection{Universal Committer}
\label{Appendix:UniversalCommitter}
\quad Listing~\ref{listing:universalcommitter} shows a snippet of code that shows the primary modifications made to the Mysticeti \texttt{UniversalCommitter} structure \cite{MYSTICETI}. Primarily, we include a new variable \texttt{wave\_length\_async} which stores how long our asynchronous mode wave length is going to be (i.e., either $4$ or $5$ rounds). Furthermore, we allow for this value to be set during the creation of the structure, via \texttt{with\_async\_wave\_length}. Finally, we modified the \texttt{try\_commit} function to also pass the \texttt{AsyncState} structure, which holds all the information required for the function to determine if a round should be run via its asynchronous mode or not (this structure is expanded upon in Appendix~\ref{appendix:AsyncState}). 
   \begin{lstlisting}[language=Rust, caption={UniversalCommitter async state integration},label={listing:universalcommitter}]
  pub fn try_commit(
      &self,
      last_decided: BlockReference,
      async_state: &AsyncState,
  ) -> Vec<LeaderStatus> {
      let mut leaders = VecDeque::new();
      for round in (last_decided_round..=highest_known_round).rev() {
          let is_async = (round - async_state.last_round) % async_state.interval == 0;

          for committer in self.committers.iter().rev() {
              let mut status = committer.try_direct_decide(leader, round, is_async);

              if !status.is_decided() {
                  status = committer.try_indirect_decide(
                      leader, round, leaders.iter(), is_async
                  );
              }
              leaders.push_front(status);
          }
      }
      // ... filter and return decided leaders
  }

  pub struct UniversalCommitterBuilder {
      wave_length_async: RoundNumber,
      // ... other fields
  }

  impl UniversalCommitterBuilder {
      pub fn with_async_wave_length(mut self, wave_length_async: RoundNumber) -> Self {
          self.wave_length_async = wave_length_async;
          self
      }
  }
\end{lstlisting}
\subsubsection{BaseCommitter}
\label{Appendix:BaseCommitter}
\quad Similarly to how we have modified the Universal Committer structure, Listing~\ref{listing:basecommitter} shows the primary modifications made to the Mysticeti \texttt{BaseCommitter} structure \cite{MYSTICETI}. Both \texttt{try\_direct\_decide} and \texttt{try\_indirect\_decide} have a new \texttt{is\_async} boolean that is set to true if the current \texttt{leader\_round} should be run under the asynchronous mode. If it is, then the \texttt{decision\_round} for this \texttt{leader\_round} is calculated differently, via \texttt{decision\_round\_async}. Furtermore, depending on mode of operation,  \texttt{try\_indirect\_decide}'s threshold for anchor's round also changes depedning on the \texttt{is\_async} flag.
  \begin{lstlisting}[language=Rust, caption={BaseCommitter async decision logic}, label={listing:basecommitter}]
  fn decision_round_async(&self, wave: WaveNumber) -> RoundNumber {
      wave * self.options.wave_length + self.options.wave_length_async - 1
          + self.options.round_offset}

  pub fn try_direct_decide(
      &self,
      leader: AuthorityIndex,
      leader_round: RoundNumber,
      is_async: bool,
  ) -> LeaderStatus {
      let wave_length = if is_async {
          self.options.wave_length_async
      } else {
          self.options.wave_length
      };

      let decision_round = if is_async {
          self.decision_round_async(wave)
      } else {
          self.decision_round(wave)
      };
      // ... rest of logic
  }

  pub fn try_indirect_decide<'a>(
      &self,
      leader: AuthorityIndex,
      leader_round: RoundNumber,
      leaders: impl Iterator<Item = &'a LeaderStatus>,
      is_async: bool,
  ) -> LeaderStatus {
      let anchor_round = if is_async {
          leader_round + self.options.wave_length_async
      } else {
          leader_round + self.options.wave_length
      };
      // ... rest of decision logic
  }
  \end{lstlisting}

\subsection{Dual-Mode Scheduler}
\subsubsection{Asynchronous State}
\label{appendix:AsyncState}
\quad Listing~\ref{listing:asyncstate} shows the structure that keeps of the current state of the asynchronous round. The \texttt{interval} represents the current interval between asynchronous rounds, while \texttt{last\_round} represents the last round which an asynchronous round was executed. Furthermore, hypothetically this structure can be included in the write-ahead-log (WAL) that our protocol utilises, allowing crashed nodes to recover their session correctly.
\newpage
\begin{lstlisting}[caption={Asynchronous state structure},label={listing:asyncstate}]
#[derive(Clone, Debug)]
pub struct AsyncState {
    pub interval: RoundNumber,
    pub last_round: RoundNumber,
}

impl Default for AsyncState {
    fn default() -> Self {
        Self {
            interval: DEFAULT_ASYNC_RATE,
            last_round: 0,
        }
    }
}

impl AsyncState {
    pub fn is_async_round(&self, round: RoundNumber) -> bool {
        // If genesis round OR current round is still last_async_round
        if round == 0 || round <= self.last_round {
            return false;
        }

        (round - self.last_round) % self.interval == 0
    }
}

\end{lstlisting}
\subsubsection{Dual Mode Configuration}
\label{appendix:DualModeConfig}
\quad Listing~\ref{listing:dualmodeconfig} represents static configuration values that allow for ease of modularity regarding the aggressivity of the changes made by the dual mode scheduler. The following variables can be changed:
 \begin{itemize}
  \item \texttt{wave\_length}: Length of partially-synchronous waves in rounds.
  \item \texttt{leaders\_per\_round}: Number of leaders elected per consensus round.
  \item \texttt{pipeline}: Enables/disables pipelined consensus execution. This is usually passed by \texttt{core.rs}, similarly to how \texttt{UniversalCommitter} is built.
  \item \texttt{target\_direct\_commit\_ratio}: Desired fraction of direct commits (0.0-1.0).
  \item \texttt{adjustment\_factor}: Rate of interval adaptation changes. Default set to $0.1$, meaning that the \texttt{interval} value may only change by $10\%$ at each schedule change. 
  \item \texttt{min\_async\_interval}: Lower bound for asynchronous timing intervals.
  \item \texttt{max\_async\_interval}: Upper bound for asynchronous timing intervals.
  \end{itemize}
\vspace{2em}
\begin{lstlisting}[caption={Dual mode configuration structure},label={listing:dualmodeconfig}]
   /// Dual Mode Scheduler configurations
#[derive(Clone, Debug)]
pub struct DualModeConfig {
    pub wave_length: RoundNumber,
    pub leaders_per_round: u64,
    pub pipeline: bool,
    pub target_direct_commit_ratio: f64,
    pub adjustment_factor: f64,
    pub min_async_interval: u64,
    pub max_async_interval: u64,
}

impl Default for DualModeConfig {
    fn default() -> Self {
        Self {
            wave_length: DEFAULT_WAVE_LENGTH,
            leaders_per_round: 1,
            pipeline: false,
            // Adjustment Values
            target_direct_commit_ratio: 0.8,
            adjustment_factor: 0.1,
            // Intervals
            min_async_interval: 100,
            max_async_interval: 900,
        }
    }
}
\end{lstlisting}
\newpage
\subsubsection{Dual Mode Scheduler}
\label{appendix:DualModeScheduler}
\paragraph{DualModeScheduler} Listing~\ref{listing:strucDualModeScheduler} shows the implementation of \texttt{DualModeScheduler}. Specifically, the scheduler requires Creates a \texttt{DualModeConfig}, the \texttt{committee} used by the protocol and the current, or new, \texttt{AsyncState}. Furthermore, the scheduler creates a vector to store leaders slots marked as \texttt{to-commit}, a \texttt{sequence\_length} which stores the length of the decided leader slots (decided leader slots include either slots marked as \texttt{to-commit} or \texttt{to-skip}. If the protocol is run with the \texttt{"simulator"} flag, then we also keep track of all updates intervals to make sure that different validators have an identical history (schedule agreement).  
\begin{lstlisting}[caption={Structure and implementation of \texttt{DualModeScheduler}},label={listing:strucDualModeScheduler}]
pub struct DualModeScheduler {
    async_state: AsyncState,
    #[cfg(feature = "simulator")]
    async_rate_history: Vec<u64>,
    committee: Arc<Committee>,
    committed_leaders: Vec<BlockReference>,
    sequence_length: usize,
    config: DualModeConfig,
}
impl DualModeScheduler {
    /// Creates a new dual mode scheduler with the given configuration
    pub fn new(config: DualModeConfig, committee: Arc<Committee>, async_state: AsyncState) -> Self {
        Self {
            config,
            committee,
            // Extra
            #[cfg(feature = "simulator")]
            async_rate_history: vec![async_state.interval],
            async_state,
            sequence_length: 0,
            committed_leaders: Vec::new(),
        }
    }
\end{lstlisting}
\paragraph{IsAsync} Listing~\ref{listing:isasyncfunction} shows a simple function \texttt{is\_async\_round}, where it takes a \texttt{round}, and given the current \texttt{async\_state}, it returns \texttt{true} if the current round is supposed to be run with the asynchronous mode, or \texttt{false} if not.
\vspace{2em}
\begin{lstlisting}[caption={\texttt{is\_async\_round} function of \texttt{DualModeScheduler}},label={listing:isasyncfunction}]
    pub fn is_async_round(&self, round: RoundNumber) -> bool {
        // If genesis round OR current round is still last_async_round
        if round == 0 || round <= self.async_state.last_round {
            return false;
        }

        (round - self.async_state.last_round) % self.async_state.interval == 0
    }

\end{lstlisting}
\paragraph{CollectSubdag} Listing~\ref{listing:collectsubdag} shows the function \texttt{collect\_subdag}, that, given an \texttt{anchor}, we collect all the blocks in its subDAG, until the round right after the last asynchronous round (stored in the \texttt{AsyncState} structure).
\begin{lstlisting}[caption={\texttt{collect\_subdag} function of \texttt{DualModeScheduler}},label={listing:collectsubdag}]
    fn collect_subdag(&self,block_store: &BlockStore,
        anchor: Data<StatementBlock>,) -> Vec<Data<StatementBlock>> {
        
        let until = self.async_state.last_round + 1;
        let mut visited = HashSet::new();
        let mut to_visit = vec![anchor];
        let mut subdag_blocks = Vec::new();

        while let Some(block) = to_visit.pop() {
            if visited.insert(*block.reference()) {
                subdag_blocks.push(block.clone());
                for include_ref in block.includes() {
                    if include_ref.round() >= (until) {
                        if let Some(included_block) = block_store.get_block(*include_ref) {
                            if !visited.contains(include_ref) {
                                to_visit.push(included_block);
                            }
                        }
                    }
                }
            }
        }
        subdag_blocks
    }
\end{lstlisting}
\paragraph{Update Dual-Mode Schedule} Listing~\ref{listing:updatedualmodeschedule} shows the function  \texttt{update\_dualmode\_schedule}, which is what is called by \texttt{core.rs} when an asynchronous leader slot is marked as \texttt{to-commit}. First of all, we pop the first leader from the vector \texttt{committed\_leaders}, and then we first run \texttt{collect\_subdag}, using the popped leader, then we count the number of direct commits found within this subdag via the \texttt{count\_direct\_commits} function. Then, we get a ratio of direct commits against expected commits, which will allow us to calculate the new interval via the \texttt{calculate\_new\_async\_interval)} function. Lastly, we update the \texttt{AsyncStruct} accordingly, and set the \texttt{sequence\_length} to 0. If the function is run with the \texttt{"simulator"} flag (i.e., we are running a simulation), we store the new interval in the \texttt{async\_rate\_history} vector for later analysis. 
\begin{lstlisting}[caption={\texttt{update\_dualmode\_schedule} function of \texttt{DualModeScheduler}},label={listing:updatedualmodeschedule}]
    pub fn update_dualmode_schedule(&mut self, block_store: &BlockStore) {
        let Some(anchor_ref) = self.committed_leaders.pop() else {
            return;
        };
        let anchor = block_store
            .get_block(anchor_ref)
            .expect("Anchor block must exist in block store");
        let blocks = self.collect_subdag(block_store, anchor.clone());
        let ratio = self.count_direct_commits(&anchor, &blocks, block_store) as f64
            / self.calculate_expected_commits(anchor.round()) as f64;
        let new_interval = self.calculate_new_async_interval(ratio);
        self.async_state.interval = new_interval;
        self.async_state.last_round = anchor.round();
        self.sequence_length = 0;

        #[cfg(feature = "simulator")]
        self.async_rate_history.push(new_interval);
    }
\end{lstlisting}
\paragraph{Expected Commits} Listing~\ref{listing:expectedcommits} shows the \texttt{calculate\_expected\_commits\_sequence} function. This function takes the decided leaders length, obtained by the \texttt{try\_commit} in \texttt{core.rs}, and returns the correct amount. As explained in Section~\ref{chap:dual_mode_protocol}, some leader slots cannot be determined (such as the two previous rounds from the anchor selected), thus, depending on whether pipeline is turned on or off, and alongside the amount of leaders per round, we need to remove some values from the \texttt{sequence\_length}.
\begin{lstlisting}[caption={\texttt{calculate\_expected\_commits\_sequence} function of \texttt{DualModeScheduler}},label={listing:expectedcommits}]
    fn calculate_expected_commits_sequence(&self) -> usize {
        if self.config.pipeline {
            self.sequence_length.saturating_sub(
                self.config.wave_length as usize * self.config.leaders_per_round as usize,
            )
        } else {
            self.sequence_length
                .saturating_sub(self.config.leaders_per_round as usize)
        }
    }
\end{lstlisting}
\paragraph{Direct Commits} Listing~\ref{listing:countdirectcommits} and Listing~\ref{listing:isdirectlycommitted} show the \texttt{calculate\_expected\_commits\_sequence} and \texttt{is\_directly\_committed}. \texttt{calculate\_expected\_commits\_sequence}, as the name suggests, loops through all the leader slots in \texttt{committed\_leaders}, alongside empting the vector at the same time, and then uses \texttt{is\_directly\_committed} to check if the current leader slot has been directly committed, given a \texttt{subdag}. Furthermore, \texttt{is\_directly\_committed} utilises the \texttt{committee} vector to determine the Stake of a given validator, making sure that validators with higher stake in the system will have the appropriate weight applied to the vote. 
\begin{lstlisting}[caption={\texttt{calculate\_expected\_commits\_sequence} function of \texttt{DualModeScheduler}},label={listing:countdirectcommits}]
    fn count_direct_commits(
        &mut self,
        anchor: &Data<StatementBlock>,
        blocks: &Vec<Data<StatementBlock>>,
        block_store: &BlockStore,
    ) -> usize {
        let leader_refs: Vec<_> = self.committed_leaders.drain(..).collect();
        leader_refs
            .iter()
            .filter_map(|leader_ref| block_store.get_block(*leader_ref))
            .filter(|leader| leader.round() <= anchor.round() - self.config.wave_length)
            .filter(|leader| self.is_directly_committed(leader, blocks))
            .count()
    }

\end{lstlisting}
\vspace{2em}
\begin{lstlisting}[caption={\texttt{is\_directly\_committed} function of \texttt{DualModeScheduler}},label={listing:isdirectlycommitted}]
    fn is_directly_committed(
        &self,
        leader: &Data<StatementBlock>,
        blocks: &[Data<StatementBlock>],
    ) -> bool {
        let vote_round = leader.round() + 1;
        let cert_round = vote_round + 1;
        let leader_ref = leader.reference();

        // Collect votes for the leader in round r+1
        let votes: Vec<_> = blocks
            .iter()
            .filter(|block| block.round() == vote_round)
            .filter(|block| block.includes().contains(leader_ref))
            .collect();

        // In round r+2, use StakeAggregator to count certificates with proper stake weighting
        let mut certificate_stake_aggregator = StakeAggregator::<QuorumThreshold>::new();
        for block in blocks.iter().filter(|block| block.round() == cert_round) {
            let mut vote_stake_aggregator = StakeAggregator::<QuorumThreshold>::new();
            let has_enough_votes = votes
                .iter()
                .filter(|vote| block.includes().contains(vote.reference()))
                .any(|vote| vote_stake_aggregator.add(vote.reference().authority, &self.committee));
            if has_enough_votes
                && certificate_stake_aggregator.add(block.reference().authority, &self.committee)
            {
                return true;
            }
        }
        false
    }
\end{lstlisting}
\vspace{2em}
\paragraph{Calculate New Asynchronous Rate} Listing~\ref{listing:calculatenewasyncinterval} shows the function \texttt{calculate\_new\_async\_interval}. This is a deterministic function that, given a \texttt{commit\_ration}, will increase or decrease the \texttt{interval} stored in \texttt{async\_state}. Furthermore, \texttt{min\_async\_interval} and \texttt{max\_async\_interval} are used to make sure that the asynchronous\texttt{interval} does not go below or above a certain threshold.  
\begin{lstlisting}[caption={\texttt{calculate\_new\_async\_interval} function of \texttt{DualModeScheduler}},label={listing:calculatenewasyncinterval}]
\begin{lstlisting}
    fn calculate_new_async_interval(&self, commit_ratio: f64) -> u64 {
        let current_interval = self.async_state.interval;
        let ratio_diff = commit_ratio - self.config.target_direct_commit_ratio;

        let mut adjustment = if ratio_diff >= 0.0 {
            // Good amount of direct commits
            (current_interval as f64 / (1.0 - self.config.adjustment_factor)) as u64
        } else {
            // Too few direct commits
            (current_interval as f64 * (1.0 - self.config.adjustment_factor)) as u64
        };

        // Ensure async_round is a multiple of wave_length if pipeline is disabled
        if !self.config.pipeline {
            adjustment = (adjustment + 1) / self.config.wave_length * self.config.wave_length;
        }

        adjustment.clamp(
            self.config.min_async_interval,
            self.config.max_async_interval,
        )
    }
\end{lstlisting}
\paragraph{Add Committed Leaders} Listing~\ref{listing:addcommittedleaders} shows the function \texttt{add\_committed\_leaders}. A simple function always called by the \texttt{try\_commit} function, in \texttt{core.rs}, that adds leader slots marked as \texttt{to-commit} to the \texttt{committed\_leaders}, alongside the length of decided leaders to \texttt{sequence\_length}.
\vspace{2em}
\begin{lstlisting}[caption={\texttt{add\_committed\_leaders} function of \texttt{DualModeScheduler}},label={listing:addcommittedleaders}]
    pub fn add_committed_leaders(
        &mut self,
        leaders: Vec<Data<StatementBlock>>,
        sequence_length: usize,
    ) {
        self.sequence_length += sequence_length;
        let leader_refs: Vec<BlockReference> =
            leaders.iter().map(|block| *block.reference()).collect();
        self.committed_leaders.extend(leader_refs);
    }
\end{lstlisting}
\vspace{2em}
\paragraph{Asynchronous Rate History} Listing~\ref{listing:getasyncrate} shows the function \texttt{get\_async\_rate\_history}. This function can only be run if the flag \texttt{"simulator"} is used. It is used by the simulator tests, found in \texttt{net\_sync.rs}, to make sure that all validators have the same asynchronous interval history.
\begin{lstlisting}[caption={\texttt{get\_async\_rate\_history} function of \texttt{DualModeScheduler}},label={listing:getasyncrate}]
    #[cfg(feature = "simulator")]
    pub fn get_async_rate_history(&self) -> Vec<u64> {
        self.async_rate_history.clone()
    }
}
\end{lstlisting}
\subsubsection{Safety Mechanism}
\label{appendix:safetymechanism}
\quad Listing~\ref{listing:try_commit_core} represent how \texttt{core.rs} (the core of the validator that runs the decider rule every time it receives a new block) handles unlikely long periods of asynchrony. In such periods, as stated in Section~\ref{subsection:dynamic_scheduling}, it is possible for a validator to go for a long time without a direct commit, causing a long list of  undecided leaders to begin stacking. It is possible, once a leader has been committed, that we go through two distinct asynchronous rounds (based on our \texttt{AsyncState}), which could cause a desynchronization of the asynchronous schedule between the validators. Thus, as we can see from Listing~\ref{listing:try_commit_core}, we correct this by making sure that, once we run the decision rule and we get a sequence of leaders, we truncate the sequence at the first found asynchronous leader. Then, we check if the top leader is asynchronous, and if it is, we update the schedule accordingly. This ensures that all validators agree on the sequence of asynchronous schedules, as proven by Lemma~\ref{lem:schedule_agreement}.
\vspace{2em}
\begin{lstlisting}[caption={\texttt{core.rs} try\_commit function},label={listing:try_commit_core}]
    pub fn try_commit(&mut self) -> Vec<Data<StatementBlock>> {
        let async_state = self.dual_mode_scheduler.get_async_state();
        let mut sequence: Vec<_> = self
            .committer
            .try_commit(self.last_commit_leader, async_state);
            
        if let Some(pos) = sequence
            .iter()
            .position(|x| self.dual_mode_scheduler.is_async_round(x.round()) && x.is_commit())
        {
            sequence.truncate(pos + 1);
        }
        let sequence_length = sequence.len();

        let sequence: Vec<_> = sequence
            .into_iter()
            .filter_map(|leader| leader.into_decided_block())
            .collect();
            
        self.dual_mode_scheduler
            .add_committed_leaders(sequence.clone(), sequence_length);
            
        if let Some(last) = sequence.last() {
            if self.dual_mode_scheduler.is_async_round(last.round()) {
                self.dual_mode_scheduler
                    .update_dualmode_schedule(&self.block_store);
            }
            self.last_commit_leader = *last.reference();
        }
        if self.last_commit_leader.round() > self.rounds_in_epoch {
            self.epoch_manager.epoch_change_begun();
        }
        sequence
    }
\end{lstlisting}
\section{Tests}
\subsection{Decision Rule}
\label{appendixL:testConsensusProtocol}
\quad We now move onto the tests performed on the decision rule of our consensus protocol. These tests make sure that proposed blocks are correctly marked depending on a given structure of a DAG. To reduce redundancy, we only describe the tools utilised to conduct the tests on the base dual-mode functionality of our protocol. All tests passed successfully across the complete range of testing modes, as explained in Section~\ref{section:tests}, and all testing modes pass successfully when \texttt{DEFAULT\_WAVE\_LENGTH\_ASYNC} is set to either $4$ or $5$. 
\begin{itemize}
  \item \texttt{committee}: Creates test committee with a specified amount of validators.
  \item \texttt{build\_dag}: Create a completely connected DAG up to a specific round. 
  \item \texttt{build\_dag\_layer}: Create a single DAG layer with customised validator connections, and block references. Specifically, this function is utilised to create scenarios where the decision rule marks a slot as either \texttt{to-skip} or \texttt{undecided}.
  \item \texttt{TestBlockWriter::new}: Initializes mock block writer for creating test blockchain structures.
  \item \texttt{test\_metrics}: Provides instrumentation metrics collection for consensus performance testing.
  \item \texttt{UniversalCommitterBuilder::new}: Creates builder for configuring consensus committer with custom parameters.
  \item \texttt{try\_commit}: Executes consensus commit attempt with given last committed block and asynchronous state.
  \item \texttt{elect\_leader}: Deterministically selects leader validator for specified consensus round.
  \item \texttt{BlockReference::new\_test}: Creates test block reference with authority ID and round number.
  \item \texttt{quorum\_threshold}: Returns minimum validator count needed for consensus quorum decisions ($2f+1$).
  \item \texttt{validity\_threshold}: Returns minimum validator count required for block validity certification ($f+1$).
  \end{itemize}
\quad All decision rules related tests can be found in the following files:
\begin{itemize}
    \item \texttt{base\_committer\_tests.rs}
    \item \texttt{multi\_committer\_tests.rs}
    \item \texttt{pipelined\_committer\_tests.rs}
\end{itemize}
\quad For brevity, we present only a representative subset of tests. The following tests create scenarios where \tarpon has to commit the first leader via its partially-synchronous mode, then either commit or skip the asynchronous mode leader (depending on the test), and lastly, commit another partially-synchronous leader.
\subsubsection{Dual-Mode Commit}
\begin{lstlisting}
#[test]
#[tracing_test::traced_test]
fn direct_commit_dual_mode() {
    let committee = committee(4);
    let wave_length = DEFAULT_WAVE_LENGTH;
    let wave_length_async = DEFAULT_WAVE_LENGTH_ASYNC;

    // async round begins after one full wave of mysticeti is completed
    let async_round = wave_length*2;
    let decision_round = async_round + wave_length_async+1;

    let mut block_writer = TestBlockWriter::new(&committee);
    build_dag(&committee, &mut block_writer, None, decision_round);

    let committer = UniversalCommitterBuilder::new(
        committee.clone(),
        block_writer.into_block_store(),
        test_metrics(),
    )
    .build();

    let last_committed = BlockReference::new_test(0, 0);
    let sequence = committer.try_commit(
        last_committed,
        &AsyncState {
            interval: async_round,
            last_round: 0,
        },
    );
    tracing::info!("Commit sequence: {sequence:?}");

    // 3 leaders committed, 1 p-sync, 1 async, then 1 p-sync again
    assert_eq!(sequence.len(), 3);
}
\end{lstlisting}
\subsubsection{Dual-Mode Skip}
\begin{lstlisting}
#[test]
#[tracing_test::traced_test]
fn direct_skip_dual_mode() {
    let committee = committee(4);
    let wave_length = DEFAULT_WAVE_LENGTH;
    let wave_length_async = DEFAULT_WAVE_LENGTH_ASYNC;

    let async_round = wave_length*2;

    let mut block_writer = TestBlockWriter::new(&committee);
    let references_1 = build_dag(&committee, &mut block_writer, None, async_round);

    // Filter out leader of async round.
    let references_without_leader_1: Vec<_> = references_1
        .into_iter()
        .filter(|x| x.authority != committee.elect_leader(async_round))
        .collect();

    let decision_round_1 = async_round + wave_length_async + 1;
    build_dag(
        &committee,
        &mut block_writer,
        Some(references_without_leader_1),
        decision_round_1,
    );

    let committer = UniversalCommitterBuilder::new(
        committee.clone(),
        block_writer.into_block_store(),
        test_metrics(),
    )
    .build();

    let last_committed = BlockReference::new_test(0, 0);
    let sequence = committer.try_commit(
        last_committed,
        &AsyncState {
            interval: async_round,
            last_round: 0,
        },
    );
    tracing::info!("Commit sequence: {sequence:?}");

    // 3 leaders decided, 1 p-sync - commit, 1 async - skip, then 1 p-sync again - commit
    assert_eq!(sequence.len(), 3);
    if let LeaderStatus::Skip(leader, round) = sequence[1] {
        assert_eq!(leader, committee.elect_leader(async_round));
        assert_eq!(round, async_round);
    } else {
        panic!("Expected to directly skip the asynchronous leader");
    }
}
\end{lstlisting}
\subsection{Dual-Mode Scheduler}
\quad Beyond simple unit tests that check edge cases of simple calculations function (e.g., checking that the \texttt{is\_async\_rate(anchor.round)} returns the correct value, either \texttt{true} if the anchor's round uses the asynchronous mode, or \texttt{false} otherwise), the important tests is making sure that our \texttt{count\_direct\_commits()} and \texttt{collect\_subdag()} work correctly. In the case of \texttt{count\_direct\_commits()}, this meant making sure that the function works with pipeline turned on, and also cases with multi-leaders. Emphasis is placed on these two functions, as they are what determine the deterministic change in schedule consistent across all validators. All tests mentioned in this section can be found in \texttt{dual\_mode\_scheduler.rs}.
\subsubsection{Collect Subdag}
\label{appendix:collect_subdag}
\quad To test the collecting function of \tarpon, a specific test has been created. Specifically, this test creates a DAG where only one leader has includes all the blocks from the previous round, causing that one leader to have a longer subDAG then the others. This test can be found in \texttt{dual\_mode\_scheduler.rs}.
\subsubsection{Count Direct Commits}
\label{appendix:count_direct_commits}
\quad We then created tests to makes sure that, in the scenario of a DAG with all direct commits, the counting function returns the correct amount (in this case being 1). This specific behaviour has also been tested in an instance with pipeline set to \texttt{true} and another with a multi-leader setup ($2f+1$ leaders per round).
\subsubsection{Count Direct Commits (Indirect)}
\label{appendix:count_indirect_commits}
\quad In the scenario of a DAG with a indirect commits, the counting function has to return the correct amount (in this case being 0). This test is important as it shows that we are able to distinguish direct commits from indirect commits within a subDAG. The specific behaviour has also been tested in an instance with pipeline set to \texttt{true} and another with a multi-leader setup ($2f+1$ leaders per round). 
\subsubsection{Sequence Cutoff}
\label{appendix:seqoff}
\quad Lastly we have created a test that makes sure a long enough sequence of leaders (as mentioned in Section~\ref{subsection:dynamic_scheduling}), is correctly cut off. The test creates enough blocks (given an initial interval rate of every $300$ rounds, our tests creates $300*5$ blocks) which guarantees the scenario we are looking for. We then run \texttt{try\_commit} in a loop until no more leaders can be committed. Each loop should either stop at the correct async\_rate\_round (and update the \texttt{async\_rate\_round} accordingly) or until the sequence is empty. This test can be found in This test can be found in \texttt{core.rs}.
\section{Simulations}
\label{appendix:simulation}
\subsubsection{Good Network Simulations}
\label{appendix:fastNetwork}
\quad We will now go over the simulation tests that were run to test \tarpon's dynamic scheduler. We first run \tarpon under its ideal conditions (i.e., all validators reach their leader slot within the timeout). After running the simulation, we make sure that each validators have an increasing interval for the asynchronous round (i.e., the frequency of the asynchronous mode decreases), and that each validator has the same asynchronous interval history. This test can be found in \texttt{net\_sync.rs}.
\subsubsection{Slow Network Simulations}
\label{appendix:slowNetwork}
\quad We also ran \tarpon under a high-latency simulated network with $10$ validators. In these simulations, $f$ validators experience high-level of latency, causing them to miss their leader slot, subsequentially making every other validator mark the slot as \texttt{to-skip}. Specifically, the function \texttt{simulated\_network\_connect.all\_high\_latency()} connects some validators with low-latency, and some validators with high-latency After running the simulation, we make sure that each validators have a decreasing interval for the asynchronous round (i.e., the frequency of the asynchronous mode increases), and that each validator has the same asynchronous interval history, including the slow validators. This test can be found in \texttt{net\_sync.rs}.
\chapter{Consensus Protocol Structures}
\label{Appendix:ConsProtStruct}
\subsubsection{PBFT}
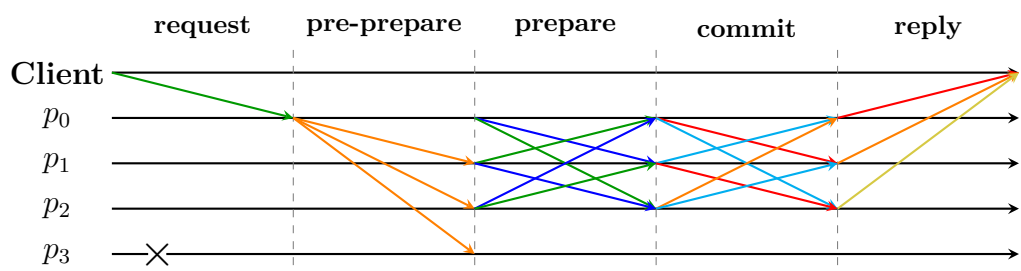
\begin{figure}[H]
    \centering
    \begin{tikzpicture}[
scale=0.6,
    >=stealth,
]

\node[font=\bfseries] at (2, 4.5) {\footnotesize request};
\node[font=\bfseries] at (6, 4.5) {\footnotesize pre-prepare};
\node[font=\bfseries] at (10, 4.5) {\footnotesize prepare};
\node[font=\bfseries] at (14, 4.5) {\footnotesize commit};
\node[font=\bfseries] at (18, 4.5) {\footnotesize reply};

\node[font=\bfseries] at (-1.2, 3.5) {Client};
\node[font=\bfseries] at (-1.2, 2.5) {$p_0$};
\node[font=\bfseries] at (-1.2, 1.5) {$p_1$};
\node[font=\bfseries] at (-1.2, 0.5) {$p_2$};
\node[font=\bfseries] at (-1.2, -0.5) {$p_3$};

\draw[->, thick] (0, 3.5) -- (20, 3.5);
\draw[->, thick] (0, 2.5) -- (20, 2.5);
\draw[->, thick] (0, 1.5) -- (20, 1.5);
\draw[->, thick] (0, 0.5) -- (20, 0.5);

\draw[thick] (0, -0.5) -- (0.8, -0.5);
\node[font=\Large] at (1, -0.5) {$\times$};
\draw[->, thick] (1.2, -0.5) -- (20, -0.5);

\foreach \x in {4, 8, 12, 16} {
    \draw[dashed, gray] (\x, 4) -- (\x, -1);
}

\coordinate (c2) at (0, 3.5);
\coordinate (n02) at (4, 2.5);
\coordinate (n12) at (4, 1.5);
\coordinate (n22) at (4, 0.5);
\coordinate (n32) at (4, -0.5);

\coordinate (c4) at (8, 3.5);
\coordinate (n04) at (8, 2.5);
\coordinate (n14) at (8, 1.5);
\coordinate (n24) at (8, 0.5);
\coordinate (n34) at (8, -0.5);

\coordinate (c6) at (12, 3.5);
\coordinate (n06) at (12, 2.5);
\coordinate (n16) at (12, 1.5);
\coordinate (n26) at (12, 0.5);
\coordinate (n36) at (12, -0.5);

\coordinate (c8) at (16, 3.5);
\coordinate (n08) at (16, 2.5);
\coordinate (n18) at (16, 1.5);
\coordinate (n28) at (16, 0.5);
\coordinate (n38) at (16, -0.5);

\coordinate (c10) at (20, 3.5);
\coordinate (n010) at (20, 2.5);
\coordinate (n110) at (20, 1.5);
\coordinate (n210) at (20, 0.5);
\coordinate (n310) at (20, -0.5);

\draw[->, green!60!black, thick] (c2) -- (n02);

\draw[->, orange, thick] (n02) -- (n14);
\draw[->, orange, thick] (n02) -- (n24);
\draw[->, orange, thick] (n02) -- (n34);

\draw[->, blue, thick] (n04) -- (n16);
\draw[->, green!60!black, thick] (n04) -- (n26);
\draw[->, green!60!black, thick] (n14) -- (n06);
\draw[->, blue, thick] (n14) -- (n26);
\draw[->, blue, thick] (n24) -- (n06);
\draw[->, green!60!black, thick] (n24) -- (n16);

\draw[->, red, thick] (n06) -- (n18);
\draw[->, cyan, thick] (n06) -- (n28);
\draw[->, cyan, thick] (n16) -- (n08);
\draw[->, red, thick] (n16) -- (n28);
\draw[->, orange, thick] (n26) -- (n08);
\draw[->, cyan, thick] (n26) -- (n18);

\draw[->, red, thick] (n08) -- (c10);
\draw[->, orange, thick] (n18) -- (c10);
\draw[->, yellow!80!black, thick] (n28) -- (c10);

\end{tikzpicture}
    \caption{Depiction of how consensus is reached in PBFT \cite{PBFT}. Process $p_3$ is a Byzantine/faulty node, thus, it does not communicate with the rest.}
    \label{fig:PBFT}
\end{figure}
\vspace{5em}
\subsubsection{Hotstuff}
\begin{figure}[H]
    \centering
    \begin{tikzpicture}[
    >=stealth,
    scale=.6
]
\node[font=\bfseries] at (2, 4.5) {\footnotesize request};
\node[font=\bfseries] at (6, 4.5) {\footnotesize prepare};
\node[font=\bfseries] at (10, 4.5) {\footnotesize pre-commit};
\node[font=\bfseries] at (14, 4.5) {\footnotesize commit};
\node[font=\bfseries] at (18, 4.5) {\footnotesize decide};
\node[font=\bfseries] at (-1.2, 3.5) {Client};
\node[font=\bfseries] at (-1.2, 2.5) {$L$};
\node[font=\bfseries] at (-1.2, 1.5) {$p_1$};
\node[font=\bfseries] at (-1.2, 0.5) {$p_2$};
\node[font=\bfseries] at (-1.2, -0.5) {$p_3$};
\draw[->, thick] (0, 3.5) -- (20, 3.5);
\draw[->, thick] (0, 2.5) -- (20, 2.5);
\draw[->, thick] (0, 1.5) -- (20, 1.5);
\draw[->, thick] (0, 0.5) -- (20, 0.5);
\draw[thick] (0, -0.5) -- (0.8, -0.5);
\node[font=\Large] at (1, -0.5) {$\times$};
\draw[->, thick] (1.2, -0.5) -- (20, -0.5);
\foreach \x in {4, 8, 12, 16} {
    \draw[dashed, gray] (\x, 4) -- (\x, -1);
}
\foreach \x in {6, 10, 14} {
    \draw[dotted, gray!50] (\x, 4) -- (\x, -1);
}

\coordinate (c2) at (0, 3.5);
\coordinate (L4) at (4, 2.5);

\coordinate (L5) at (6, 2.5);
\coordinate (n15) at (6, 1.5);
\coordinate (n25) at (6, 0.5);
\coordinate (n35) at (6, -0.5);
\coordinate (L7) at (8, 2.5);
\coordinate (n17) at (6, 1.5);
\coordinate (n27) at (6, 0.5);

\coordinate (L9) at (8, 2.5);
\coordinate (n19) at (10, 1.5);
\coordinate (n29) at (10, 0.5);
\coordinate (n39) at (10, -0.5);
\coordinate (L11) at (12, 2.5);
\coordinate (n111) at (10, 1.5);
\coordinate (n211) at (10, 0.5);

\coordinate (L13) at (12, 2.5);
\coordinate (n113) at (14, 1.5);
\coordinate (n213) at (14, 0.5);
\coordinate (n313) at (14, -0.5);
\coordinate (L15) at (16, 2.5);
\coordinate (n115) at (14, 1.5);
\coordinate (n215) at (14, 0.5);

\coordinate (c20) at (20, 3.5);
\coordinate (L17) at (16, 2.5);
\coordinate (n117) at (18, 1.5);
\coordinate (n217) at (18, 0.5);
\coordinate (n317) at (18, -0.5);
\coordinate (L20) at (20, 2.5);
\coordinate (n120) at (18, 1.5);
\coordinate (n220) at (18, 0.5);
\coordinate (n320) at (18, -0.5);

\draw[->, green!60!black, thick] (c2) -- (L4);

\draw[->, orange, thick] (L4) -- (n15);
\draw[->, orange, thick] (L4) -- (n25);
\draw[->, orange, thick] (L4) -- (n35);
\draw[->, blue, thick] (n17) -- (L7);
\draw[->, blue, thick] (n27) -- (L7);

\draw[->, red, thick] (L7) -- (n19);
\draw[->, red, thick] (L7) -- (n29);
\draw[->, red, thick] (L7) -- (n39);
\draw[->, cyan, thick] (n111) -- (L11);
\draw[->, cyan, thick] (n211) -- (L11);

\draw[->, green!60!black, thick] (L11) -- (n113);
\draw[->, green!60!black, thick] (L11) -- (n213);
\draw[->, green!60!black, thick] (L11) -- (n313);
\draw[->, purple, thick] (n115) -- (L15);
\draw[->, purple, thick] (n215) -- (L15);

\draw[->, yellow!80!black, thick] (L15) -- (n117);
\draw[->, yellow!80!black, thick] (L15) -- (n217);
\draw[->, yellow!80!black, thick] (L15) -- (n317);
\draw[->, red, thick] (L17) -- (c20);
\draw[->, yellow!80!black, thick] (n120) -- (c20);
\draw[->, yellow!80!black, thick] (n220) -- (c20);

\end{tikzpicture}
    \caption{Depiction of how consensus is reached in Hotstuff \cite{HOTSTUFF}. Process $p_3$ is a Byzantine/faulty node, thus, it does not communicate with the rest.}
    \label{fig:HotStuff}
\end{figure}
\subsubsection{Tusk}
\begin{figure}[H]
    \centering
    \begin{tikzpicture}[
    node distance=1.5cm and 2.5cm,
    validator/.style={rectangle, draw, minimum size=0.8cm},
    propose/.style={validator, thick},
    boost/.style={validator, draw=black},
    vote/.style={validator, draw=orange, thick},
    certify/.style={validator, draw=green!60!black, thick},
    supporting_validator/.style={validator, fill=green!20},
    support/.style={->, green!60!black, thick},
    >=Stealth,
    first/.style={validator, fill=green!20},
    second/.style={validator, fill=blue!20},
    vote_support/.style={->, green!60!black, thick},
    skipping_validator/.style={validator, fill=pink!40},
    cert_support/.style={->, blue, thick},
    >=Stealth
]
\footnotesize
\node at (0, 5) {$r$};
\node at (2.5, 5) {$r+1$};
\node at (5, 5) {$r+2$};
\node at (7.5, 5) {$r+3$};
\node at (10, 5) {$r+4$};

\node at (-1.5, 4) {$v_0$};
\node at (-1.5, 3) {$v_1$};
\node at (-1.5, 2) {$v_2$};
\node at (-1.5, 1) {$v_3$};

\node[boost] (v0r) at (0, 4) {};
\node[boost] (v1r) at (0, 3) {};
\node[boost] (v2r) at (0, 2) {};
\node[boost, first] (v3r) at (0, 1) {$L_0$};

\node[boost] (v0r1) at (2.5, 4) {};
\node[boost] (v1r1) at (2.5, 3) {};
\node[boost] (v2r1) at (2.5, 2) {};
\node[boost, first] (v3r1) at (2.5, 1) {};

\node[boost] (v0r2) at (5, 4) {};
\node[boost, second] (v1r2) at (5, 3) {$L_1$};
\node[boost] (v2r2) at (5, 2) {};
\node[boost] (v3r2) at (5, 1) {};

\node[boost, second] (v0r3) at (7.5, 4) {};
\node[boost, second] (v1r3) at (7.5, 3) {};
\node[boost, second] (v2r3) at (7.5, 2) {};
\node[boost, second] (v3r3) at (7.5, 1) {};

\node[boost] (v0r4) at (10, 4) {};
\node[boost] (v1r4) at (10, 3) {};
\node[boost] (v2r4) at (10, 2) {};
\node[boost] (v3r4) at (10, 1) {};

\draw[->] (v0r1.west) -- (v0r);
\draw[->] (v0r1.west) -- (v1r);
\draw[->] (v0r1.west) -- (v2r);
\draw[->] (v1r1.west) -- (v0r);
\draw[->] (v1r1.west) -- (v1r);
\draw[->] (v1r1.west) -- (v2r);
\draw[->] (v2r1.west) -- (v0r);
\draw[->] (v2r1.west) -- (v1r);
\draw[->] (v2r1.west) -- (v2r);
\draw[->] (v3r1.west) -- (v1r);
\draw[->] (v3r1.west) -- (v2r);
\draw[vote_support] (v3r1.west) -- (v3r);

\draw[->] (v0r2.west) -- (v0r1);
\draw[->] (v0r2.west) -- (v1r1);
\draw[->] (v0r2.west) -- (v2r1);
\draw[->] (v1r2.west) -- (v0r1);
\draw[->] (v1r2.west) -- (v1r1);
\draw[vote_support] (v1r2.west) -- (v3r1);
\draw[->] (v2r2.west) -- (v0r1);
\draw[->] (v2r2.west) -- (v1r1);
\draw[->] (v2r2.west) -- (v2r1);
\draw[->] (v3r2.west) -- (v1r1);
\draw[->] (v3r2.west) -- (v2r1);
\draw[->] (v3r2.west) -- (v3r1);

\draw[->] (v0r3.west) -- (v0r2);
\draw[cert_support] (v0r3.west) -- (v1r2);
\draw[->] (v0r3.west) -- (v2r2);
\draw[->] (v1r3.west) -- (v0r2);
\draw[cert_support] (v1r3.west) -- (v1r2);
\draw[->] (v1r3.west) -- (v2r2);
\draw[->] (v2r3.west) -- (v0r2);
\draw[cert_support] (v2r3.west) -- (v1r2);
\draw[->] (v2r3.west) -- (v2r2);
\draw[cert_support] (v3r3.west) -- (v1r2);
\draw[->] (v3r3.west) -- (v2r2);
\draw[->] (v3r3.west) -- (v3r2);

\draw[->] (v0r4.west) -- (v0r3);
\draw[->] (v0r4.west) -- (v1r3);
\draw[->] (v0r4.west) -- (v2r3);
\draw[->] (v1r4.west) -- (v0r3);
\draw[->] (v1r4.west) -- (v1r3);
\draw[->] (v1r4.west) -- (v2r3);
\draw[->] (v2r4.west) -- (v0r3);
\draw[->] (v2r4.west) -- (v1r3);
\draw[->] (v2r4.west) -- (v2r3);
\draw[->] (v3r4.west) -- (v1r3);
\draw[->] (v3r4.west) -- (v2r3);
\draw[->] (v3r4.west) -- (v3r3);

\node[draw,green!60!black, dashed, rounded corners, fit=(v0r) (v3r) (v0r2) (v3r2), inner sep=8pt, label={[align=center]below:$w$}] {};
\node[draw, blue!60!black, dashed, rounded corners, fit=(v0r2) (v3r2) (v0r4) (v3r4), inner sep=8pt, label={[align=center]below:$w+1$}] {};

\end{tikzpicture}
    \caption{Illustration of a DAG structure as interpreted by the Tusk protocol \cite{TUSK}. The proposed leader block $L_0$ (by validator $v_3$) does not have enough votes to be committed. The proposed leader block $L_1$ (by validator $v_1$) has $f+1$ votes necessary to be committed, but, since through the causal history of $L_1$ we encounter $L_0$, $L_0$ is committed before $L_1$. Thus the final sequence would be $\{L_0,L_1\}$.}
    \label{fig:TUSK}
\end{figure}
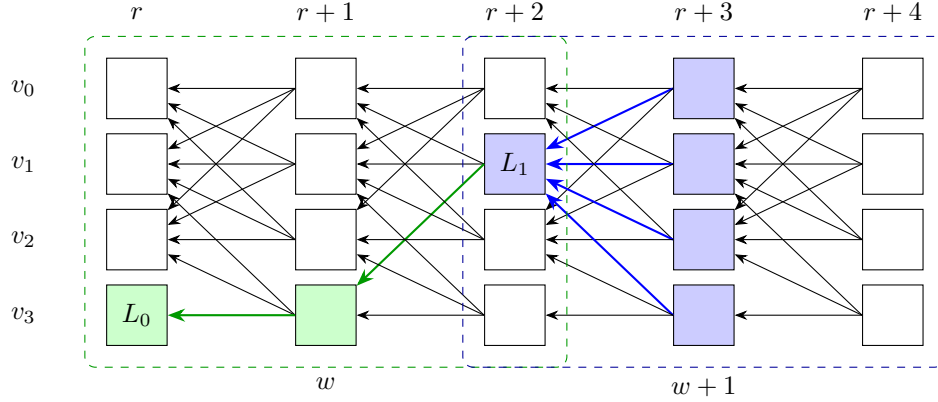
\subsubsection{Bullshark Partially-Synchronous}
\begin{figure} [H]
    \centering
    \begin{tikzpicture}[
    node distance=1.5cm and 2.5cm,
    validator/.style={draw, minimum size=0.8cm},
    propose/.style={validator, thick},
    boost/.style={validator, draw=black},
    vote/.style={validator, draw=orange, thick},
    certify/.style={validator, draw=green!60!black, thick},
    supporting_validator/.style={validator, fill=green!20},
    support/.style={->, blue, thick},
    >=Stealth,
    first/.style={validator, fill=green!20},
    second/.style={validator, fill=blue!20},
    vote_support/.style={->, blue, thick},
    skipping_validator/.style={validator, fill=pink!40},
    cert_support/.style={->, green!60!black, thick},
    >=Stealth
]
\footnotesize
\node at (0, 5) {$r$};
\node at (2.5, 5) {$r+1$};
\node at (5, 5) {$r+2$};
\node at (7.5, 5) {$r+3$};
\node at (10, 5) {$r+4$};
\node at (12.5, 5) {$r+5$};
\node at (-1.5, 4) {$v_0$};
\node at (-1.5, 3) {$v_1$};
\node at (-1.5, 2) {$v_2$};
\node at (-1.5, 1) {$v_3$};
\node[boost] (v0r) at (0, 4) {};
\node[boost, first] (v1r) at (0, 3) {$L_0$};
\node[boost] (v2r) at (0, 2) {};
\node[boost] (v3r) at (0, 1) {};
\node[boost] (v0r1) at (2.5, 4) {};
\node[boost, second] (v1r1) at (2.5, 3) {};
\node[boost] (v2r1) at (2.5, 2) {};
\node[boost] (v0r2) at (5, 4) {};
\node[boost] (v1r2) at (5, 3) {};
\node[boost] (v2r2) at (5, 2) {};
\node[boost, first, dashed] (v3r2) at (5, 1) {$L_1$};
\node[boost] (v0r3) at (7.5, 4) {};
\node[boost] (v1r3) at (7.5, 3) {};
\node[boost] (v2r3) at (7.5, 2) {};
\node[boost] (v0r4) at (10, 4) {};
\node[boost] (v1r4) at (10, 3) {};
\node[boost, first] (v2r4) at (10, 2) {$L_2$};
\node[boost, second] (v0r5) at (12.5, 4) {};
\node[boost, second] (v1r5) at (12.5, 3) {};
\node[boost, second] (v2r5) at (12.5, 2) {};
\node[boost, second] (v3r5) at (12.5, 1) {};
\draw[->] (v0r1.west) -- (v0r);
\draw[->] (v0r1.west) -- (v1r);
\draw[->] (v0r1.west) -- (v2r);
\draw[->] (v1r1.west) -- (v0r);
\draw[support] (v1r1.west) -- (v1r);
\draw[->] (v1r1.west) -- (v2r);
\draw[->] (v2r1.west) -- (v0r);
\draw[->] (v2r1.west) -- (v2r);
\draw[->] (v2r1.west) -- (v3r);
\draw[->] (v0r2.west) -- (v0r1);
\draw[->] (v0r2.west) -- (v1r1);
\draw[->] (v0r2.west) -- (v2r1);
\draw[->] (v1r2.west) -- (v0r1);
\draw[->] (v1r2.west) -- (v1r1);
\draw[->] (v1r2.west) -- (v2r1);
\draw[->] (v2r2.west) -- (v0r1);
\draw[vote_support] (v2r2.west) -- (v1r1);
\draw[->] (v2r2.west) -- (v2r1);
\draw[->] (v3r2.west) -- (v0r1);
\draw[->] (v3r2.west) -- (v1r1);
\draw[->] (v3r2.west) -- (v2r1);
\draw[->] (v0r3.west) -- (v0r2);
\draw[->] (v0r3.west) -- (v1r2);
\draw[->] (v0r3.west) -- (v2r2);
\draw[->] (v1r3.west) -- (v0r2);
\draw[->] (v1r3.west) -- (v1r2);
\draw[->] (v1r3.west) -- (v2r2);
\draw[->] (v2r3.west) -- (v0r2);
\draw[->] (v2r3.west) -- (v1r2);
\draw[vote_support] (v2r3.west) -- (v2r2);
\draw[->] (v0r4.west) -- (v0r3);
\draw[->] (v0r4.west) -- (v1r3);
\draw[->] (v0r4.west) -- (v2r3);
\draw[->] (v1r4.west) -- (v0r3);
\draw[->] (v1r4.west) -- (v1r3);
\draw[->] (v1r4.west) -- (v2r3);
\draw[->] (v2r4.west) -- (v0r3);
\draw[->] (v2r4.west) -- (v1r3);
\draw[vote_support] (v2r4.west) -- (v2r3);
\draw[->] (v0r5.west) -- (v0r4);
\draw[->] (v0r5.west) -- (v1r4);
\draw[support] (v0r5.west) -- (v2r4);
\draw[->] (v1r5.west) -- (v0r4);
\draw[->] (v1r5.west) -- (v1r4);
\draw[support] (v1r5.west) -- (v2r4);
\draw[->] (v2r5.west) -- (v0r4);
\draw[->] (v2r5.west) -- (v1r4);
\draw[support] (v2r5.west) -- (v2r4);
\draw[->] (v3r5.west) -- (v0r4);
\draw[->] (v3r5.west) -- (v1r4);
\draw[support] (v3r5.west) -- (v2r4);
\node[draw, dashed, rounded corners, fit=(v0r) (v3r) (v0r1) (v2r1), inner sep=8pt, label={[align=center]below:$w$}] {};
\node[draw, dashed, rounded corners, fit=(v0r2) (v3r2) (v0r3) (v2r3), inner sep=8pt, label={[align=center]below:$w+1$}] {};
\node[draw, dashed, rounded corners, fit=(v0r4) (v2r4) (v0r5) (v3r5), inner sep=8pt, label={[align=center]below:$w+2$}] {};
\end{tikzpicture}
    \caption{Illustration of a DAG structure as interpreted by the partially-synchronous version of the Bullshark protocol \cite{PARTIALLYSYNCHRONOUSBULL}. The proposed leader block $L_0$ (by validator $v_1$) does not have the necessary votes to be committed, and the proposed leader block $L_1$ (by validator $v_3$) has no votes. The proposed leader block $L_2$ (by validator $v_2$) has the necessary $2f+1$ votes to be committed, but since $L_0$ appears in the causal history of $L_2$, $L_0$ is committed first. $L_1$ can be safely skipped as it does not appear in the causal history of $L_2$. Thus, the final sequence is $\{L_0,L_2\}$ \cite{PARTIALLYSYNCHRONOUSBULL}.}
    \label{fig:bullshark}
\end{figure}
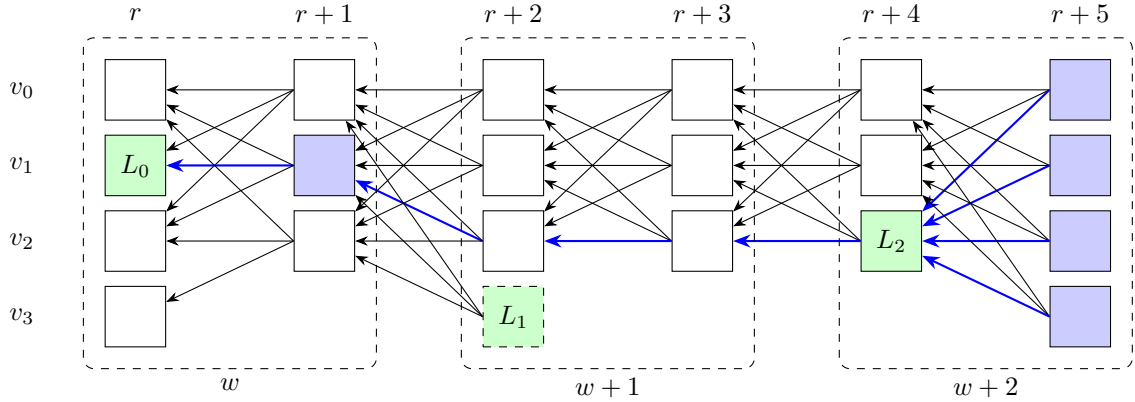
\subsubsection{Bullshark}
\begin{figure} [H]
    \centering
    \begin{tikzpicture}[
    node distance=1.5cm and 2.5cm,
    validator/.style={draw, minimum size=0.8cm},
    propose/.style={validator},
    boost/.style={validator, draw=black},
    vote/.style={validator, draw=orange, thick},
    certify/.style={validator, draw=green!60!black, thick},
    supporting_validator/.style={validator, fill=green!20},
    support/.style={->, green!60!black, thick},
    >=Stealth,
    first/.style={validator, fill=green!20},
    second/.style={validator, fill=blue!20},
    leader/.style={validator, fill=green!20},
    vote_support/.style={->, blue, thick},
    skipping_validator/.style={validator, fill=pink!40},
    cert_support/.style={->, green!60!black, thick},
    >=Stealth,
]
\footnotesize
\node at (0, 5) {$r$};
\node at (2, 5) {$r+1$};
\node at (4, 5) {$r+2$};
\node at (6, 5) {$r+3$};
\node at (8, 5) {$r+4$};
\node at (10, 5) {$r+5$};
\node at (12, 5) {$r+6$};
\node at (14, 5) {$r+7$};

\node at (-1.5, 4) {$v_0$};
\node at (-1.5, 3) {$v_1$};
\node at (-1.5, 2) {$v_2$};
\node at (-1.5, 1) {$v_3$};

\node[propose, first] (v0r) at (0, 4) {$S_{0a}$};
\node[propose] (v1r) at (0, 3) {};
\node[propose] (v2r) at (0, 2) {};
\node[propose] (v3r) at (0, 1) {$F_{0}$};

\node[propose, first] (v0r1) at (2, 4) {};
\node[propose, first] (v1r1) at (2, 3) {};
\node[propose, first] (v2r1) at (2, 2) {};
\node[propose] (v3r1) at (2, 1) {};

\node[propose] (v0r2) at (4, 4) {};
\node[propose] (v1r2) at (4, 3) {};
\node[propose] (v2r2) at (4, 2) {};
\node[propose] (v3r2) at (4, 1) {$S_{0b}$};

\node[propose] (v0r3) at (6, 4) {};
\node[propose] (v1r3) at (6, 3) {};
\node[propose] (v2r3) at (6, 2) {};
\node[supporting_validator] (v3r3) at (6, 1) {};

\node[propose, second] (v0r4) at (8, 4) {$F_1$};
\node[propose] (v1r4) at (8, 3) {};
\node[propose] (v2r4) at (8, 2) {};
\node[propose] (v3r4) at (8, 1) {$S_{2a}$};

\node[propose] (v0r5) at (10, 4) {};
\node[propose] (v1r5) at (10, 3) {};
\node[propose] (v2r5) at (10, 2) {};
\node[propose] (v3r5) at (10, 1) {};

\node[propose] (v0r6) at (12, 4) {};
\node[propose] (v1r6) at (12, 3) {};
\node[propose] (v2r6) at (12, 2) {$S_{2b}$};
\node[propose] (v3r6) at (12, 1) {};

\node[propose, second] (v0r7) at (14, 4) {};
\node[propose, second] (v1r7) at (14, 3) {};
\node[propose, second] (v2r7) at (14, 2) {};
\node[propose] (v3r7) at (14, 1) {};

\draw[cert_support] (v0r1.west) -- (v0r);
\draw[->] (v0r1.west) -- (v1r);
\draw[->] (v0r1.west) -- (v2r);

\draw[cert_support] (v1r1.west) -- (v0r);
\draw[->] (v1r1.west) -- (v1r);
\draw[->] (v1r1.west) -- (v2r);

\draw[cert_support] (v2r1.west) -- (v0r);
\draw[->] (v2r1.west) -- (v1r);
\draw[->] (v2r1.west) -- (v2r);

\draw[->] (v3r1.west) -- (v1r);
\draw[->] (v3r1.west) -- (v2r);
\draw[->] (v3r1.west) -- (v3r);

\draw[->] (v0r2.west) -- (v0r1);
\draw[->] (v0r2.west) -- (v1r1);
\draw[->] (v0r2.west) -- (v2r1);

\draw[->] (v1r2.west) -- (v0r1);
\draw[->] (v1r2.west) -- (v1r1);
\draw[->] (v1r2.west) -- (v2r1);

\draw[->] (v2r2.west) -- (v0r1);
\draw[->] (v2r2.west) -- (v1r1);
\draw[->] (v2r2.west) -- (v2r1);

\draw[->] (v3r2.west) -- (v1r1);
\draw[->] (v3r2.west) -- (v2r1);
\draw[->] (v3r2.west) -- (v3r1);

\draw[->] (v0r3.west) -- (v0r2);
\draw[->] (v0r3.west) -- (v1r2);
\draw[->] (v0r3.west) -- (v2r2);

\draw[->] (v1r3.west) -- (v0r2);
\draw[->] (v1r3.west) -- (v1r2);
\draw[->] (v1r3.west) -- (v2r2);

\draw[->] (v2r3.west) -- (v0r2);
\draw[->] (v2r3.west) -- (v1r2);
\draw[->] (v2r3.west) -- (v2r2);

\draw[->] (v3r3.west) -- (v1r2);
\draw[->] (v3r3.west) -- (v2r2);
\draw[cert_support] (v3r3.west) -- (v3r2);

\draw[->] (v0r4.west) -- (v0r3);
\draw[->] (v0r4.west) -- (v1r3);
\draw[->] (v0r4.west) -- (v2r3);

\draw[->] (v1r4.west) -- (v0r3);
\draw[->] (v1r4.west) -- (v1r3);
\draw[->] (v1r4.west) -- (v2r3);

\draw[->] (v2r4.west) -- (v0r3);
\draw[->] (v2r4.west) -- (v1r3);
\draw[->] (v2r4.west) -- (v2r3);

\draw[->] (v3r4.west) -- (v1r3);
\draw[->] (v3r4.west) -- (v2r3);
\draw[->] (v3r4.west) -- (v3r3);

\draw[vote_support] (v0r5.west) -- (v0r4);
\draw[->] (v0r5.west) -- (v1r4);
\draw[->] (v0r5.west) -- (v2r4);

\draw[vote_support] (v1r5.west) -- (v0r4);
\draw[->] (v1r5.west) -- (v1r4);
\draw[->] (v1r5.west) -- (v2r4);

\draw[vote_support] (v2r5.west) -- (v0r4);
\draw[->] (v2r5.west) -- (v1r4);
\draw[->] (v2r5.west) -- (v2r4);

\draw[->] (v3r5.west) -- (v1r4);
\draw[->] (v3r5.west) -- (v2r4);
\draw[->] (v3r5.west) -- (v3r4);

\draw[vote_support] (v0r6.west) -- (v0r5);
\draw[->] (v0r6.west) -- (v1r5);
\draw[->] (v0r6.west) -- (v2r5);

\draw[->] (v1r6.west) -- (v0r5);
\draw[vote_support] (v1r6.west) -- (v1r5);
\draw[->] (v1r6.west) -- (v2r5);

\draw[->] (v2r6.west) -- (v0r5);
\draw[->] (v2r6.west) -- (v1r5);
\draw[vote_support] (v2r6.west) -- (v2r5);

\draw[->] (v3r6.west) -- (v1r5);
\draw[->] (v3r6.west) -- (v2r5);
\draw[->] (v3r6.west) -- (v3r5);

\draw[vote_support] (v0r7.west) -- (v0r6);
\draw[->] (v0r7.west) -- (v1r6);
\draw[->] (v0r7.west) -- (v2r6);

\draw[->] (v1r7.west) -- (v0r6);
\draw[vote_support] (v1r7.west) -- (v1r6);
\draw[->] (v1r7.west) -- (v2r6);

\draw[->] (v2r7.west) -- (v0r6);
\draw[->] (v2r7.west) -- (v1r6);
\draw[vote_support] (v2r7.west) -- (v2r6);

\draw[->] (v3r7.west) -- (v1r6);
\draw[->] (v3r7.west) -- (v2r6);
\draw[->] (v3r7.west) -- (v3r6);

\node[draw, dashed, rounded corners, fit=(v0r) (v3r) (v0r3) (v3r3), inner sep=8pt, label={[align=center]below:$w$}] {};

\node[draw, dashed, rounded corners, fit=(v0r4) (v3r4) (v0r7) (v3r7), inner sep=8pt, label={[align=center]below:$w+1$}] {};

\end{tikzpicture}
    \caption{Illustration of a DAG structure as interpreted by the Bullshark protocol \cite{BULLSHARK}. [$S_{0a}$,$S_{0b}$,$S_{1a}$,$S_{1b}$] represent the \texttt{steady-state} leaders in sequential order, while [$F_0,F_1$] represent the \texttt{fallback} leaders. $S_{0a}$ is directly committed as it has the $2f+1$ necessary votes to be committed, while $S_{0b}$ does not. Since the second \texttt{steady-state} leader of wave $w$ failed to commit, wave $w+1$ will utilise the \texttt{fallback} leader $F_1$. Once round $r+7$ is reached, leader $F_1$ is revealed by combining the shares of the common coin. Since there are enough validators voting for $F_1$, $F_1$ is committed, and since $S_{0b}$ does not appear in the causal history of $F_1$, it can be safely skipped. The next wave $w+2$ will go back to using \texttt{steady-state} leaders.}
    \label{fig:bullshark_both_ways}
\end{figure}
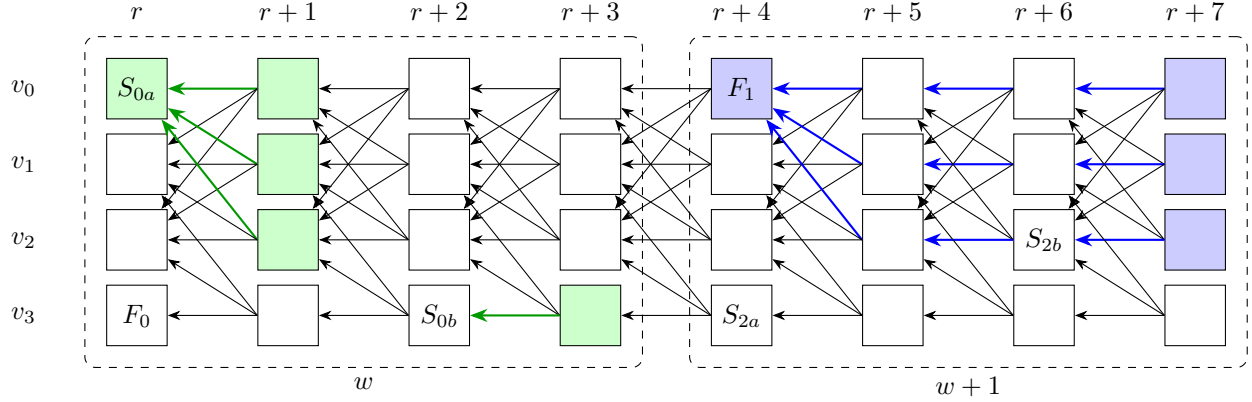
\subsubsection{Mysticeti}
\begin{figure} [H]
    \centering
    \begin{tikzpicture}[
    node distance=1.5cm and 2.5cm,
    validator/.style={rectangle, draw, minimum size=0.8cm},
    propose/.style={validator},
    boost/.style={validator, draw=black},
    vote/.style={validator, draw=orange, thick},
    certify/.style={validator, draw=green!60!black, thick},
    supporting_validator/.style={validator, fill=green!20},
    support/.style={->, green!60!black, thick},
    >=Stealth,
    leader/.style={validator, fill=green!20},
    vote_support/.style={->, orange, thick},
    skipping_validator/.style={validator, fill=pink!40},
    cert_support/.style={->, green!60!black, thick},
    >=Stealth
]
\footnotesize

\node at (0, 5) {$r$};
\node at (2.5, 5) {$r+1$};
\node at (5, 5) {$r+2$};

\node at (-1.5, 4) {$v_0$};
\node at (-1.5, 3) {$v_1$};
\node at (-1.5, 2) {$v_2$};
\node at (-1.5, 1) {$v_3$};

\node[boost] (v0r) at (0, 4) {};
\node[boost] (v1r) at (0, 3) {$L1_{b}$};
\node[boost] (v2r) at (0, 2) {};
\node[boost] (v3r) at (0, 1) {$ L1_{a}$};

\node[boost] (v0r1) at (2.5, 4) {$L2_{a}$};
\node[boost] (v1r1) at (2.5, 3) {};
\node[boost] (v2r1) at (2.5, 2) {$L2_{b}$};
\node[boost] (v3r1) at (2.5, 1) {};

\node[boost] (v0r2) at (5, 4) {};
\node[boost] (v1r2) at (5, 3) {$L3_{a}$};
\node[boost] (v2r2) at (5, 2) {};
\node[boost] (v3r2) at (5, 1) {$L3_{b}$};

\draw[->] (v0r1.west) -- (v0r);
\draw[->] (v0r1.west) -- (v1r);
\draw[->] (v0r1.west) -- (v2r);

\draw[->] (v1r1.west) -- (v0r);
\draw[->] (v1r1.west) -- (v1r);
\draw[->] (v1r1.west) -- (v2r);

\draw[->] (v2r1.west) -- (v0r);
\draw[->] (v2r1.west) -- (v1r);
\draw[->] (v2r1.west) -- (v2r);

\draw[->] (v3r1.west) -- (v1r);
\draw[->] (v3r1.west) -- (v2r);
\draw[->] (v3r1.west) -- (v3r);

\draw[->] (v0r2.west) -- (v0r1);
\draw[->] (v0r2.west) -- (v1r1);
\draw[->] (v0r2.west) -- (v2r1);

\draw[->] (v1r2.west) -- (v0r1);
\draw[->] (v1r2.west) -- (v1r1);
\draw[->] (v1r2.west) -- (v2r1);

\draw[->] (v2r2.west) -- (v0r1);
\draw[->] (v2r2.west) -- (v1r1);
\draw[->] (v2r2.west) -- (v2r1);

\draw[->] (v3r2.west) -- (v1r1);
\draw[->] (v3r2.west) -- (v2r1);
\draw[->] (v3r2.west) -- (v3r1);

\node[anchor=east] at (-0.2,0.4-0) {$\scriptstyle w$};
\node[anchor=east] at (2.3,0.4-0.2) {$\scriptstyle w+1$};
\node[anchor=east] at (4.8,0.4-0.4) {$\scriptstyle w+2$};

\draw[black, line width=1pt] (0,0.4) -- (5,0.4);
\draw[black, line width=1pt] (2.5,0.2) -- (5,0.2);
\draw[black, line width=1pt, dotted] (5,0.2) -- (5.5,0.2);
\draw[black, line width=1pt] (5,0) -- (5.5,0);
\draw[black, line width=1pt, dotted] (5,0) -- (6,0);

\end{tikzpicture}
    \caption{Illustration of the DAG structure as interpreted by the \textsf{Mysticeti}\xspace protocol. Of note, we can see how \textsf{Mysticeti}\xspace utilises two leaders per round (\textbf{multi-leader}), and how each round signifies the beginning of a new wave (\textbf{pipeline}).}
    \label{fig:}
\end{figure}
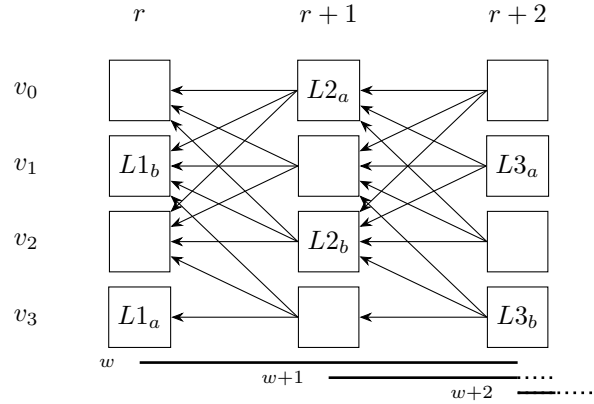
\chapter{\tarpon Additional Figures}
\label{appendix:TARPADDITIONAL}
\vspace{3em}
\subsubsection{Simulation Test Visualization}
\begin{figure}[H]
\centering
\begin{tikzpicture}[scale=0.8]
\begin{axis}[
    width=20cm,
    height=9cm,
    title={\large\textbf{Consensus Latency Performance Across Different Modes}},
    title style={yshift=8pt},
    ylabel={\textbf{Transaction Commit Latency} ($ms$)},
    xlabel={\textbf{Wave Length} ($w$)},
    xtick={3,4,5},
    xticklabels={3,4,5},
    x tick label style={font=\normalsize},
    y tick label style={font=\normalsize},
    ylabel style={font=\normalsize},
    xlabel style={font=\normalsize},
    legend style={
        at={(0.02,0.98)}, 
        anchor=north west,
        font=\normalsize,
        draw=black,
        fill=white,
        fill opacity=0.95,
        rounded corners=3pt,
        inner sep=6pt,
        line width=0.8pt
    },
    ymin=300,
    ymax=580,
    xmin=2.7,
    xmax=5.3,
    tick style={black, line width=0.6pt},
    ymajorgrids=true,
    xmajorgrids=true,
    grid style={gray!20, line width=0.4pt},
    axis background/.style={fill=gray!2},
    axis line style={line width=0.8pt}
]
\draw[
    color=green!60!black,
    line width=2pt,
    dashed
] (axis cs:\pgfkeysvalueof{/pgfplots/xmin},350.8) -- (axis cs:\pgfkeysvalueof{/pgfplots/xmax},350.8);
\addplot[
    color=green!60!black,
    mark=diamond*,
    mark size=5pt,
    line width=2pt,
    mark options={fill=green!60!black, draw=green!80!red, line width=0.2pt}
] coordinates {
    (3, 350.8)
};
\addplot[
    color=blue,
    mark=square*,
    mark size=3pt,
    line width=2pt,
    mark options={fill=blue!70, draw=blue!80!red, line width=0.2pt}
] coordinates {
    (4, 355.2)
    (5, 362.2)
};

\addplot[
    color=red,
    mark=triangle*,
    mark size=4pt,
    line width=2pt,
    mark options={fill=red!60, draw=red!75!black, line width=0.2pt}
] coordinates {
    (4, 439.5)
    (5, 530.9)
};
\draw[dashed, black!60, line width=1pt] (axis cs:4,355.2) -- (axis cs:4,439.5);
\node[font=\large, black!70, anchor=west] at (axis cs:4.02,397) {+23.7\%};
\draw[dashed, black!60, line width=1pt] (axis cs:5,362.2) -- (axis cs:5,530.9);
\node[font=\large, black!70, anchor=west] at (axis cs:5.02,446) {+46.6\%};
\legend{
    \textit{p-async} mode,
    \textit{dual} mode,
    \textit{async} mode
}
\end{axis}
\end{tikzpicture}
\caption{Visualisation of \tarpon's simulation results.}
\label{visualizationtarpon}
\end{figure}
\newpage
\subsubsection{Direct Rule: Asynchronous Mode}
\label{tarpon_direct_rule}
\begin{figure}[H]
\centering
\begin{tikzpicture}[
    node distance=1.5cm and 2.5cm,
    validator/.style={draw, minimum size=0.8cm},
    propose/.style={validator, thick},
    boost/.style={validator, draw=black},
    vote/.style={validator, draw=orange, thick},
    certify/.style={validator, draw=green!60!black, thick},
    supporting_validator/.style={validator, fill=green!20},
    support/.style={->, green!60!black, thick},
    >=Stealth,
    leader/.style={validator, fill=green!20},
    vote_support/.style={->, orange, thick},
    skipping_validator/.style={validator, fill=pink!40},
    cert_support/.style={->, green!60!black, thick},
    >=Stealth
]
\footnotesize
\node at (0, 5) {$r$};
\node at (2.5, 5) {$r+1$};
\node at (5, 5) {$r+2$};
\node at (7.5, 5) {$r+3$};
\node at (10, 5) {$r+4$};

\node at (-1.5, 4) {$v_0$};
\node at (-1.5, 3) {$v_1$};
\node at (-1.5, 2) {$v_2$};
\node at (-1.5, 1) {$v_3$};

\node[boost, leader] (v0r) at (0, 4) {$L1_{a}$};
\node[boost] (v1r) at (0, 3) {};
\node[boost] (v2r) at (0, 2) {};
\node[boost] (v3r) at (0, 1) {$L1_b$};

\node[supporting_validator] (v0r1) at (2.5, 4) {};
\node[supporting_validator] (v1r1) at (2.5, 3) {};
\node[supporting_validator] (v2r1) at (2.5, 2) {};
\node[boost] (v3r1) at (2.5, 1) {};

\node[supporting_validator] (v0r2) at (5, 4) {};
\node[supporting_validator] (v1r2) at (5, 3) {};
\node[supporting_validator] (v2r2) at (5, 2) {};
\node[supporting_validator] (v3r2) at (5, 1) {};

\node[supporting_validator] (v0r3) at (7.5, 4) {};
\node[supporting_validator] (v1r3) at (7.5, 3) {};
\node[supporting_validator] (v2r3) at (7.5, 2) {};
\node[supporting_validator] (v3r3) at (7.5, 1) {};

\node[supporting_validator] (v0r4) at (10, 4) {};
\node[supporting_validator] (v1r4) at (10, 3) {};
\node[supporting_validator] (v2r4) at (10, 2) {};
\node[supporting_validator] (v3r4) at (10, 1) {};

\draw[support] (v0r1.west) -- (v0r);
\draw[->] (v0r1.west) -- (v1r);
\draw[->] (v0r1.west) -- (v2r);

\draw[support] (v1r1.west) -- (v0r);
\draw[->] (v1r1.west) -- (v1r);
\draw[->] (v1r1.west) -- (v2r);

\draw[support] (v2r1.west) -- (v0r);
\draw[->] (v2r1.west) -- (v1r);
\draw[->] (v2r1.west) -- (v2r);

\draw[->] (v3r1.west) -- (v1r);
\draw[->] (v3r1.west) -- (v2r);
\draw[->] (v3r1.west) -- (v3r);

\draw[support] (v0r2.west) -- (v0r1);
\draw[support] (v0r2.west) -- (v1r1);
\draw[support] (v0r2.west) -- (v2r1);

\draw[support] (v1r2.west) -- (v0r1);
\draw[support] (v1r2.west) -- (v1r1);
\draw[support] (v1r2.west) -- (v2r1);

\draw[support] (v2r2.west) -- (v0r1);
\draw[support] (v2r2.west) -- (v1r1);
\draw[support] (v2r2.west) -- (v2r1);

\draw[support] (v3r2.west) -- (v1r1);
\draw[support] (v3r2.west) -- (v2r1);
\draw[->] (v3r2.west) -- (v3r1);

\draw[support] (v0r3.west) -- (v0r2);
\draw[support] (v0r3.west) -- (v1r2);
\draw[support] (v0r3.west) -- (v2r2);

\draw[support] (v1r3.west) -- (v0r2);
\draw[support] (v1r3.west) -- (v1r2);
\draw[support] (v1r3.west) -- (v2r2);

\draw[support] (v2r3.west) -- (v0r2);
\draw[support] (v2r3.west) -- (v1r2);
\draw[support] (v2r3.west) -- (v2r2);

\draw[support] (v3r3.west) -- (v1r2);
\draw[support] (v3r3.west) -- (v2r2);
\draw[support] (v3r3.west) -- (v3r2);

\draw[support] (v0r4.west) -- (v0r3);
\draw[support] (v0r4.west) -- (v1r3);
\draw[support] (v0r4.west) -- (v2r3);

\draw[support] (v1r4.west) -- (v0r3);
\draw[support] (v1r4.west) -- (v1r3);
\draw[support] (v1r4.west) -- (v2r3);

\draw[support] (v2r4.west) -- (v0r3);
\draw[support] (v2r4.west) -- (v1r3);
\draw[support] (v2r4.west) -- (v2r3);

\draw[support] (v3r4.west) -- (v1r3);
\draw[support] (v3r4.west) -- (v2r3);
\draw[support] (v3r4.west) -- (v3r3);

\node[draw, dashed, rounded corners, fit=(v0r) (v3r) (v0r4) (v3r4), inner sep=8pt, label={[align=center]below:$w$}] {};

\end{tikzpicture}
\label{fig:mahi_mahi_direct_commit}
\label{fig:mahi_mahi_direct_commit}
\caption{Certificate Pattern (Asynchronous mode $w=5$): $L1_a$ is marked as \texttt{to-commit} because validator $v_0$, $v_1$ and $v_2$ voted for leader block $L1_a$. This means, in the case of \tarpon's asynchronous mode, leader block $L1_b$ has been marked as \texttt{to-commit} because it appeared in the causal history of the blocks proposed by validators $v_0,v_1,v_2$ in round $r+3$ (i.e., the voting round).}
\end{figure}
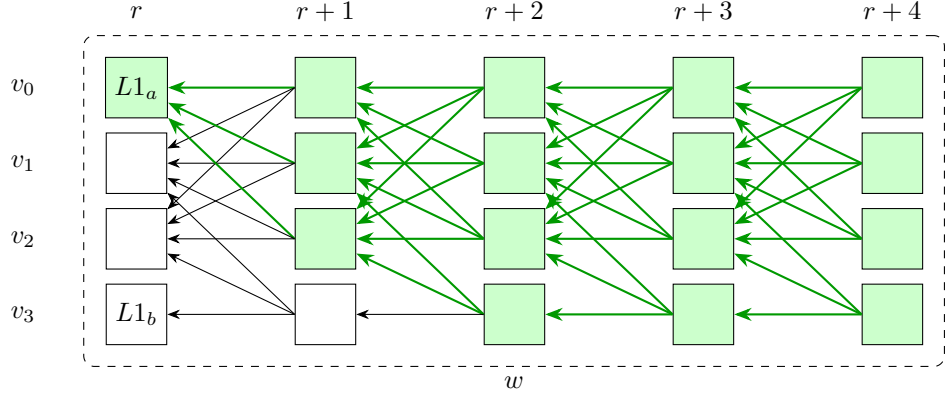
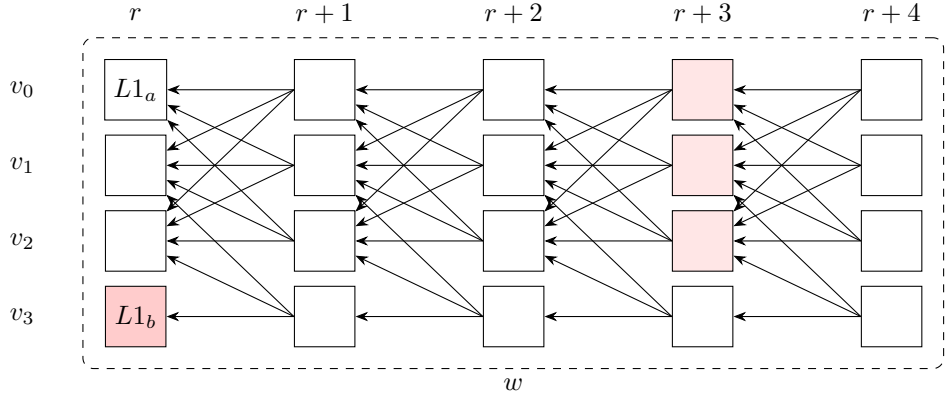
\begin{figure}[H]
\centering
\begin{tikzpicture}[
    node distance=1.5cm and 2.5cm,
    validator/.style={draw, minimum size=0.8cm},
    propose/.style={validator, thick},
    boost/.style={validator, draw=black},
    vote/.style={validator, draw=orange, thick},
    certify/.style={validator, draw=green!60!black, thick},
    supporting_validator/.style={validator, fill=green!20},
    support/.style={->, green!60!black, thick},
    nosupport/.style={->, red!60!black, thick},
    >=Stealth,
    leader/.style={validator, fill=red!20},
    vote_support/.style={->, orange, thick},
    skipping_validator/.style={validator, fill=pink!40},
    cert_support/.style={->, green!60!black, thick},
    >=Stealth
]
\footnotesize
\node at (0, 5) {$r$};
\node at (2.5, 5) {$r+1$};
\node at (5, 5) {$r+2$};
\node at (7.5, 5) {$r+3$};
\node at (10, 5) {$r+4$};

\node at (-1.5, 4) {$v_0$};
\node at (-1.5, 3) {$v_1$};
\node at (-1.5, 2) {$v_2$};
\node at (-1.5, 1) {$v_3$};

\node[boost] (v0r) at (0, 4) {$L1_{a}$};
\node[boost] (v1r) at (0, 3) {};
\node[boost] (v2r) at (0, 2) {};
\node[boost,leader] (v3r) at (0, 1) {$L1_b$};

\node[boost] (v0r1) at (2.5, 4) {};
\node[boost] (v1r1) at (2.5, 3) {};
\node[boost] (v2r1) at (2.5, 2) {};
\node[boost] (v3r1) at (2.5, 1) {};

\node[boost] (v0r2) at (5, 4) {};
\node[boost] (v1r2) at (5, 3) {};
\node[boost] (v2r2) at (5, 2) {};
\node[boost] (v3r2) at (5, 1) {};

\node[skipping_validator] (v0r3) at (7.5, 4) {};
\node[skipping_validator] (v1r3) at (7.5, 3) {};
\node[skipping_validator] (v2r3) at (7.5, 2) {};
\node[boost] (v3r3) at (7.5, 1) {};

\node[boost] (v0r4) at (10, 4) {};
\node[boost] (v1r4) at (10, 3) {};
\node[boost] (v2r4) at (10, 2) {};
\node[boost] (v3r4) at (10, 1) {};

\draw[->] (v0r1.west) -- (v0r);
\draw[->] (v0r1.west) -- (v1r);
\draw[->] (v0r1.west) -- (v2r);

\draw[->] (v1r1.west) -- (v0r);
\draw[->] (v1r1.west) -- (v1r);
\draw[->] (v1r1.west) -- (v2r);

\draw[->] (v2r1.west) -- (v0r);
\draw[->] (v2r1.west) -- (v1r);
\draw[->] (v2r1.west) -- (v2r);

\draw[->] (v3r1.west) -- (v1r);
\draw[->] (v3r1.west) -- (v2r);
\draw[->] (v3r1.west) -- (v3r);

\draw[->] (v0r2.west) -- (v0r1);
\draw[->] (v0r2.west) -- (v1r1);
\draw[->] (v0r2.west) -- (v2r1);

\draw[->] (v1r2.west) -- (v0r1);
\draw[->] (v1r2.west) -- (v1r1);
\draw[->] (v1r2.west) -- (v2r1);

\draw[->] (v2r2.west) -- (v0r1);
\draw[->] (v2r2.west) -- (v1r1);
\draw[->] (v2r2.west) -- (v2r1);

\draw[->] (v3r2.west) -- (v1r1);
\draw[->] (v3r2.west) -- (v2r1);
\draw[->] (v3r2.west) -- (v3r1);

\draw[->] (v0r3.west) -- (v0r2);
\draw[->] (v0r3.west) -- (v1r2);
\draw[->] (v0r3.west) -- (v2r2);

\draw[->] (v1r3.west) -- (v0r2);
\draw[->] (v1r3.west) -- (v1r2);
\draw[->] (v1r3.west) -- (v2r2);

\draw[->] (v2r3.west) -- (v0r2);
\draw[->] (v2r3.west) -- (v1r2);
\draw[->] (v2r3.west) -- (v2r2);

\draw[->] (v3r3.west) -- (v1r2);
\draw[->] (v3r3.west) -- (v2r2);
\draw[->] (v3r3.west) -- (v3r2);

\draw[->] (v0r4.west) -- (v0r3);
\draw[->] (v0r4.west) -- (v1r3);
\draw[->] (v0r4.west) -- (v2r3);

\draw[->] (v1r4.west) -- (v0r3);
\draw[->] (v1r4.west) -- (v1r3);
\draw[->] (v1r4.west) -- (v2r3);

\draw[->] (v2r4.west) -- (v0r3);
\draw[->] (v2r4.west) -- (v1r3);
\draw[->] (v2r4.west) -- (v2r3);

\draw[->] (v3r4.west) -- (v1r3);
\draw[->] (v3r4.west) -- (v2r3);
\draw[->] (v3r4.west) -- (v3r3);

\node[draw, dashed, rounded corners, fit=(v0r) (v3r) (v0r4) (v3r4), inner sep=8pt, label={[align=center]below:$w$}] {};

\end{tikzpicture}
\label{fig:mahi_mahi_direct_skip}
    \caption{Skip Pattern (Asynchronous mode $w=5$): Leader block $L1_d$ is marked as \texttt{to-skip} because validator $v_0$, $v_1$ and $v_2$ did not vote for $L1_b$. This means, in the case of \tarpon's asynchronous mode, leader block $L1_b$ has not been marked as \texttt{to-commit} because it did not appear in the causal history of the blocks proposed by validators $v_0,v_1,v_2$ in round $r+3$ (i.e., the voting round).}
\end{figure}
\newpage
\subsubsection{Indirect Rule: Partially-Synchronous Mode}
\label{tarpon_indirect_rule}
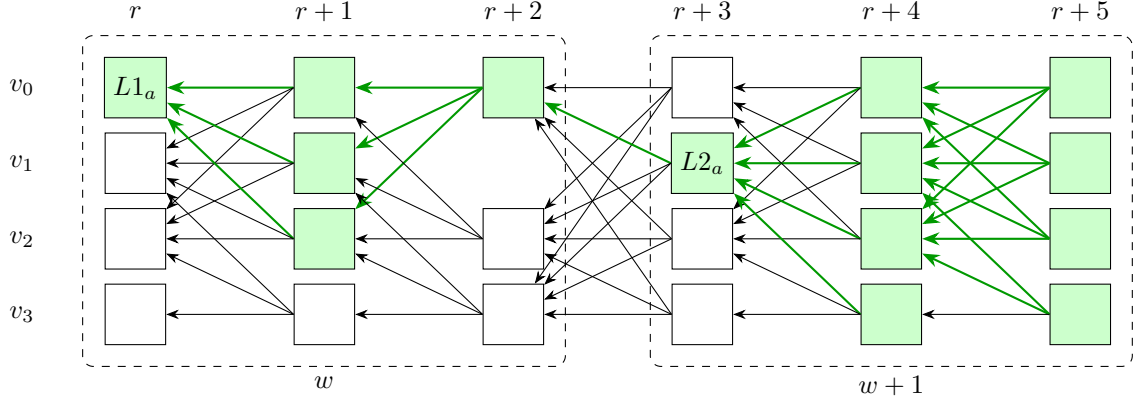
\begin{figure}[H]
\centering
\begin{tikzpicture}[
    node distance=1.5cm and 2.5cm,
    validator/.style={draw, minimum size=0.8cm},
    propose/.style={validator, thick},
    boost/.style={validator, draw=black},
    vote/.style={validator, draw=orange, thick},
    certify/.style={validator, draw=green!60!black, thick},
    supporting_validator/.style={validator, fill=green!20},
    support/.style={->, green!60!black, thick},
    >=Stealth,
    leader/.style={validator, fill=green!20},
    vote_support/.style={->, orange, thick},
    skipping_validator/.style={validator, fill=pink!40},
    cert_support/.style={->, green!60!black, thick},
    >=Stealth
]
\footnotesize
\node at (0, 5) {$r$};
\node at (2.5, 5) {$r+1$};
\node at (5, 5) {$r+2$};
\node at (7.5, 5) {$r+3$};
\node at (10, 5) {$r+4$};
\node at (12.5, 5) {$r+5$};
\node at (-1.5, 4) {$v_0$};
\node at (-1.5, 3) {$v_1$};
\node at (-1.5, 2) {$v_2$};
\node at (-1.5, 1) {$v_3$};

\node[supporting_validator, boost] (v0r) at (0, 4) {$L1_{a}$};
\node[boost] (v1r) at (0, 3) {};
\node[boost] (v2r) at (0, 2) {};
\node[boost] (v3r) at (0, 1) {};

\node[supporting_validator] (v0r1) at (2.5, 4) {};
\node[supporting_validator] (v1r1) at (2.5, 3) {};
\node[supporting_validator] (v2r1) at (2.5, 2) {};
\node[boost] (v3r1) at (2.5, 1) {};

\node[supporting_validator] (v0r2) at (5, 4) {};
\node[boost] (v2r2) at (5, 2) {};
\node[boost] (v3r2) at (5, 1) {};

\node[boost] (v0r3) at (7.5, 4) {};
\node[boost, leader] (v1r3) at (7.5, 3) {$L2_a$};
\node[boost] (v2r3) at (7.5, 2) {};
\node[boost] (v3r3) at (7.5, 1) {};

\node[supporting_validator] (v0r4) at (10, 4) {};
\node[supporting_validator] (v1r4) at (10, 3) {};
\node[supporting_validator] (v2r4) at (10, 2) {};
\node[supporting_validator] (v3r4) at (10, 1) {};

\node[supporting_validator] (v0r5) at (12.5, 4) {};
\node[supporting_validator] (v1r5) at (12.5, 3) {};
\node[supporting_validator] (v2r5) at (12.5, 2) {};
\node[supporting_validator] (v3r5) at (12.5, 1) {};
\draw[support] (v0r1.west) -- (v0r);
\draw[->] (v0r1.west) -- (v1r);
\draw[->] (v0r1.west) -- (v2r);

\draw[support] (v1r1.west) -- (v0r);
\draw[->] (v1r1.west) -- (v1r);
\draw[->] (v1r1.west) -- (v2r);

\draw[support] (v2r1.west) -- (v0r);
\draw[->] (v2r1.west) -- (v1r);
\draw[->] (v2r1.west) -- (v2r);

\draw[->] (v3r1.west) -- (v1r);
\draw[->] (v3r1.west) -- (v2r);
\draw[->] (v3r1.west) -- (v3r);

\draw[support] (v0r2.west) -- (v0r1);
\draw[support] (v0r2.west) -- (v1r1);
\draw[support] (v0r2.west) -- (v2r1);

\draw[->] (v2r2.west) -- (v0r1);
\draw[->] (v2r2.west) -- (v1r1);
\draw[->] (v2r2.west) -- (v2r1);

\draw[->] (v3r2.west) -- (v1r1);
\draw[->] (v3r2.west) -- (v2r1);
\draw[->] (v3r2.west) -- (v3r1);

\draw[->] (v0r4.west) -- (v0r3);
\draw[support] (v0r4.west) -- (v1r3);
\draw[->] (v0r4.west) -- (v2r3);

\draw[->] (v1r4.west) -- (v0r3);
\draw[support] (v1r4.west) -- (v1r3);
\draw[->] (v1r4.west) -- (v2r3);

\draw[->] (v2r4.west) -- (v0r3);
\draw[support] (v2r4.west) -- (v1r3);
\draw[->] (v2r4.west) -- (v2r3);

\draw[support] (v3r4.west) -- (v1r3);
\draw[->] (v3r4.west) -- (v2r3);
\draw[->] (v3r4.west) -- (v3r3);

\draw[->] (v0r3.west) -- (v0r2);
\draw[->] (v0r3.west) -- (v2r2);
\draw[->] (v0r3.west) -- (v3r2);

\draw[support] (v1r3.west) -- (v0r2);
\draw[->] (v1r3.west) -- (v2r2);
\draw[->] (v1r3.west) -- (v3r2);

\draw[->] (v2r3.west) -- (v0r2);
\draw[->] (v2r3.west) -- (v2r2);
\draw[->] (v2r3.west) -- (v3r2);

\draw[->] (v3r3.west) -- (v0r2);
\draw[->] (v3r3.west) -- (v2r2);
\draw[->] (v3r3.west) -- (v3r2);

\draw[support] (v0r5.west) -- (v0r4);
\draw[support] (v0r5.west) -- (v1r4);
\draw[support] (v0r5.west) -- (v2r4);

\draw[support] (v1r5.west) -- (v0r4);
\draw[support] (v1r5.west) -- (v1r4);
\draw[support] (v1r5.west) -- (v2r4);

\draw[support] (v2r5.west) -- (v0r4);
\draw[support] (v2r5.west) -- (v1r4);
\draw[support] (v2r5.west) -- (v2r4);

\draw[support] (v3r5.west) -- (v1r4);
\draw[support] (v3r5.west) -- (v2r4);
\draw[->] (v3r5.west) -- (v3r4);

\node[draw, dashed, rounded corners, fit=(v0r) (v3r) (v0r2) (v3r2), inner sep=8pt, label={[align=center]below:$w$}] {};

\node[draw, dashed, rounded corners, fit=(v0r3) (v3r3) (v0r5) (v3r5), inner sep=8pt, label={[align=center]below:$w+1$}] {};
\end{tikzpicture}
\caption{Example of \tarpon's partially-synchronous mode using the indirect rule to mark $L1_a$ as \texttt{to-commit}. While $L1_a$ has the required $2f+1$ votes from blocks in round $r+1$, it's lacking $2f+1$ certificates. Block $L2_a$ has been directly marked as \texttt{to-commit}, and, since there is a certificate link (i.e. $L2_a$ observes $2f+1$ votes for $L1_a$ in its causal history) between $L2_a$ and $L1_a$, $L1_a$ can be marked as \texttt{to-commit}.}
\label{fig:mysticeti_indirect_commit}
\end{figure}
\begin{figure}[H]
\centering
\begin{tikzpicture}[
    node distance=1.5cm and 2.5cm,
    validator/.style={draw, minimum size=0.8cm},
    propose/.style={validator, thick},
    boost/.style={validator, draw=black},
    vote/.style={validator, draw=orange, thick},
    certify/.style={validator, draw=green!60!black, thick},
    supporting_validator/.style={validator, fill=green!20},
    support/.style={->, green!60!black, thick},
    nosupport/.style={->, red!60!black, thick},
    >=Stealth,
    leader/.style={validator, fill=green!20},
    vote_support/.style={->, orange, thick},
    skipping_validator/.style={validator, fill=pink!40},
    cert_support/.style={->, green!60!black, thick},
    >=Stealth
]
\footnotesize
\node at (0, 5) {$r$};
\node at (2.5, 5) {$r+1$};
\node at (5, 5) {$r+2$};
\node at (7.5, 5) {$r+3$};
\node at (10, 5) {$r+4$};
\node at (12.5, 5) {$r+5$};
\node at (-1.5, 4) {$v_0$};
\node at (-1.5, 3) {$v_1$};
\node at (-1.5, 2) {$v_2$};
\node at (-1.5, 1) {$v_3$};

\node[boost, skipping_validator] (v0r) at (0, 4) {$L1_{a}$};
\node[boost] (v1r) at (0, 3) {};
\node[boost] (v2r) at (0, 2) {};
\node[boost] (v3r) at (0, 1) {};

\node[boost, skipping_validator] (v0r1) at (2.5, 4) {};
\node[boost, skipping_validator] (v1r1) at (2.5, 3) {};
\node[boost] (v2r1) at (2.5, 2) {};
\node[boost] (v3r1) at (2.5, 1) {};

\node[boost, skipping_validator] (v0r2) at (5, 4) {};
\node[boost, skipping_validator] (v1r2) at (5, 3) {};
\node[boost, skipping_validator] (v2r2) at (5, 2) {};
\node[boost] (v3r2) at (5, 1) {};

\node[boost] (v0r3) at (7.5, 4) {};
\node[boost, leader] (v1r3) at (7.5, 3) {$L2_a$};
\node[boost] (v2r3) at (7.5, 2) {};
\node[boost] (v3r3) at (7.5, 1) {};

\node[supporting_validator] (v0r4) at (10, 4) {};
\node[supporting_validator] (v1r4) at (10, 3) {};
\node[supporting_validator] (v2r4) at (10, 2) {};
\node[supporting_validator] (v3r4) at (10, 1) {};

\node[supporting_validator] (v0r5) at (12.5, 4) {};
\node[supporting_validator] (v1r5) at (12.5, 3) {};
\node[supporting_validator] (v2r5) at (12.5, 2) {};
\node[supporting_validator] (v3r5) at (12.5, 1) {};

\draw[nosupport] (v0r1.west) -- (v0r);
\draw[->] (v0r1.west) -- (v1r);
\draw[->] (v0r1.west) -- (v2r);

\draw[nosupport] (v1r1.west) -- (v0r);
\draw[->] (v1r1.west) -- (v1r);
\draw[->] (v1r1.west) -- (v2r);

\draw[->] (v2r1.west) -- (v1r);
\draw[->] (v2r1.west) -- (v2r);
\draw[->] (v2r1.west) -- (v3r);

\draw[->] (v3r1.west) -- (v1r);
\draw[->] (v3r1.west) -- (v2r);
\draw[->] (v3r1.west) -- (v3r);

\draw[nosupport] (v0r2.west) -- (v0r1);
\draw[nosupport] (v0r2.west) -- (v1r1);
\draw[->] (v0r2.west) -- (v2r1);

\draw[nosupport] (v1r2.west) -- (v0r1);
\draw[nosupport] (v1r2.west) -- (v1r1);
\draw[->] (v1r2.west) -- (v2r1);

\draw[nosupport] (v2r2.west) -- (v0r1);
\draw[nosupport] (v2r2.west) -- (v1r1);
\draw[->] (v2r2.west) -- (v2r1);

\draw[->] (v3r2.west) -- (v1r1);
\draw[->] (v3r2.west) -- (v2r1);
\draw[->] (v3r2.west) -- (v3r1);

\draw[->] (v0r4.west) -- (v0r3);
\draw[support] (v0r4.west) -- (v1r3);
\draw[->] (v0r4.west) -- (v2r3);

\draw[->] (v1r4.west) -- (v0r3);
\draw[support] (v1r4.west) -- (v1r3);
\draw[->] (v1r4.west) -- (v2r3);

\draw[->] (v2r4.west) -- (v0r3);
\draw[support] (v2r4.west) -- (v1r3);
\draw[->] (v2r4.west) -- (v2r3);

\draw[support] (v3r4.west) -- (v1r3);
\draw[->] (v3r4.west) -- (v2r3);
\draw[->] (v3r4.west) -- (v3r3);

\draw[->] (v0r3.west) -- (v0r2);
\draw[->] (v0r3.west) -- (v1r2);
\draw[->] (v0r3.west) -- (v2r2);

\draw[nosupport] (v1r3.west) -- (v0r2);
\draw[nosupport] (v1r3.west) -- (v1r2);
\draw[nosupport] (v1r3.west) -- (v2r2);

\draw[->] (v2r3.west) -- (v0r2);
\draw[->] (v2r3.west) -- (v1r2);
\draw[->] (v2r3.west) -- (v2r2);

\draw[->] (v3r3.west) -- (v1r2);
\draw[->] (v3r3.west) -- (v2r2);
\draw[->] (v3r3.west) -- (v3r2);

\draw[support] (v0r5.west) -- (v0r4);
\draw[support] (v0r5.west) -- (v1r4);
\draw[support] (v0r5.west) -- (v2r4);

\draw[support] (v1r5.west) -- (v0r4);
\draw[support] (v1r5.west) -- (v1r4);
\draw[support] (v1r5.west) -- (v2r4);

\draw[support] (v2r5.west) -- (v0r4);
\draw[support] (v2r5.west) -- (v1r4);
\draw[support] (v2r5.west) -- (v2r4);

\draw[support] (v3r5.west) -- (v1r4);
\draw[support] (v3r5.west) -- (v2r4);
\draw[->] (v3r5.west) -- (v3r4);

\node[draw, dashed, rounded corners, fit=(v0r) (v3r) (v0r2) (v3r2), inner sep=8pt, label={[align=center]below:$w$}] {};

\node[draw, dashed, rounded corners, fit=(v0r3) (v3r3) (v0r5) (v3r5), inner sep=8pt, label={[align=center]below:$w+1$}] {};
\end{tikzpicture}
\caption{Example of \tarpon's partially-synchronous mode using the indirect rule to mark $L1_a$ as \texttt{to-skip}. While $L1_a$ does not have the required $2f+1$ votes from blocks in round $r+1$, it also does not have $2f+1$ blames required to directly skip the block. Leader block $L2_a$ has been directly marked as \texttt{to-commit}, and, since there is no certificate link (i.e. $L2_a$ does not observes $2f+1$ votes for $L1_a$ in its causal history) between $L2_a$ and $L1_a$, $L1_a$ can be safely marked as \texttt{to-skip}.} 
\label{fig:mysticeti_indirect_skip}
\end{figure}
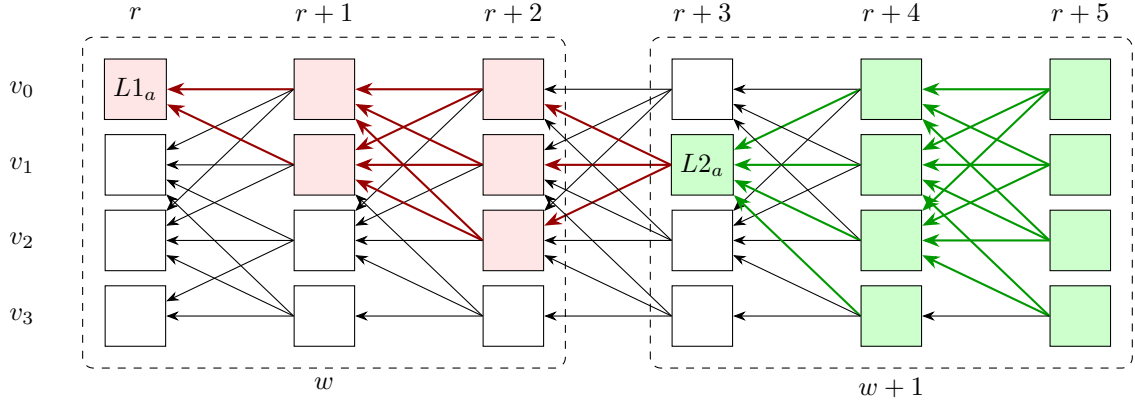
\subsubsection{Example Evaluation Figures}
\begin{figure}[H]
    \centering
    \includegraphics[width=1\linewidth, height=4cm]{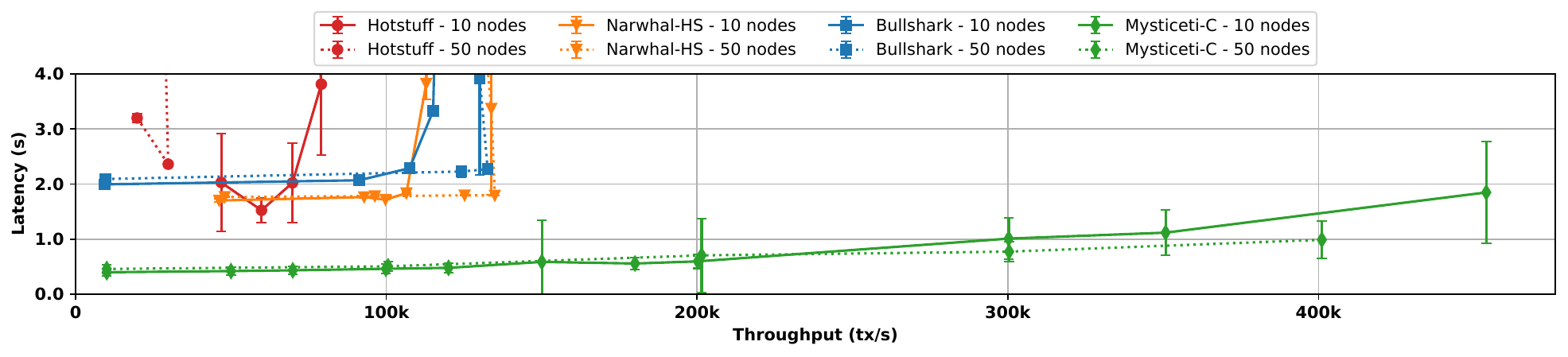}
    \caption{Throughput and Latency graph comparing \textsf{Mysticeti-C}\xspace with other state-of-the-art consensus protocols, as presented in the works of Babel et al. \cite{MYSTICETI}}
    \label{fig:mysticeti_results}
\end{figure}
\vspace{5em}
\begin{figure}[H]
    \centering
    \includegraphics[width=1\linewidth,height=4.5cm]{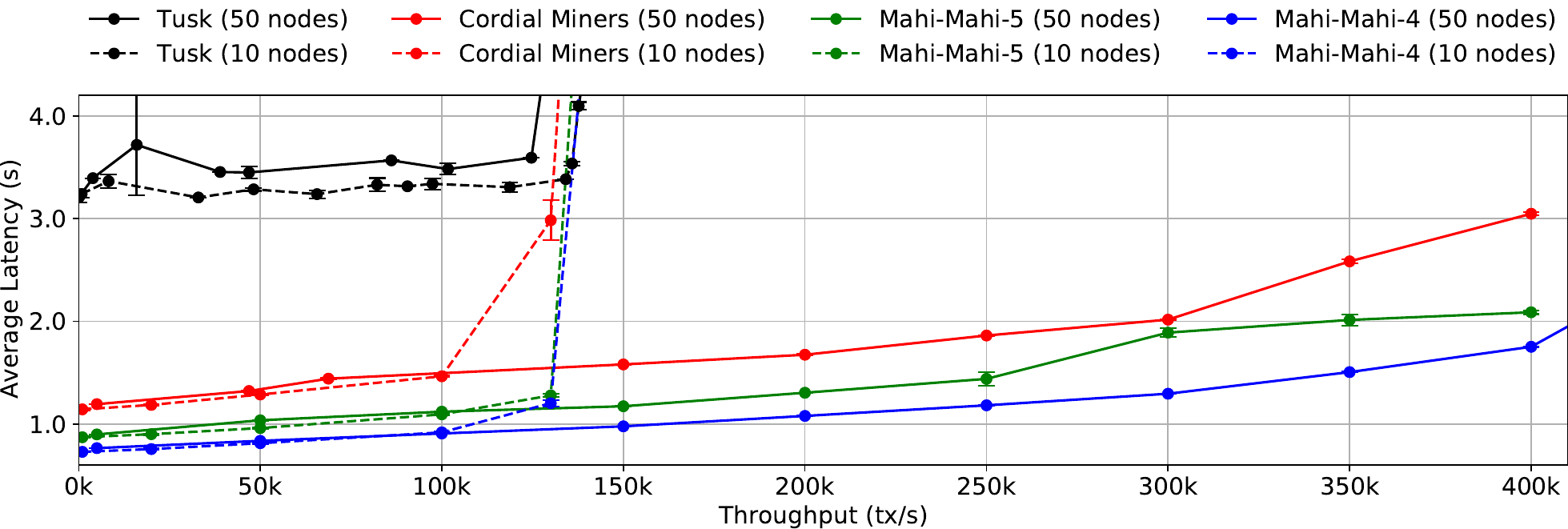}
    \caption{Throughput and Latency graph comparing \textsf{Mahi-Mahi}\xspace with other state-of-the-art consensus protocols, as presented in the works of Jovanovic et al, \cite{MAHIMAHI}}
    \label{fig:mahi_results}
\end{figure}
\chapter{Algorithms}
\label{Appendix:Algorithms}
\quad In this section we present the algorithms for \tarpon in pseudo-code. Algorithm~\ref{alg:main} and Algorithm~\ref{alg:decider} represents \tarpon's main function and decider instance. As \tarpon builds upon \textsf{Mysticeti}\xspace and \textsf{Mahi-Mahi}\xspace, the algorithms presented here adapt those used in the two protocols \cite{MYSTICETI, MAHIMAHI}, modified to support the dual-mode functionality. We begin with Algorithm~\ref{alg:main}, which assumes a \texttt{Scheduler} instance is already available. Everything passes through the \textsc{TryCommit} process, and is called whenever a valid block is presented to a validator. This process  calls the \textsc{TryDecide} process which will use the decision rules (talked about in Section~\ref{subsection:decision_rules}) to try and classify each slot in the DAG from $r_{committed}$ to $r_{highest}$. It will first try the \textsc{TryDirectDecide} to mark the slot via the direct rule, if it fails, it will then try to mark the slot via the indirect rule, via the \textsc{TryIndirectDecide} procedure. Before calling either decision rule, the round number is given to the \texttt{scheduler} to determine which mode of operation the decision rule should follow. The \textsc{TryCommit} process will then return a sequence of leaders, and commit as many in the sequence as possible. The returned sequence is truncated into two scenarios: one is, when the first undecided leader is encountered, and the other is, when one of the committed leaders is an asynchronous leader (i.e., the leader's round equals the current asynchronous rate). In the latter case, the sequence is cut short and we invoke process \textsc{addSequence} from Algorithm~\ref{alg:dual_mode}. In all cases, the sequence of leaders is also given to \tarpon's dual-mode scheduler, for later analysis, by calling the process \textsc{addSequence} as defined in Algorithm~\ref{alg:dual_mode}.
\begin{algorithm}[H]
\caption{Dual Mode Protocol Main Function}
\label{alg:main}
\scriptsize
\begin{algorithmic}[1]
\State $\texttt{waveLength}$ \Comment{See Section~\ref{subsection:roundandwave}}
\State $\texttt{proposersPerRound}$ \Comment{See Section~\ref{sec:prop_slot_leader}}
\State $\texttt{scheduler}$ \Comment{Algorithm \ref{alg:dual_mode}}
\State
\Procedure{TryCommit}{$r_{\text{committed}}, r_{\text{highest}}$}
\State $L \gets$ \textsc{TryDecide}($r_{\text{committed}}, r_{\text{highest}}$) 
\State $L_{\text{commit}} \gets [ ]$ \Comment{Hold decided leader sequence}
\For{status $\in L$}
\If{status $= \bot$}
\State \textbf{break} \Comment{Stop at first undecided leader}
\EndIf
\If{status $= $ \textsc{commit}($b_{\text{leader}}$)}
\If{\texttt{scheduler}.\textsc{isAsync}($r_{b_{\text{leader}}}$)} \Comment{Stop if committed leader is asynchronous}
\State $L_{\text{commit}} \gets L_{\text{commit}} \parallel b_{\text{leader}}$
\State \texttt{scheduler}.\textsc{addSequence}($L_{\text{commit}}$) \Comment{Add leader sequence to \texttt{scheduler}}
\State \texttt{scheduler}.\textsc{UpdateAsyncSchedule} \Comment{Update Asynchronous Schedule}
\State \textbf{return} $L_{\text{commit}}$ \Comment{Stop processing}
\Else
\State $L_{\text{commit}} \gets L_{\text{commit}} \parallel b_{\text{leader}}$ \Comment{Keep going}
\EndIf
\EndIf
\EndFor
\State \texttt{scheduler}.\textsc{addSequence}($L_{\text{commit}}$) \Comment{Add leader sequence to \texttt{scheduler}}
\State \textbf{return} $L_{\text{commit}}$
\EndProcedure
\State
\Procedure{TryDecide}{$r_{\text{committed}}, r_{\text{highest}}$}
\State $L \gets [ ]$ \Comment{Leader decision sequence}
\For{$r \gets r_{\text{highest}}$ \textbf{down to} $r_{\text{committed}} + 1$}
\State $\text{async} \gets \texttt{scheduler}.\textsc{isAsync($r$)}$ \Comment{Determine mode of operation for $r$}
\For{$l \gets 0$ \textbf{to} \texttt{proposerPerRound} $- 1$}
\State $D \gets$ \textsc{GetWaveData}($r$, $l$) 
\State $w \gets$ D.\textsc{WaveNumber}($r$)
\If{D.\textsc{ProposeRound}($w$) $= r$}
\State $\text{status} \gets$ D.\textsc{TryDirectDecide}($w$, \text{async}) \Comment{Apply direct decision rule}
\If{$\text{status} = \bot$}
\State $\text{status} \gets$ D.\textsc{TryIndirectDecide}($w$, \text{async}) \Comment{Apply indirect decision rule}
\EndIf
\State $L \gets \text{status} \parallel L$
\EndIf
\EndFor
\EndFor
\State \textbf{return} $L$ \Comment{May still contain undecided leaders}
\EndProcedure
\end{algorithmic}
\end{algorithm}

\quad Algorithm~\ref{alg:decider} represents the decider functions, functions that handle the logic for proposer slots. These functions are, for the most part, very similar to the decider functions found within \textsf{Mysticeti}\xspace and \textsf{Mahi-Mahi}\xspace \cite{MYSTICETI,MAHIMAHI}. However, they differ specifically in the \textsc{CertifyRound} and \textsc{LeaderBlock} processes. 
\begin{itemize}
\item \textsc{CertifyRound} (Line 12): This function takes two arguments: $w$ represents the wave in question, and $async$ represents whether the wave should be interpreted using asynchronous mode or partially-synchronous mode.
\item \textsc{LeaderBlock} (Line 21): This procedure combines the logic of \textsf{Mysticeti}\xspace's and \textsf{Mahi-Mahi}\xspace's leader election systems. If the wave is to be run in asynchronous mode, an additional offset, obtained from the combined coin shares, is added before electing the leader, introducing randomness to the leader election. If the wave is to be run under the partially-synchronous mode, then the  deterministic process \textsc{GetLeader} will be used, just like Mysticeti \cite{MAHIMAHI, MYSTICETI}.
\end{itemize}
\quad The certify round will change depending on the mode of execution of the specific round. For example, if we are checking the certify round for round $r$, utilising the partially-synchronous wave length, we will have $r+1$, while within an asynchronous wave we might have either $r+3$ or $r+4$ depending on the chosen asynchronous wave length. Similarly, in the \textsc{LeaderBlock} process, if \textsc{async}, then when determining the leader, we also release the common coin. However, it is important to note, the implementation of the common coin in asynchronous consensus protocols is an open research question. This means that \tarpon, like all other asynchronous protocols we have seen, does not actually implement the mechanism, but adoption of such mechanism can be easily inserted. Lastly, regarding \textsc{WaveNumber}, since \tarpon's partially-synchronous mode is the "base" mode for \tarpon, we do not need to change how the wave number is calculated. \tarpon's asynchronous mode will always begin after a partially-synchronous wave. Thus, to maintain simplicity, we utilize the partially-synchronous mode as baseline counting.
\\
\\
\null \quad Lastly, we have Algorithm~\ref{alg:dual_mode}. This represents \tarpon's dual-mode scheduler, and contains all the logic for the scheduling of \tarpon's asynchronous mode. Of note, decided leaders obtained from \textsc{TryCommit} are stored inside a schedule instance (added through the \textsc{addSequence} process).
\begin{itemize}
    \item \textsc{UpdateAsyncSchedule} (Line 7): This process is called by the\textsc{TryCommit} process, and initiates the asynchronous scheduling update. It will always be the case, when this process is called, that the newest leader block in \texttt{decidedLeaders} is a leader which has been committed through the asynchronous mode. This leader is popped, and its subDAG is obtained via \textsc{getSubDag}, which is then used to count the direct commits that have happened. Lastly, the schedule is updated accordingly, \texttt{decidedLeaders} is drained, \texttt{lastAsyncRound} and \texttt{asyncInterval} are updated.  
    \item \textsc{getSubDag} (Line 15): This process collects all the blocks from an anchor $v$, till a specified minimum round $r_{min}$. $r_{min}$ is the oldest leader block found in \texttt{decidedLeaders}  because we want the subDAG to find if leaders in \texttt{decidedLeaders} have been directly committed. We then collect all blocks that are reachable from $v$, including ones from $r_{min}$. The process then returns a topological linearization of $v$'s subDAG. 
    \item \textsc{CountDirectCommits} (Line 30): counts direct commits by first finding and storing blocks that have voted for a leader $\ell$ inside of list $V_{\ell}$. This is done by looking at blocks found in round $r_\ell +1$ (i.e., the voting round). Then, in round $r_\ell + 2$, we store, in $C_\ell$, how many blocks reference at least $2f+1$ of the blocks found in $V_\ell$. If $C_\ell$ has a magnitude of at least $2f+1$, then we know leader block $\ell$ has been directly committed within this specific subDAG. We repeat this process for all leader blocks marked as \texttt{to-commit} in \texttt{decidedLeaders}.
\end{itemize}
\quad \textsc{CountDirectCommits} thus returns the count of direct commits in the subDAG, allowing us to obtain a ratio of directly committed leaders against the expected number of directly committed leaders (\textsc{decidedLeaders.len}). This approach provides a deterministic decision regarding the network state, as all leaders are processed in order, and the first encountered asynchronous leader will trigger a possible schedule update.
\begin{algorithm}[H]
\caption{Decider Functions}
\label{alg:decider}
\begin{algorithmic}[1]
\scriptsize
\State $\texttt{waveLength}$ \Comment{Set to 3}
\State $\texttt{waveLengthAsync}$ \Comment{Set to 4 or 5}
\State $\texttt{waveOffset}$ \Comment{for pipeline}
\State $\texttt{leaderOffset}$ \Comment{base\_committer}

\State
\Procedure{WaveNumber}{$r$}
\State \textbf{return} $(r - \texttt{waveOffset})/\texttt{waveLength}$
\EndProcedure

\State
\Procedure{ProposeRound}{$w$}
\State \textbf{return} $w \cdot \texttt{waveLength} + \texttt{waveOffset}$
\EndProcedure

\State
\Procedure{CertifyRound}{$w$, async}
\If{async}
\State \textbf{return} $w \cdot \texttt{waveLengthAsync} + \texttt{waveLength} - 1 + \texttt{waveOffset}$ \Comment{Calculate certify round via \texttt{waveLengthAsync}}
\Else
\State \textbf{return} $w \cdot \texttt{waveLength} + \texttt{waveLength} - 1 + \texttt{waveOffset}$ \Comment{Calculate certify round via \texttt{waveLength}}
\EndIf
\EndProcedure

\State
\Procedure{VoteRound}{$w$}
\State \textbf{return} \textsc{CertifyRound}($w$) $- \;1$ 
\EndProcedure

\State
\Procedure{LeaderBlock}{$w$, async}
\State $r_{\text{propose}}, r_{\text{certify}} \gets$ \textsc{ProposeRound}($w$), \textsc{CertifyRound}($w$)
\If{async}
\State $c \gets$ \textsc{CombineCoinShares}($\{b_{\text{share}} :$ $b \in \text{DAG}[r_{\text{certify}}, *]\}$) \Comment{Common coin for async leader election}
\State $l \gets c +$ \texttt{leaderOffset} \Comment{Modulo committee size}
\Else
\State $l \gets$ \textsc{GetLeader}($r_{\text{propose}} +$ \texttt{leaderOffset}) \Comment{Deterministic}
\EndIf
\State \textbf{return} $\text{DAG}[r_{\text{propose}}, l]$ \Comment{May return more than one block in case of equivocations}
\EndProcedure

\State
\Procedure{SkippedLeader}{$w$, $b_{\text{leader}}$}
\State $r_{\text{vote}} \gets$ \textsc{VoteRound}($w$)
\State \textbf{return} $|\{\neg$\textsc{IsVote}($b$, $b_{\text{leader}}$) $: b \in \text{DAG}[r_{\text{vote}}, *]\}| \geq 2f + 1$
\EndProcedure

\State
\Procedure{SupportedLeader}{$w$, $b_{\text{leader}}$}
\State $r_{\text{certify}} \gets$ \textsc{CertifyRound}($w$)
\State \textbf{return} $|\{$\textsc{IsCert}($b$, $b_{\text{leader}}$) $:b \in \text{DAG}[r_{\text{certify}}, *]\}| \geq 2f + 1$
\EndProcedure

\State
\Procedure{TryDirectDecide}{$w$}
\For{$b_{\text{leader}} \in$ LeaderBlock($w$)} \Comment{Loop over equivocations}
\If{\textsc{SkippedLeader}($w$, $b_{\text{leader}}$)}
\State \textbf{return} skip($w$)
\EndIf
\If{\textsc{SupportedLeader}($w$, $b_{\text{leader}}$)}
\State \textbf{return} commit($b_{\text{leader}}$)
\EndIf
\EndFor
\State \textbf{return} $\bot$
\EndProcedure

\State
\Procedure{TryIndirectDecide}{$w$, $S$}
\State $s_{\text{anchor}} \gets$ $\min\{$$s \in S:$ $r_{\text{certify}} < s_{\text{round}} \land s \neq$ \textsc{skip}($w$)$\}$\If{$s_{\text{anchor}} =$ \textsc{commit}($b_{\text{anchor}}$)}
\If{$\exists b_{\text{leader}} \in$ \textsc{LeaderBlock}($w$) $:$ \textsc{IsCertifiedLink}($b_{\text{anchor}}$, $b_{\text{leader}}$)}
\State \textbf{return} \textsc{commit}($b_{\text{leader}}$)
\Else
\State \textbf{return} \textsc{skip}($w$)
\EndIf
\EndIf
\State \textbf{return} $\bot$ \Comment{The anchor is undecided or not found}
\EndProcedure

\end{algorithmic}
\end{algorithm}
\begin{algorithm}[H]
\caption{Dual Mode Scheduler Algorithm}
\label{alg:dual_mode}
\scriptsize
\begin{algorithmic}[1]
\State $\texttt{decidedLeaders} \gets \{\}$ \Comment{Stored sequence of decided leaders}
\State $\texttt{lastAsyncRound} \gets 0$ \Comment{Last asynchronous round, begins at 0}
\State $\texttt{asyncInterval}$ \Comment{Current asynchronous interval frequency}
\State $\texttt{targetRatio}$\Comment{Target ratio between direct and expected commits}
\State $\texttt{scalingFactor}$\Comment{Scaling Factor for changes to \texttt{asyncInterval}}

\State
  \Procedure{UpdateAsyncSchedule}{}
  \State $v \gets \textsc{pop}(\texttt{decidedLeaders})$
  \State $O \gets \text{\textsc{getSubdag}($v$)}$
  \State $\mathcal{R} \gets \frac{\textsc{CountDirectCommits}(O)}{\texttt{decideLeaders.len}}$
  \State $\texttt{asyncInterval} \gets \textsc{changeAsyncRate}(\mathcal{R})$ \Comment{Deterministic function that changes the interval}
  \State $\texttt{lastAsyncRound} \gets v_\text{round}$ 
  \State $\texttt{decidedLeaders} \gets \textsc{\{\}}$
  \EndProcedure
  
\State
\Procedure{GetSubDag}{$v$}
\State $r_{\text{min}} \gets \min\{\ell.\text{round} : \ell \in \texttt{decidedLeaders}\}$
\State $O \gets []$
\For{$b_{\text{vertex}} \in v$}
\State $B \gets \{b \in \bigcup_{r \geq r_{\text{min}}} DAG[r, *] : \textsc{IsLink}(b, b_{\text{vertex}}) \land b \notin O\}$ \Comment{Get leaders from $r_\text{min}$ to $r$}
\For{$b \in B$}
\State $O \gets O \parallel b$
\EndFor
\EndFor
\State \textbf{return} $O$
\EndProcedure

\State
\Procedure{ChangeAsyncRate}{$\mathcal{R}$}
 \Comment{Compute direct commit ratio}
\If{$\mathcal{R} \geq \texttt{targetRatio}$}
\State $\texttt{asyncInterval} \gets \texttt{asyncInterval} / \texttt{scalingFactor}$  \Comment{Decrease asynchronous interval frequency}
\Else
\State $\texttt{asyncInterval} \gets \texttt{asyncInterval} \cdot \texttt{scalingFactor}$  \Comment{Increase asynchronous interval frequency}
\EndIf
\EndProcedure

\State
\Procedure{CountDirectCommits}{$O$}
\State count $\gets 0$
\For{$\ell \in \{\texttt{decidedLeaders} \mid \texttt{decidedLeaders}.\texttt{decision} \text{ == commit}\}$}
\State $V_\ell \gets \{b \in \text{$O$}[r_\ell + 1] : \ell \in b\}$ \Comment{Get blocks that voted for $\ell$}
\State $C_\ell \gets \{c \in \text{$O$}[r_\ell + 2] : |V_\ell \cap c| \geq 2f+1\}$ \Comment{Count blocks that contain $2f+1$ blocks in $V_{\ell}$}
\If{$|C_\ell| \geq 2f+1$}
\State count $\gets$ count $+ 1$
\EndIf
\EndFor
\State \textbf{return} count
\EndProcedure
 \State
  \Procedure{addSequence}{$\text{leaders}$}
  \For{$\ell \in \text{leaders}$}
  \State $\texttt{decidedLeaders} \gets \texttt{decidedLeaders} \parallel \ell$
  \EndFor
  \EndProcedure
 \State
\Procedure{isAsync}{$r$}
\State \textbf{return} $r \; \% \; \texttt{lastAsyncRound}+\texttt{asyncInterval}==0$
 \EndProcedure
  
\end{algorithmic}
\end{algorithm}
\newpage
\quad Algorithm~\ref{alg:DAG_helper} represents some helper functions that are utilised by Algorithm~\ref{alg:main} and Algorithm~\ref{alg:decider} to utilise the DAG. These have been mainly taken from \textsf{Mahi-Mahi}\xspace and \textsf{Mysticeti}\xspace, as the operations performed on the DAG itself remain the same \cite{MAHIMAHI,MYSTICETI}.
\begin{algorithm}[H]
\caption{DAG Helper Functions}
\label{alg:DAG_helper}
\footnotesize
\begin{algorithmic}[1]
\Procedure{IsVote}{$b_{vote}, b_{leader}$}
\State \textbf{function} \textsc{VotedBlock}($b, id, r$)
\If{$r \geq b.round$} \textbf{return} $\perp$
\EndIf
\For{$b' \in b_\text{parents}$}
\If{$(b'_\text{author}, b'_\text{round}) = (id, r)$} \textbf{return} $b'$
\EndIf
\State $res \gets$ \textsc{VotedBlock}$(b', id, r)$
\If{$res \neq \perp$} \textbf{return} $res$
\EndIf
\EndFor
\State \textbf{return} $\perp$
\State $(id, r) \gets (b_{leader}.author, b_{leader}.round)$
\State \textbf{return} VotedBlock$(b_{vote}, id, r) = b_{leader}$
\EndProcedure
\State
\Procedure{IsCert}{$b_{cert}, b_{leader}$}
\State $res \gets |\{b \in b_{cert}.parents : \text{IsVote}(b, b_{leader})\}|$
\State \textbf{return} $res \geq 2f + 1$
\EndProcedure
\State
\Procedure{IsLink}{$b_{old}, b_{new}$}
\State \textbf{return} $\exists$ path $b_1, \ldots, b_k:$ $b_1 = b_{old}, b_k = b_{new}, b_{j-1} \in b_j.parents$
\EndProcedure
\State
\Procedure{IsCertifiedLink}{$b_{anchor}, b_{leader}$}
\State $w \gets$ WaveNumber$(b_{leader}.round)$
\State $B \gets$ GetDecisionBlocks$(w)$
\State \textbf{return} $\exists b \in B :$ IsCert$(b, b_{leader}) \land$ IsLink$(b, b_{anchor})$
\EndProcedure
\State
\end{algorithmic}
\end{algorithm}
\newpage
\chapter{Research Assistance Methods}
\quad I acknowledge the use of Anthropic's Claude (version 3.5 and 4 of Claude Sonnet, Anthropic, \url{https://www.anthropic.com/}) for guidance through the \textsf{Mysticeti}\xspace's codebase and assistance in learning the \texttt{Rust} programming language \cite{MYSTICETI, RUST}. Furthermore, I also acknowledge Claude for technical guidance in \LaTeX{} formatting and figure creation.
\end{document}